\documentclass{statsoc}

\usepackage{preamble}

\newtheorem{theorem}{Theorem}[section]
\newtheorem{lemma}[theorem]{Lemma}

\newtheorem{proposition}[theorem]{Proposition}

\newdefinition{definition}[theorem]{Definition}
\crefname{definition}{Definition}{Definitions}
\newdefinition{assumption}[theorem]{Assumption}
\crefname{assumption}{Assumption}{Assumptions}

\title{Frequentist Inference for Semi-mechanistic Epidemic Models with Interventions}

\author[Bong {\it et al.}]{Heejong Bong}
\address{Department of Statistics, University of Michigan, Ann Arbor, MI, USA.}
\email{hbong@umich.edu}

\author[Bong {\it et al.}]{Val\'erie Ventura}
\address{Department of Statistics \& Data Science and Delphi Research Group, Carnegie Mellon University, Pittsburgh, PA, USA.}
\email{vventura@stat.cmu.edu}

\author[Bong {\it et al.}]{Larry Wasserman}
\address{Department of Statistics \& Data Science, Machine Learning Department and Delphi Research Group, Carnegie Mellon University, Pittsburgh, PA, USA.}
\email{larry@stat.cmu.edu}

\begin{document}
\maketitle

\begin{abstract}
The effect of public health interventions
on an epidemic
are often estimated by
adding the intervention to
epidemic models.
During the Covid-19 epidemic, numerous papers 
used such methods for making scenario predictions.
The majority of these papers use Bayesian methods
to estimate the parameters of the model.
In this paper we show how to use
frequentist methods for estimating these effects
which avoids having to specify prior distributions.
We also use model-free shrinkage methods
to improve estimation when there are many different geographic regions.
This allows us to borrow strength from different regions
while still getting confidence intervals
with correct coverage and
without having to specify a hierarchical model.
Throughout, we focus on a semi-mechanistic model
which provides a simple, tractable alternative to compartmental methods.
\end{abstract}

\section{Introduction}

We consider the problem
of estimating epidemic models
that include public health interventions.
We use the term
``intervention'' very broadly
to refer to any variable whose effect on the epidemic
is of interest such as:
mask usage, changes in social mobility, lockdowns, vaccines, etc.
We focus on the semi-mechanistic model
from \citet{bhatt2020semi}
which we describe in \cref{sec::themodel}.
The model
provides a simple, tractable alternative to compartmental models
(such as the SIR model of \citealp{sir})
but, as shown in \citet{bhatt2020semi},
it still captures enough dynamics to
provide accurate modeling of the epidemic process.

During the Covid-19 epidemic, numerous papers 
used similar models for making scenario predictions;
see, for example, \cite{perra2021non, g2022, L2022s, vytla2021, baker2020}.
Virtually all of these papers use Bayesian methods
to estimate the parameters of the model.
Typically, the posterior is approximated by
Monte Carlo methods.
Bayesian methods provide a powerful approach for combining
prior information with data.
But these methods do have some drawbacks  \citep{Bong2023}.
First, one has to specify many priors
and the inference can be sensitive to the choice of priors.
Second, Bayesian inferences do not come with frequency guarantees:
a 95 percent posterior interval does not contain the true parameter value
95 percent of the time.
Frequentist methods require no priors
and have proper frequency guarantees.

Our comparisons of frequentist and
Bayesian methods in this paper
is not meant as a criticism of the Bayesian approach ---
which can be a powerful way to incoporate prior information ---
but rather to \V{provide methods to perform frequentist inferences and } highlight the differences in the two approaches
when modeling epidemics.

In some cases, we have data from several different regions.
The accuracy of the estimates can be improved by properly combining
the information from the different regions.
In the Bayesian framework, information is combined by positing
a hierarchical model. This can introduce bias if this
model is mis-specified. We show how information can be combined in a frequentist manner 
in a way that does not require specifying a hierarchical model.

\medskip

{\bf Related Work.} 
The literature on inference
for epidemic models is huge
and we cannot provide a complete review here.
In addition to the references mentioned earlier
we also point to
\cite{chatzilena2019, fintzi2022linear, li2021efficient}
as recent examples of some of the approaches.
Again we emphasize that almost all work
uses Bayesian inference in contrast to our approach.
The work in \citet{bhatt2020semi}
is the closest to ours and is, in fact,
the main inspiration for our work.

\medskip

{\bf Paper Outline.}
In \cref{sec::themodel} we describe the model and in \cref{sec::inference} we 
describe how to fit the model to data 
by maximum likelihood (ML).
It is important to check the model fit so that we can trust inferences about model parameters
and derived forecasts, which we address in \cref{sec::diagnostics}.
In \cref{sec::predictions} we develop inference for ML estimates and scenario predictions, which are robust to model mis-specifications. 
In \cref{sec::shrinkage} we show how to use shrinkage methods
to improve estimation when there are several different geographic regions,
without having to specify a hierarchical model.
In \cref{sec::examples}
we apply our methods to simulated and observational data. We conclude in \cref{sec::discussion}.

\section{The Semi-mechanistic Model for Infections, Outcomes and Treatments} 
\label{sec::themodel}

Let
$Y_1,Y_2,\ldots, Y_T$ denote \V{a scalar time series of observed} outcomes (such as deaths) \V{at a location} and
$I_{-\infty},\ldots, I_0,I_1,\ldots, I_T$ 
denote \V{the associated process of} new infections where
$t$ typically indexes days or weeks;
 we start their indexing at $-\infty$ because $Y_t$ depends on infections in the past.
The $I_t$'s are \V{typically} latent (unobserved) \V{because they are difficult to measure}.
We use overlines to denote histories, for example,
$\overline I_t = (I_{-\infty},\ldots,I_t)$ and
$\overline Y_t = (Y_1,\ldots,Y_t)$.

We focus on 
the model from \citet{bhatt2020semi} which is
\begin{equation} \label{eq::model0}
\begin{aligned}
    \E[Y_t| \overline I_t, \overline Y_{t-1}] &= \alpha_t \sum_{s:s < t}  I_s \pi_{t-s}\\
    \E[I_t| \overline I_{t-1}, \overline Y_{t-1}] &= R_t \sum_{s:s < t} I_s g_{t-s}
\end{aligned}
\end{equation}
for $t>0$, where $g_r \ge 0$, 
$\sum_{t=1}^\infty g_t = 1$, is the generating distribution that specifies the probabilities of the 
possible lags between an infection and a secondary infection, 
and the lag between infection and outcome has distribution
$\pi_r \ge 0$, with $\sum_{t=1}^\infty \pi_t = 1$. 
The ascertainment rate $\alpha_t$ is the probability that an infection at $t$ leads to an outcome
and the reproduction number $R_t > 0$ is the mean number of infections generated in the future by an infection at $t$. \V{The reproduction number is typically the parameter of main interest because it quantifies the progression of the epidemic. Finally, as is standard in time series analyses, the expectations are with respect to the conditional distributions of $Y_t$ and $I_t$ given the past. These distributions are often assumed to be Poisson or Negative Binomial (NB), the latter being more flexible and having the former as a limiting case.
}

This model, in its general form,
where $g$ and $\pi$ are unknown, and $\alpha_t$ and $R_t$ can vary with $t$, is not identified,
as there are more parameters than data points.
Hence there is no consistent estimator of the parameters.
Typically, $g$ and $\pi$ are estimated separately, which we assume as well.
Hence, we take them as fixed in the rest of the paper; they are plotted in \cref{fig::gpi}.
Similarly, $\alpha_t$ is also typically estimated from other data sources.
Finally, to be able to estimate the reproduction number $R_t$ consistently, one needs to assume a model with a finite number of degrees of freedom,
for example $R_t = \sum_{j=1}^k \gamma_j b_j(t)$
for basis functions $b_1,\ldots, b_k$.

Now let $A_{-\infty},\ldots, A_0,A_1,\ldots,A_T$ represent a \V{scalar or vector} time series of some interventions
such as mobility, masks, vaccines, lockdowns etc. 
To model the effect of this intervention we will modify the model
to include the $A_t$'s.
(There are other ways to model the intervention 
such as marginal structural models; see \citealp{Bonvini2022covid}).
We consider
the model 
in \citet{bhatt2020semi}:
\begin{equation}\label{eq::model1}
\E[I_t| \overline I_{t-1}, \overline Y_{t-1}, \overline A_t] = R(\overline A_t ,\beta) \sum_{s:s<t}  I_s g_{t-s},
\end{equation}
where $\overline A_t = (A_{-\infty},\ldots,A_t)$ and 
$R(\overline A_t,\beta)$
is some parametric model with the property that
it
does not depend on $A$ when $\beta$
(or some component of $\beta$) is 0. 
\V{This model assumes that the reproduction number changes only through the interventions, which is reasonable if the period of study is short enough that conditions other than the interventions -- e.g. the prevalence of population immunity, number of susceptibles, etc. -- that can affect disease transmission remain unchanged.} 
The parametrization used in \citet{bhatt2020semi} is
\begin{equation} 
\label{eq::Rt.model0}
R_t \equiv R(\overline A_t,\beta) = K \left(1 + e^{ -(\beta_0 + \beta_1^\top A_t)}\right)^{-1},
\end{equation}
where $\beta=(\beta_0, \beta_1)$, $K$ is the maximum possible transmission rate and 
$A_t$ is a vector 
representing \V{one or} several interventions. \V{That is, $R_t/K$ is modelled as the logistic function 
of $(\beta_0 + \beta_1^\top A_t)$.}
\V{Note that $\beta_1$ measures
the change in the reproduction number $R_t$ 
with respect to changes in $A_t$, but it is hard to interpret it more meaningfully because 
the logistic function is non-linear in $A_t$. 
We can, however, look at $R_t$ itself, which is more interpretable.
We discuss this point more in Section \ref{sec::data}.
(In the application of 
\cref{sec::data}, $\beta_1^\top A_t$ is not sufficiently small to linearize the logistic function.)}
An alternative to the reproduction number in \cref{eq::Rt.model0} would be to include the intervention in the sum, leading to the exponential model used in \cite{Bonvini2022covid}
\begin{equation*} 
    \E[I_t| \overline I_{t-1}, \overline Y_{t-1}, \overline A_t] = \sum_{s:s<t} e^{\beta_0 + \beta_1^\top A_{s}}I_s g_{t-s}.
\end{equation*}
In either case, 
the assumption is that
the intervention $A$ affects the outcome $Y$ by way of infections $I$, \V{which is reasonable since $Y$, e.g. deaths or cases, cannot happen without infections.}
See \cref{fig::dag1}.
In the rest of the paper, we focus on the model in \cref{eq::Rt.model0} \V{so that our results can be compared to \citet{bhatt2020semi}.}

Our two primary goals are to obtain point estimates and confidence intervals for $\beta$ and $R_t$, so we can quantify the effects of interventions on the outcome, and \V{use these estimates to produce} scenario predictions.

{\bf Remark:}
{\em \V{
The parameter determines
the local effect of $A$ on $R_t$.
But if we want the average treatment effect
of the entire treatment vector $a_1,\ldots, a_T$
on the final value $Y_T$
then one needs to use the $g$-formula
\cite{robins2013estimation, bates2022causal}. 
}}

\begin{figure}
\begin{center}
\begin{tikzpicture}[->, shorten >=2pt,>=stealth, node distance=1cm, noname/.style={ ellipse, minimum width=5em, minimum height=3em, draw } ]
\node[] (1) {$A_1$};
\node[circle,draw,fill=lightgray] (2) [right= of 1] {$I_1$};
\node (3) [right= of 2] {$Y_1$};
\node (4) [right= of 3] {$A_2$};
\node[circle,draw,fill=lightgray] (5) [right= of 4] {$I_2$};
\node (6) [right= of 5] {$Y_2$};

\path (1) edge node {} (2);
\path (1) edge [bend left=60pt] node {} (4);
\path (1) edge [bend left=60pt] node {} (5);

\path (2) edge node {} (3);
\path (2) edge [bend right=60pt] node {} (4);
\path (2) edge [bend right=60pt] node {} (5);
\path (2) edge [bend right=60pt] node {} (6);

\path (3) edge node {} (4);

\path (4) edge node {} (5);

\path (5) edge node {} (6);

\end{tikzpicture} 
\end{center}
\caption{\textbf{ DAG for the model}. Note that $A$ affects $Y$ only through $I$.}
\label{fig::dag1}
\end{figure}
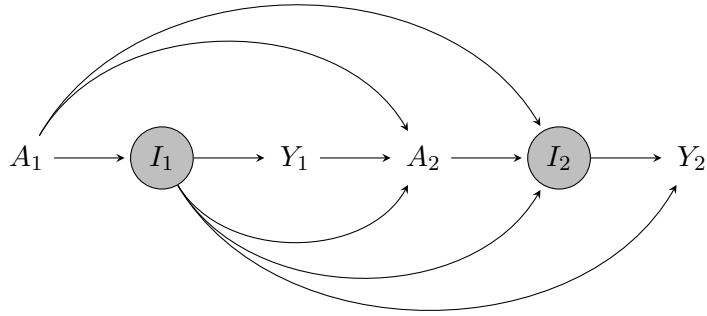

\section{Estimation \label{sec::inference}}

Here we discuss maximum likelihood inference for the model parameters.
We focus on one single observed time series but
in \cref{sec::shrinkage}
we consider shrinkage methods for dealing with time series at multiple locations. So far we have only specified the means of the infection and death processes in \cref{eq::model0}. To complete the model,
we assume that the $Y_t$ are either Gaussian or Negative Binomial (NB) distributed, 
but our methods can
easily be extended to other distributions.
Just as in \citet{bhatt2020semi}, we only model the distribution of the outcome and let the latent infection process be deterministic but unobserved, \V{which implies that the $Y_t$ are independent conditional on its past $\overline Y_{t-1}$. }
This is an unrealistic assumption --
infections are random --
but because we observe only one repeat of the outcome process, we
cannot identify the distributions of both $I_t$ and $Y_t$.
Therefore, instead of fitting \cref{eq::model0}, we fit
\begin{equation} \label{eq::model00}
\begin{aligned} 
\E[Y_t| \overline I_t, \overline Y_{t-1}, \overline{A}_t] \V{ = \E[Y_t| \overline{A}_t] } &= \alpha_t \sum_{s: s < t}  I_s \pi_{t-s}\\
I_t &= R(\overline A_t, \beta) \sum_{s: s < t} I_s g_{t-s}
\end{aligned}
\end{equation}
where $Y_t$ is either Gaussian or NB.
Both distributions have an additional nuisance parameter $\nu_t$, the standard deviation parameter for the Gaussian distribution and the inverse dispersion parameter for the NB distribution.
We assume $\nu_t = \nu$ for all $t$. 
\V{The parameter to be estimated is $\theta=(\beta, \nu, \mu)$, where $\beta$ is the reproduction number parameter, $\nu$ is the nuisance parameter and $\mu$ is the seeding value, described below. We obtain the maximum likelihood estimator (MLE) of $\theta$ 
using the Block Coordinate Descent Algorithm described in
\cref{sec::append.algos}.}

\paragraph{Infection process seeding} 
\V{The expectation of $Y_t$ in \cref{eq::model00} depends on all infections $I_t$ prior to time $t$. 
Because $Y_t$ is observed for $t \ge 1$, we cannot infer $I_t$ prior to $t=1$.
%
\cite{bhatt2020semi} thus assumed
$I_{t}=0$ for $t \leq -T_0$ and $I_{t}=e^\mu$ for $t = -T_0+1, \dots, 0$, where $\mu$ is a parameter to be estimated and $T_0=6$.
We proceeded similarly but with $T_0 = 40$, because the infection-to-death distribution $\pi$ in \cref{fig::gpi} has substantial mass on $\{1, \dots, 40\}$, and we 
obtained better fits than with $T_0=6$, especially for small $t$.
}

\V{
Note that the chosen ``seeding" does not conform 
with the data or the model -- infections cannot remain constant over time -- which induces some bias.
We initially tried seeding only one value at time $-T_0$, specifically 
$I_{-T_0} = e^\mu$, 
with $\mu$ and $T_0$ to be estimated, 
and $I_{t}=0$ for $t \leq -T_0$.
But because there is little information in the data about $I_t$ for $t \le 0$, this approach produced very large confidence intervals for the parameters, so we abandoned it.
The optimal seeding strategy involves a bias-variance tradeoff and remains an open problem.
}

\paragraph{Parameter identifiability}

We deal with the nonidentifiability that arises from having too many
parameters by imposing restrictions
such as taking $\alpha_t = \alpha$ to be constant \V{and assuming that $R_t$ 
varies with $t$ only through $A_t$. Both these assumptions are reasonable for the data analysis in \cref{sec::data} because the study period is short enough that $\alpha_t$ and $R_t$ would not naturally vary on their own. Otherwise one could include a time varying intercept with a small number of parameters in the model for $R_t$, as in \cite{Bonvini2022covid}.}
In the Bayesian framework, it is not uncommon to ignore the
identifiability issues since, formally, the posterior can be computed
even when there is nonidentifiability.
But this is dangerous because then the posterior is driven by the 
prior even with an infinite sample size.
The effect of the nonidentifiability is then hidden from the user.
\V{However, there are methods
to deal with identifiability issues
in Bayesian inference; see \cite{gustafson2015, giacomini2021}.}

\begin{proposition} \label{thm::identifiability}
Assume that the data follow model (\ref{eq::model00}) 
\V{and that $Y_t$ given $(\overline I_T, \overline Y_{t-1}, \overline{A}_t)$ is either Gaussian or Negative Binomial with dispersion parameter $\nu$.}
Also assume that
(a) if $\beta \neq 0$ then
$P(\beta_0 + \beta_1^\top A_t = 0,\ \text{for\ all\ }t=1,\ldots, T-T_0)=0$ and 
(b) $\sum_{t=1}^{T_0} \pi_t > 0$ and $\sum_{t=1}^{T_0} g_t > 0$. 
Then $\theta = (\beta, \nu, \mu)$ is identified.
\end{proposition}

\V{The proof is in \cref{sec::append.proofs}.
Condition (a) requires that $\beta_0 + \beta_1^\top A_t$ not always be 0, which is a very weak assumption. The generating and infection-to-death distributions $g$ and $\pi$ (\cref{fig::gpi}) satisfy
condition (b) for any small $T_0$ -- in fact the sums are close to 1 for $T_0=40$.
}

\section{Model Checking \label{sec::diagnostics}}

It is important to check the adequacy of the regression mean so that we can trust inferences about model parameters, and while we could set confidence intervals for parameters based on the asymptotic normal distribution of ML estimates, checking the distribution of the data is also important because epidemiological models are often used to forecast the future time course of epidemics (see \cref{sec::predictions}).

Regression methods provide a battery of diagnostics to assess model fit and identify outliers and influential observations.
Because our model assumes that the latent infection process is deterministic, it implies that the observed deaths are independent conditional on the past. Hence our model fits within the regression framework, although we cannot write analytically the regression mean. We can nevertheless apply standard diagnostics.

Most diagnostics are based on residuals.
For non-Gaussian data, deviance residuals are close to normally distributed, and therefore more appropriate for diagnostics, than raw residuals, although they are comparable when the outcome count data are large. For brevity and simplicity here, we therefore focus on raw residuals. They are defined as 
$e_t \equiv Y_t - \hat Y_t$, where 
$\hat Y_t = \Exp_{\hat\theta}[Y_t|\overline A_{t}]$. 
Standardized residuals are given by
$r_t = e_t / \sqrt{(1 - h_t) \hat \Var(Y_t)}$, where
$\hat \Var(Y_t) = \hat \nu^2 = n^{-1} \sum(Y_t - \hat Y_t)^2$ for Gaussian data, 
$\hat \Var(Y_t) = \hat{Y}_t (\hat Y_t^2 / \nu)$ for NB data, and $h_t$ is the leverage of observation $t$, obtained from the diagonal of the hat matrix,
that is, the projection matrix corresponding to
the linearization at the MLE.
 We have no explicit hat matrix here, but because the covariates, time and intervention(s), are equally spaced for the former and designed or taking values in a range for the later, all observations should have similar, and thus very small, leverage, so that $(1-h_t) \approx 1$ for all $t$.
Studentized residuals are defined like standardized residuals, except that to obtain the residual for $Y_t$, all parameters are estimated without $Y_t$.
 
Model diagnostics aim to assess the adequacy of the regression mean function and the assumed distribution of the data. For the former we plot the standardized residuals versus all covariates and fitted values -- they should look random and homoskedastic -- and for the later, we produce a Q-Q plot of the standardized residuals versus the standard normal quantiles.

Case diagnostics aim to identify influential points, that is, points that have large influence on parameter estimates, measured by Cook's distance:
$$ 
D_t = \sum_{t=1}^{T} (\hat Y_t - \hat Y_{(-t)})^2 / (p s^2),
$$
where $s^2=  \sum(Y_t - \hat Y_t)^2/(n-p)$,
$p$ is the number of parameters, and the subscript $(-t)$ indicates that estimates were obtained without $Y_t$. A value of $D_t$ larger than one usually signals influence.
Large values of 
$C_t = (\hat \beta_1 - \hat \beta_{1,(-t)}) / \sqrt{\Var(\hat \beta_1)}$ more specifically detect points that influence $\hat \beta_1$, and similarly for other $\beta_j$'s if there are several interventions.
Influential points are either high leverage points -- but there shouldn't be any since $h_t \approx 0$ for all $t$ -- or outliers, so it is also useful to identify outliers using standardized residuals.

\section{Model Free Standard Errors, CLT and Confidence Set}
\label{sec::CLT}


We now turn to the issue of obtaining standard errors and confidence sets for parameter $\theta=(\beta, \nu, \mu)$
defined in \cref{sec::inference}.
Note that under model misspecification, the stochastic process $\overline Y_{T}$ may exhibit temporal dependence conditional on $\overline I_{T}$. 
We want an estimate of the standard error that is robust to
extra dependence of this sort.
In order to establish consistency and asymptotic normality of the MLE under this setting, we require that the process has a sufficiently fading memory as discussed by \citet{potscher1997dynamic}. As an illustration, we will assume that $\overline Y_{T}$ is well-approximated by an $\alpha$-mixing process, which is a class of random processes that have the suitable ``fading memory" property; see Chapter 6 in \citet{potscher1997dynamic}, and below for the definitions of (a) $\alpha$-mixing processes and (b) the concept of approximation.

\begin{definition} \label{def::alpha}
Let $\overline Y_{T}$ be a stochastic process on a 
probability space $(\reals, \mathfrak{F}, \Pr)$.
    \begin{enumerate}
        \item Define
        \begin{equation*}
            \alpha(j) = \sup_{k \in \nats} \sup \{ \abs{ \Pr[F \cap G] - \Pr[F] \Pr[G]}: F \in \mathfrak{F}_{[1,k]}, G \in \mathfrak{F}_{[k+j, T]} \},
        \end{equation*}
        where $\mathfrak{F}_{[t,s]}$ is the $\sigma$-field generated by $Y_t, \dots, Y_s$ for $1 \leq t < s \leq T$.
        If $\lim_{j \rightarrow \infty} \alpha(j) = 0$, we say that the process $(Y_t)$ is {\it $\alpha$-mixing}.
        \medskip
        
        \item \V{Let $(Y^o_t \in \reals^p: t = -\infty, \dots, \infty)$ be a random process we take as a baseline. We say $\overline Y_{T}$ is {\it near epoch dependent of size $-q$ on the baseline process $Y^o_t$} if
        \begin{equation*}
            \sup_t \Exp[\abs{Y_t - \Exp[Y_t | Y^o_{t-m}, \dots, Y^o_{t+m}]}^2] = O(m^{-2q}).
        \end{equation*}}
    \end{enumerate}
\end{definition}

\V{The rest of this section relies heavily on  \citet{potscher1997dynamic} Theorems 7.1 and 11.2, and associated assumptions. The theorems entail consistency and asymptotic normality of our estimator conditional on $\overline{A}_T$, so our assumptions below are presented under the $\sigma$-algebra $\mathfrak{F}(\overline{Y}_T)$ generated by $\overline{A}_T$. First, we make the following two assumptions on $\overline Y_{T}$.} 

\begin{assumption} \label{assmp::mixing}
\V{Conditional on $\overline A_T$,}
    $\overline Y_{T}$ is near epoch dependent of size $-2(\gamma-1)/(\gamma-2)$ on an $\alpha$-mixing baseline process with mixing coefficients $\alpha(j) = O(j^{-2\gamma/(\gamma-2)})$ for some $\gamma > 2$,
    \V{almost surely.}
\end{assumption}

\begin{assumption} \label{assmp::bounded}
    For $\gamma > 2$ in Assumption~\ref{assmp::mixing}, 
    there exists $\mathfrak{F}(\overline{A}_T)$-measurable random variable $Y_{4\gamma,\max} \in (0, \infty)$ such that 
    \begin{equation*}
        \sup_t \Exp[Y_t^{4\gamma} \mid \bar{A}_t] \leq Y_{4\gamma,\max}, ~\text{a.s.}
    \end{equation*}
    \V{and the conditional distributions of $Y_t$ on $\overline{A}_t$ satisfy tightness property, i.e.,
    \begin{equation*}
        \lim_{m \rightarrow \infty} \sup_{T} T^{-1} \sum_{t=1}^T \Pr[Y_t \notin K_m | \overline{A}_t] = 0, ~\text{a.s.,}
    \end{equation*}
    for some sequence of $\mathfrak{F}(\overline{A}_T)$-measurable random compact sets $K_m \subset \reals$.}
\end{assumption}
\V{Assumption \ref{assmp::mixing} requires that
the correlations between the residuals die off
at a polynomial rate as the observations get further apart in time.
Some assumption of this form is used in most time
series analysis. This assumption appears to be met in the data analysis in \cref{sec::data}; see Appendix \cref{fig::ACF}.} 
\V{Assumption \ref{assmp::bounded} further requires that certain moments are bounded, which is automatically satisfied if the random variables are bounded, which they are in \cref{sec::data}.}

Next, we need an identifiability assumption. In \cref{thm::identifiability}, 
we made a weak one in assumption (a), but we need a stronger one for the consistency and asymptotic normality of the MLE.

\begin{assumption} \label{assmp::eigenvalue}
    There exists $t \in \{1, \dots, T_0\}$ such that $\pi_t > 0$, and $\pi_t = g_t = 0$ for any $t \geq \tau$ for some $\tau < T$. Furthermore, 
    there exists $\mathfrak{F}(\overline{A}_T)$-measurable random variables $\lambda_{\min,A}$ and $A_{\max}$ such that
    $0 < \lambda_{\min,A} \cdot \norm{\beta}_2^2 \leq \frac{1}{T-t} \norm{\beta^\top (A_{t+1}, \dots, A_{T})}_2^2$ and $\norm{A_t}_2 \leq A_{\max}$ for any $0 < t < T$ and $\beta$, almost surely.
\end{assumption}
\V{This assumption holds with
probability one if $A$ has a continuous distribution. This is like the usual
assumption that the design matrix in regression is non-singular.}

Next, we define a parameter space $\Theta$ as follows.

\begin{definition} \label{def::Theta}
For given constants $I_{\min}$, $I_{\max}$, $D_{\max}^{(1)}$, $D_{\max}^{(2)} > 0$, $r_{\min}$ and $r_{\max}$, let $\Theta$ be a $\mathfrak{F}(A)$-measurable random set of $\theta \equiv (\beta, \mu, \nu)$ 
such that
\begin{enumerate}
    \item $0 < I_{\min} \leq \Exp_\theta[I_t|\overline A_t] \leq I_{\max} < \infty$ for all $t \geq 1$,
    \medskip
    \item $\norm{\nabla_{(\beta,\mu)} \Exp_\theta[I_t|\overline A_t]}_\infty \leq D_{\max}^{(1)}$,
    $\norm{\nabla_{(\beta,\mu)}^2 \Exp_\theta[I_t|\overline A_t]}_\infty \leq D_{\max}^{(2)}$
    for any $t \geq 1$, and
    \medskip
    \item $0 < r_{\min} \leq \frac{\sum_{t=1}^T \Var_\theta[Y_t |\overline A_t]}{\sum_{t=1}^T \Exp_\theta[Y_t |\overline A_t]} \leq r_{\max}$ for all $T \geq 1$,
\end{enumerate}
almost surely.
\end{definition}
%
\V{Parameters in $\Theta$ satisfy that $I_t$ and its derivatives have
bounded moments given $\overline{A}_t$. It automatically holds if $I_t$ is bounded, which is the case 
in the data analysis in \cref{sec::data}.
}

\V{Let $\ell(\theta) \equiv \sum_{t=1}^T \ell_t(\theta)$ denotes the log-likelihood function, where $\ell_t(\theta) \equiv \log p_{Y_t}(Y_t | \overline{I}_{t}, \overline{Y}_{t-1}, \overline{A}_t, \theta)$ ($=\log p_{Y_t}(Y_t | \overline{A}_t, \theta)$ since the $I_t$ are assumed to
be deterministic). Let $\hat \theta \equiv {\arg\max}_{\theta \in \Theta} \ell(\theta)$.} 
Also define $\ell^*(\theta) \equiv \Exp[ \ell(\theta)|\overline A_{T}]$. Because our model in \cref{eq::model00} is nonlinear, the convexity of $\ell^*(\theta)$ and hence the uniqueness of the global maximum are not guaranteed. We assume that the global maximum is identifiably unique. The assumption is common in the consistency proof of MLEs \citep{potscher1997dynamic,van2000asymptotic}.

\begin{assumption} \label{assmp::identified_minimizer}
    $\ell^*$ has an identifiably unique maximizer $\theta^*$ in $\Theta$, almost surely.l 
    That is, $\theta^*$ is in the interior of $\Theta$, and for every $\epsilon > 0$,
    \begin{equation*}
        \liminf_{n \rightarrow \infty} \left[ \inf_{\theta \in \Theta: \norm{\theta - \theta^*}_2 \geq \epsilon} \ell^*(\theta) - \ell^*(\theta^*) \right] > 0.
    \end{equation*}
    We further assume that
    \begin{equation*}
    \begin{aligned}
        & \liminf_{T \rightarrow \infty} \lambda_{\min}\left(- \frac{1}{T} \nabla_\theta^2 \ell^*(\theta^*) \right) > 0, \\
       & \liminf_{T \rightarrow \infty} \lambda_{\min}\left( \frac{1}{T} \cdot \Var\left[\nabla_\theta \ell(\theta^*) | \overline{A}_T \right] \right) > 0,
    \end{aligned}
    \end{equation*}
    almost surely.
\end{assumption}
\V{Note that $\theta^*$ is defined conditional on $\overline{A}_T$, so $\theta^*$ is a $\mathfrak{F}(\overline{A}_T)$-measurable random variable. If model correctly specifes the data generation, $\theta^*$ is essentially constant at the \emph{true} value. 
}

\V{We now have all the elements to proceed with the normality and consistency of the MLE of $\theta = (\beta, \mu, \nu)$.}

\begin{theorem} \label{thm::asymp_norm}
    Under Assumptions~\ref{assmp::mixing}, \ref{assmp::bounded}, \ref{assmp::eigenvalue} and \ref{assmp::identified_minimizer}, 
    \begin{equation*}
        T^{1/2} \Upsilon^{-1/2} (\hat\theta - \theta^*) \overset{d}{\rightarrow} \distNorm(0, id),
    \end{equation*}
    where $\Upsilon \equiv (\E[H(\theta^*)|\overline{Y}_T])^{-1} \E[s(\theta^*) s(\theta^*)^\top|\overline{Y}_T] (\E[H(\theta^*)|\overline{Y}_T])^{-1}$,
    $s(\theta) \equiv \nabla_\theta \ell(\theta)$ 
    is the score function and
    $H(\theta)$ is the Hessian matrix of $\ell(\theta)$.
\end{theorem}

The proof of \cref{thm::asymp_norm} is 
in~\cref{app::CLT}.
To construct confidence intervals,
we need to estimate the variance matrix 
$\Upsilon$.
A plug-in estimator is 
$$
\hat\Upsilon = 
[H(\hat\theta)]^{-1} s(\hat\theta) s(\hat\theta)^\top [H(\hat\theta)]^{-1}.
$$
\V{However, this estimator does not have full rank.}
If $\nabla_\theta \ell_t(\hat \theta)$ for
$t=1, \ldots, T$ are mutually independent, 
then
$\sum_{t=1}^T \nabla_\theta \ell_t(\hat\theta) \nabla_\theta \ell_t(\hat\theta)^\top$ is a consistent estimator of $\Exp[s(\theta^*) s(\theta^*)^\top|\overline{Y}_T]$. 
But because there may be
dependence due to
model mis-specification,
we instead use the \V{heteroskedasticity and autocorrelation consistent} (HAC) estimator \citep{newey1987simple}:
\begin{equation*}
\hat{S} \equiv 
\sum_{t, s \in (0,T]} w(\abs{t-s}, T) \cdot 
\nabla_\theta 
\ell_t(\hat\theta) \nabla_\theta \ell_s(\hat\theta)^\top,
\end{equation*}
where $w$ is a weight function satisfying 
$\lim_{T \rightarrow \infty} w(t, T) = 1$ for any fixed $t$. 
In particular, we set $w(t, T) = \max\{ 1 - \abs{t}/\tau, 0 \}$ 
where $\tau = \floor{4(T/100)^{2/9}}$. 
(This is a commonly used default value.)
The resulting estimator,
$$
\hat\Upsilon \equiv 
[H(\hat\theta)]^{-1} \hat{S} ~[H(\hat\theta)]^{-1},
$$ 
is usually
referred to as a {\it sandwich estimator}.

\begin{proposition} \label{thm::sandwich}
 Suppose that $\Exp[\nabla_\theta \ell_t(\theta^*)|\overline{A}_t] = 0$ at each $t$ with probability $1$. Under Assumptions~\ref{assmp::mixing},
 \ref{assmp::bounded}, \ref{assmp::eigenvalue} and \ref{assmp::identified_minimizer}, 
 $\hat\Upsilon$ is a consistent estimate of $\Upsilon$.
\end{proposition}
The proof of Proposition~\ref{thm::sandwich} is 
in~\cref{app::CLT}.
\V{For Normal and NB models, the additional assumption stated in Theorem~\ref{thm::sandwich} implies that $\Exp[Y_t|\overline{A}_t] = \Exp_{\theta^*}[Y_t|\overline{A}_t]$. That is, for our inference based on the HAC variance estimate to work, the mean of $Y_t$ must be correctly specified for each $t$; distributional misspecifications do not affect the validity
of the inference. This assumption is dubious in the data analysis of \cref{sec::data} because the residual diagnostics 
in \cref{fig::diagnostic_ny} show 
some lack of fit at the start and end of the epidemic.} 

\V{It follows that an asymptotic joint
$1-\alpha$ confidence interval 
for $\theta$ is 
\begin{equation*}
\begin{aligned}
    \mathcal{C} 
    & = \hat\theta + \chi_{d,\alpha} \cdot T^{-1/2} \hat\Upsilon^{1/2} \mathcal{B}_d
    \equiv \{ \hat\theta + \chi_{d,\alpha} \cdot T^{-1/2} \hat\Upsilon^{1/2} x: x \in \mathcal{B}_d \} \\
    & = \{\theta: (\theta-\hat\theta)^T \hat\Upsilon^{-1} 
    (\theta - \hat\theta) \leq T^{-1} \chi_{d,\alpha}^2\},
\end{aligned}
\end{equation*}
where $\chi_{d,\alpha}$ is the $\alpha$-quantile of the chi-square distribution with degree of freedom $d$,
$\hat\Upsilon^{1/2}$ is the matrix square-root of $\hat\Upsilon$, and
$\mathcal{B}_d \equiv \{x \in \reals^d: \norm{x}_2 \leq 1\}$ is the $d$-dimensional unit ball.
}
Under model misspecifcation,
the MLE is a consistent estimate of the distribution 
closest to the true distribution in Kullback-Leibler distance.
The sandwich estimator is still valid in this case.

\paragraph{Global Confidence Bands for Scenario Predictions}
\label{sec::predictions}

Having obtained a $(1-\alpha)$ \V{joint} confidence set $\mathcal{C}$ for $\theta$, we can make scenario forecasts as follows.
We fix $A_1,\ldots, A_T$.
Define
$$
l_t = \min_{\theta\in \mathcal{C}} \Exp_\theta[Y_t|\overline A_t]
$$
and
$$
u_t = \max_{\theta\in \mathcal{C}} \Exp_\theta[Y_t|\overline A_t].
$$
It follows that,
under the scenario $(A_1,\ldots, A_T)$,
$$
P(l_t \leq \Exp[Y_t] \leq u_t\ {\rm for\ all\ }t)\geq 1-\alpha.
$$

{\bf Remark.}
{\em
We can also construct prediction bands
for $Y_t$ rather than for the mean.
We find that these tend to be very large.
If one uses Bayesian prediction bands with
strong priors, it is possible to obtain narrow prediction bands
but these are very dependent on the prior and the
distributional assumptions on $Y_t$.}

\section{Shrinkage: Borrowing strength across locations \label{sec::shrinkage}}

Suppose now that we have data
from $N$ different geographic regions
with corresponding parameters
\V{$\theta_{[1]}, \ldots, \theta_{[N]}$}.
We would like to estimate each $\theta_{[j]}$
while taking advantage of any similarity between the parameters.
This is usually done by shrinking the individual estimates towards each other
which reduces the mean squared error
of the estimators (\citealp{james1960estimation}).
One approach to shrinkage is to use a hierarchical model
where we treat the
$\theta_{[j]}$'s as draws from a distribution $g(\theta;\xi)$
depending on hyper-parameters $\xi$.
For example,
the R package \texttt{epidemia}
\citep{epidemia} uses a Normal $(0,\Sigma)$ distribution.
However, this requires specifying $g(\theta;\xi)$,
which is difficult because the $\theta_{[j]}$'s are not observed.
Instead, we use the approach from \cite{armstrong2022robust} which provides 
a shrinkage method that does not require
specifying this distribution.
The method provides model-free confidence intervals for univariate parameters. 
In particular, this method corrects for the bias
that is induced by shrinkage, which is important
for getting correct frequentist coverage.
Because our parameters $\theta_{[1]}, \dots, \theta_{[N]}$ are multivariate, we extend this method to the multivariate case. 

Suppose that $\theta_{[1]}, \dots, \theta_{[N]} \in \reals^d$ are drawn from some unknown distribution with mean $\theta_o$ and variance-covariance matrix $\Phi^{(2)}$ and that,
conditional on $\theta_{[j]}$,
\V{\begin{equation} \label{eq::normality}
\hat{\theta}_{[j]} | \theta_{[j]} \sim \distNorm(\theta_{[j]}, T_{[j]}^{-1} \Upsilon_{[j]}),
\end{equation}
where $T_{[j]}=T_{[j]}(N)$ is the number of observed time points at region $j$, and $\Upsilon_{[j]} \in \reals^{d \times d}$ is the $T_{[j]}$-normalized variance matrix of the estimation error at region $j = 1, \dots, N$.}
This Normality assumption holds, at least approximately,
due to the asymptotic Normality of the MLE.
Define the 
shrinkage estimator
\begin{equation} \label{eq::shrink}
\tilde{\theta}_{[j]}
\equiv \theta_o + W_{[j]}(\hat{\theta}_{[j]} - \theta_o),
\end{equation}
where 
\V{$$
W_{[j]} \equiv T_{[j]} (T_{[j]} \Upsilon_{[j]}^{-1} + \Phi^{(2)-1})^{-1} \Upsilon_{[j]}^{-1} = 
\Phi^{(2)} (\Phi^{(2)} + T_{[j]}^{-1} \Upsilon_{[j]})^{-1}.
$$}
For now, assume that
$\theta_o$ and $\Phi^{(2)}$ are specified.
Later, $\theta_o$ and $\Phi^{(2)}$ will be replaced
by empirical estimates.
(The estimator in \cref{eq::shrink}
can be motivated by noting that it is the
posterior mean of $\theta_{[j]}$ if
$\theta_{[j]} \sim \distNorm(\theta_o, \Phi^{(2)})$.
However, this distributional assumption
is not used in what follows.)

While shrinkage has the effect of
reducing the variance
of the estimator to $T_{[j]}^{-1} \tilde\Upsilon_{[j]}$ where $\tilde{\Upsilon}_{[j]} \equiv W_{[j]} \Upsilon_{[j]} W_{[j]}$ (note that $W_{[j]} \preceq id$),
it adds bias,
and as a result, standard confidence intervals will
not have correct coverage.
Instead, 
consider a confidence region for $\theta_{[j]}$ of the form
\V{\begin{equation*}
\begin{aligned}
    \mathcal{C}_{[j]} 
    & = \tilde\theta_{[j]} + \chi \cdot T_{[j]}^{-1/2} \tilde\Upsilon_{[j]}^{1/2} \mathcal{B}_d
    \equiv \{ \tilde\theta_{[j]} + \chi \cdot T_{[j]}^{-1/2} \tilde\Upsilon_{[j]}^{1/2} x: x \in \mathcal{B}_d \} \\
    & = \{\theta_{[j]}: (\theta_{[j]}-\tilde\theta_{[j]})^T \tilde\Upsilon_{[j]}^{-1} 
    (\theta_{[j]} - \tilde\theta_{[j]}) \leq T_{[j]}^{-1} \chi^2\}
\end{aligned}
\end{equation*}}
where $\chi>0$,
$\tilde\Upsilon_{[j]}^{1/2}$ is the matrix square-root of $\tilde\Upsilon_{[j]}$, 
$\mathcal{B}_d \equiv \{x \in \reals^d: \norm{x}_2 \leq 1\}$ is the $d$-dimensional unit ball.


If the model
$\theta_{[j]} \sim \distNorm(\theta_o, \Phi^{(2)})$
were correct,
we could use
$\chi = \chi_{1-\alpha}$ where $\chi^2_{1-\alpha}$ is the $1-\alpha$ quantile of $\chi^2$ distribution with degree of freedom $d$. However, 
the resulting confidence set does not have $1-\alpha$ coverage
in general.
See \citet{armstrong2022robust} for a detailed discussion in univariate cases. 
We instead use the approach by \citet{armstrong2022robust}.
\V{Conditional on $\theta_{[j]}$, the estimator $\tilde\theta_{[j]}$ has bias $\theta_o + W_{[j]} (\theta_{[j]} - \theta_o) - \theta_{[j]}$ 
and variance $T_{[j]}^{-1} W_{[j]} \Upsilon_{[j]} W_{[j]}^\top$.}
We define the normalized bias $b_{[j]}$ by
\V{\begin{equation*}
\begin{aligned}
    b_{[j]}
    & \equiv T_{[j]}^{1/2} (W_{[j]} \Upsilon_{[j]}^{1/2})^{-1} 
    (\theta_o + W_{[j]} (\theta_{[j]} - \theta_o) - \theta_{[j]}) \\
    & = T_{[j]}^{1/2} \Upsilon_{[j]}^{-1/2} (id - W_{[j]}^{-1}) \epsilon_{[j]}
    = - T_{[j]}^{-1/2} \Upsilon_{[j]}^{1/2} \Phi^{(2)-1} \epsilon_{[j]},
\end{aligned}
\end{equation*}}
where $\epsilon_{[j]} = \theta_{[j]} - \theta_o$. 
The non-coverage $r_{d-1}$ of $\mathcal C_{[j]}$ depends only on the distribution of $\norm{b_{[j]}}_2$:
\begin{equation} \label{eq::r_d-1}
r_{d-1}(\norm{b_{[j]}}_2^2, \chi) \equiv \Pr[ \norm{Z_d - b_{[j]}}_2 \geq \chi | \theta_{[j]} ] = 
\Pr[ \chi_{d-1}^2 + (Z - \norm{b_{[j]}}_2)^2 \geq \chi^2],
\end{equation}
due to the symmetry of $Z_d$, where $Z_d$, $\chi_{d-1}^2$ and $Z$ are a $d$-dimensional standard Gaussian random vector, a Chi-squared random variable with degree of freedom $d-1$ and a univariate standard Gaussian random variable. 
\V{Now let
\begin{equation} \label{eq::maximal_noncoverage}
\begin{aligned}
\rho_{m^{(2)}_{[j]},m^{(4)}_{[j]}}(\chi) \equiv 
\sup_{F \in \mathcal{F}} \Exp_F[r_{d-1}(u, \chi)] 
\text{ such that } \Exp_F[u] = m^{(2)}_{[j]} \textand \Exp_F[u^2] = m^{(4)}_{[j]},
\end{aligned}
\end{equation}
where $m^{(2)}_{[j]} = \Exp[\norm{b_{[j]}}_2^2]$, $m^{(4)}_{[j]} = \Exp[\norm{b_{[j]}}_2^4]$ and $\mathcal{F}$ is the collection of all such probability distribution on $[0, \infty)$.} 
That is,
$\rho_{m^{(2)}_{[j]},m^{(4)}_{[j]}}(\chi)$
is the maximum non-coverage of $\mathcal C_j$ at each $\chi$, 
subject to moment constraints.
The moments of $\norm{b_{[j]}}_2$ are calculated from the moments of $\theta_{[j]}$: if
\begin{equation*}
    \Exp[(\theta_{[j]} - \theta_o)(\theta_{[j]} - \theta_o)^\top] = \Phi^{(2)} \textand
    \Exp[(\theta_{[j]} - \theta_o)^{\otimes 4}] = \Phi^{(4)},
\end{equation*}
where $\otimes$ is the tensor product, the moments of $\norm{b_{[j]}}_2$ are
\V{\begin{equation} \label{eq::m_from_Upsilon}
\begin{aligned}
    & m^{(2)}_{[j]} = T_{[j]}^{-1} \tr\left( \Upsilon_{[j]}^{1/2} \Phi^{(2)-1} \Upsilon_{[j]}^{1/2} \right) \\
    & m^{(4)}_{[j]} = T_{[j]}^{-2} \sum_{i_1=1}^d \sum_{i_2=1}^d \Psi_{i_1,i_2}, \\
\end{aligned}
\end{equation}}
where $\Psi_{i_1,i_2} \equiv \inner*{\Phi^{(4)}, (\Upsilon_{[j]}^{1/2}\Phi^{(2)-1})_{i_1}^{\otimes 2} \otimes (\Upsilon_{[j]}^{1/2}\Phi^{(2)-1})_{i_2}^{\otimes 2}}$.
As in univariate cases, $\rho_{m^{(2)}_{[j]},m^{(4)}_{[j]}}$ can be estimated by numerically solving \cref{eq::maximal_noncoverage} with discretized support of $\norm{b_{[j]}}_2$; see Appendix B of \citet{armstrong2022robust}.
Finally, we define the critical value 
\begin{equation*}
\tilde\chi_{[j],1-\alpha} \equiv 
\inf\{ \chi: \rho_{m^{(2)}_{[j]},m^{(4)}_{[j]}}(\chi) \leq \alpha\},
\end{equation*}
and define the confidence set
\V{\begin{equation} \label{eq::EBCR}
\tilde {\mathcal C}_{[j]} \equiv \tilde{\theta}_{[j]} + \tilde\chi_{[j],1-\alpha} T_{[j]}^{-1/2} \tilde\Upsilon_{[j]}^{1/2} \mathcal{B}_d.
\end{equation}}
Since $\theta_o$, $\Phi^{(2)}$ and $\Phi^{(4)}$ are unknown, 
we estimate them from $\hat\theta_{[1]} \dots \hat\theta_{[N]}$, and let $\hat\theta_o$, $\hat\Phi^{(2)}$ and $\hat\Phi^{(4)}$ be the estimates. Plugging the estimates in \cref{eq::EBCR}, we obtain the proposed empirical Bayesian confidence set. 

\V{
The form of the shrinkage estimator was motivated
by assuming that
$\theta_{[j]} \sim \distNorm(\theta_o, \Phi^{(2)})$
but we show that
$\tilde{\mathcal C}_{[j]}$
has the correct coverage
without that assumption.
Indeed, we do not even need to 
assume that $\theta_{[1, \dots, N]}$
are random.
We instead formulate the coverage conditional on $\theta_{[1,\dots,N]}$ (and $\Upsilon_{[1,\dots,N]}$) 
which implies robustness of the coverage against prior misspecification. 
As in \cite{morris1983parametric},
we say the confidence regions $\mathcal{C}_{[j]}$ has asymptotically $1 - \alpha$ 
average coverage intervals (ACIs) if 
\begin{equation*}
\liminf_{N \rightarrow \infty} \frac{1}{N} 
\sum_{j=1}^N \Pr[\theta_{[j]} \in \mathcal{C}_{[j]} | 
\theta_{[1, \dots, N]}, \Upsilon_{[1, \dots, N]}] \geq 1 - \alpha.
\end{equation*}}
\V{The asymptotic coverage 
relies on the conditional normality (\cref{eq::normality}) of $\hat\theta_{[j]}$.
Instead, we will only assume that $\hat\theta_{[j]}$ 
is asymptotically Normal with an estimated covariance matrix $\hat\Upsilon_{[j]}$ (in our case, $\hat\Upsilon_{[j]}$ in \cref{thm::sandwich}).}

\begin{assumption} \label{assmp::asymp_normality}
$T_{[j]}(N) \rightarrow \infty$ as $N \rightarrow \infty$ so that the estimators are uniformly asymptotically Normal:
given $\mathcal{U}$ be the collection of balls in $\reals^d$, suppose that
\begin{equation*}
\lim_{N \rightarrow \infty} 
\max_{1 \leq j \leq N} \sup_{U \in \mathcal{U}} 
\Bigl|\Pr[\V{T_{[j]}^{1/2}} \hat\Upsilon_{[j]}^{-1/2} 
(\hat\theta_{[j]} - \theta_{[j]}) \in U | \theta_{[1, \dots, N]}, \Upsilon_{[1, \dots, N]}]-
\Pr[Z_d \in U]\Bigr| = 0,
\end{equation*}
where $Z_d$ is the $d$-dimensional standard Normal random variable.
\end{assumption}

\begin{assumption} \label{assmp::var_consistency}
$\hat\Upsilon_{[j]}$ is consistent.
That is, for every $\epsilon >0$,
\begin{equation*}
\lim_{N \rightarrow \infty} \max_{1 \leq j \leq N} 
\Pr\Bigl[\norm{\hat\Upsilon_{[j]} - \Upsilon_{[j]}}_2 \geq 
\epsilon \Bigm| \theta_{[1, \dots, N]}, \Upsilon_{[1, \dots, N]} \Bigr] = 0.
\end{equation*}
\end{assumption}


\V{
Another essential ingredient of the asymptotic coverage is the homoskedasticity of $\theta_{[j]}$ in the second and fourth moments. Because the theoretical argument is given conditional on $\theta_{[1,\dots,N]}$, we formulate and assume the homoskedasticity as follows. 
}

\begin{assumption} \label{assmp::var_space}
    \V{Suppose that $\sup_{[j]} \lambda_{\max}(\Upsilon_{[j]}) \leq \lambda_{\max,\Upsilon}$ for some $\lambda_{\max,\Upsilon} \in (0, \infty)$ and that $\mathcal{U}$ is the collection of $d \times d$ positive semidefinite matrices $\Upsilon$ satisfying
    $
        \lambda_{\max}(\Upsilon) \leq \lambda_{\max,\Upsilon}.
    $
    We assume there exists $\theta_o$ such that $\hat\theta_o \rightarrow \theta_o$ in probability and $(\Phi^{(2)}, \Phi^{(4)})$ such that
    \begin{equation*}
        \frac{1}{N_\mathcal{X}} \sum_{j: j \in \mathcal{I}_\mathcal{X}} (\theta_{[j]} - \theta_o)^{\otimes 2} \rightarrow \Phi^{(2)} \textand
        \frac{1}{N_\mathcal{X}} \sum_{j: j \in \mathcal{I}_\mathcal{X}} (\theta_{[j]} - \theta_o)^{\otimes 4} \rightarrow \Phi^{(4)}
    \end{equation*}
    as $N_\mathcal{X} \rightarrow \infty$ for any Borel set $\mathcal{X} \subseteq \mathcal{U}$, where $\mathcal{I}_\mathcal{X} \equiv \{j: \Upsilon_{[j]} \in \mathcal{X}\}$, $N_\mathcal{X} \equiv \abs{\mathcal{I}_\mathcal{X}}$.
    In addition, we assume there exists $\lambda_{\min,\Phi}, \lambda_{\max,\Phi} \in (0, \infty)$ such that
    \begin{equation*}
    \begin{aligned}
        & 
        \lambda_{\min,\Phi} \leq \lambda_{\min}(\Phi^{(2)}) \leq \lambda_{\max}(\Phi^{(2)}) \leq \lambda_{\max,\Phi}, \\
        & 
        \lambda_{\min,\Phi} \leq \lambda_{\min}(\Phi^{(4)}) \leq \lambda_{\max}(\Phi^{(4)}) \leq \lambda_{\max,\Phi}.
    \end{aligned}
    \end{equation*}
    For brevity, let $\lambda_{\max,*} \equiv \max\{\lambda_{\max,\Upsilon}, \lambda_{\max,\Phi}, \lambda_{\min,\Phi}^{-1}$\}.}
\end{assumption}

\begin{assumption} \label{assmp::moment_consistency}
Finally, we assume the estimates $\hat\Phi^{(2)}$ and $\hat\Phi^{(4)}$ are consistent. 
    \begin{equation*}
        \hat\Phi^{(2)} \stackrel{P}{\to} \Phi^{(2)} \textand
        \hat\Phi^{(4)} \stackrel{P}{\to} \Phi^{(4)}
        ~\text{as}~ N \rightarrow \infty.
    \end{equation*}
\end{assumption}





\V{\cref{thm::average_coverage}
shows that the average miscoverage
of the confidence sets
is $\alpha$, asymptotically. The proof is in~\cref{app::shrinkage}.}


\begin{theorem} \label{thm::average_coverage}
    Under assumptions~\ref{assmp::asymp_normality}, \ref{assmp::var_consistency}, \ref{assmp::var_space} and \ref{assmp::moment_consistency},
    for any Borel set $\mathcal{X} \subseteq \mathcal{U}$, 
    \begin{equation*}
        \Exp\left[ \left. \frac{1}{N_{\mathcal{X}}} \tsum_{j \in \mathcal{I}_\mathcal{X}} \mathbb{I}\{\theta_{[j]} \notin \tilde{\mathcal{C}}_{[j]} \} \right| \theta_{[1, \dots, N]}, \Upsilon_{[1, \dots, N]} \right] \leq \alpha + o(1),
    \end{equation*}
    where $\mathcal{I}_\mathcal{X} \equiv \{j: \Upsilon_{[j]} \in \mathcal{X}\}$, and $N_\mathcal{X} \equiv \abs{\mathcal{I}_\mathcal{X}}$.
\end{theorem}

\V{We note that the average coverage result in \cref{thm::average_coverage} remains valid no matter the particular choice of $\hat\theta_o$, $\hat\Phi^{(2)}$ and $\hat\Phi^{(4)}$ provided that Assumption~\ref{assmp::moment_consistency} holds.} For the subsequent simulation study and data analysis, we use the \emph{unconstrained} estimating scheme of \citet{armstrong2022robust} (see Appendix A.1 therein); namely, we take 
\begin{equation*}
\begin{aligned}
    \hat\theta_o = \frac{\tsum_j \xi_{[j]} \hat\theta_{[j]}}{\tsum_j \xi_{[j]}},
    \ \ \ 
    \hat\Phi^{(2)} = \frac{\tsum_j \xi_{[j]} (\hat\epsilon_{[j]} \hat\epsilon_{[j]}^\top - \Upsilon_{[j]})}{\tsum_j \xi_{[j]}} \ \ \ \textand\ \ \ 
    \hat\Phi^{(4)} = \frac{\tsum_j \xi_{[j]} \hat\Phi^{(4)}_{[j]}}{\tsum_j \xi_{[j]}}, \\
\end{aligned}
\end{equation*}
where $\hat\epsilon_{[j]} \equiv \hat\theta_{[j]} - \hat\theta_o$, $\xi_{[j]}$ is some weight, and $\hat\Phi^{(4)}_{[j]}$ is a $d^4$-dimensional order $4$ tensor such that
\begin{equation*}
\begin{aligned}
    \hat\Phi^{(4)}_{[j],i_1,i_2,i_3,i_4} & \equiv 
    \hat\epsilon_{[j],i_1} \hat\epsilon_{[j],i_2} \hat\epsilon_{[j],i_3} \hat\epsilon_{[j],i_4} 
    + \Upsilon_{[j],i_1,i_2} \Upsilon_{[j],i_3,i_4}
    + \Upsilon_{[j],i_1,i_3} \Upsilon_{[j],i_2,i_4}
    + \Upsilon_{[j],i_1,i_4} \Upsilon_{[j],i_2,i_3} \\
    & \quad - \Upsilon_{[j],i_1,i_2} \hat\epsilon_{[j],i_3} \hat\epsilon_{[j],i_4} 
    - \Upsilon_{[j],i_1,i_3} \hat\epsilon_{[j],i_2} \hat\epsilon_{[j],i_4}
    - \Upsilon_{[j],i_1,i_4} \hat\epsilon_{[j],i_2} \hat\epsilon_{[j],i_3} \\
    & \quad - \Upsilon_{[j],i_2,i_3} \hat\epsilon_{[j],i_1} \hat\epsilon_{[j],i_4}
    - \Upsilon_{[j],i_2,i_4} \hat\epsilon_{[j],i_1} \hat\epsilon_{[j],i_3}
    - \Upsilon_{[j],i_3,i_4} \hat\epsilon_{[j],i_1} \hat\epsilon_{[j],i_2},
\end{aligned}
\end{equation*}
where $\hat\epsilon_{[j],i}$ is the $i$-th component of the vector $\hat\epsilon_{[j]}$, and $\Upsilon_{[j],i_1,i_2}$ is the $(i_1,i_2)$ element of the matrix $\Upsilon_{[j]}$.


\V{A simple choice of $\xi_{[j]}$ is $1$, motivated by the definition of $\Phi^{(2)}$ and $\Phi^{(4)}$ in Assumption~\ref{assmp::var_space}. However, our preferred choice is $\xi_{[j]} = T_{[j]} \det(\hat\Upsilon_{[j]})^{-1/d}$ for the sake of efficiency as suggested by \citet{armstrong2022robust}; see Appendix A.2 therein.}

\section{Examples \label{sec::examples}}

\begin{figure}[t!]
    \centering
    \subfigure[][]{\label{fig::g} \includegraphics[height=0.3\textwidth]{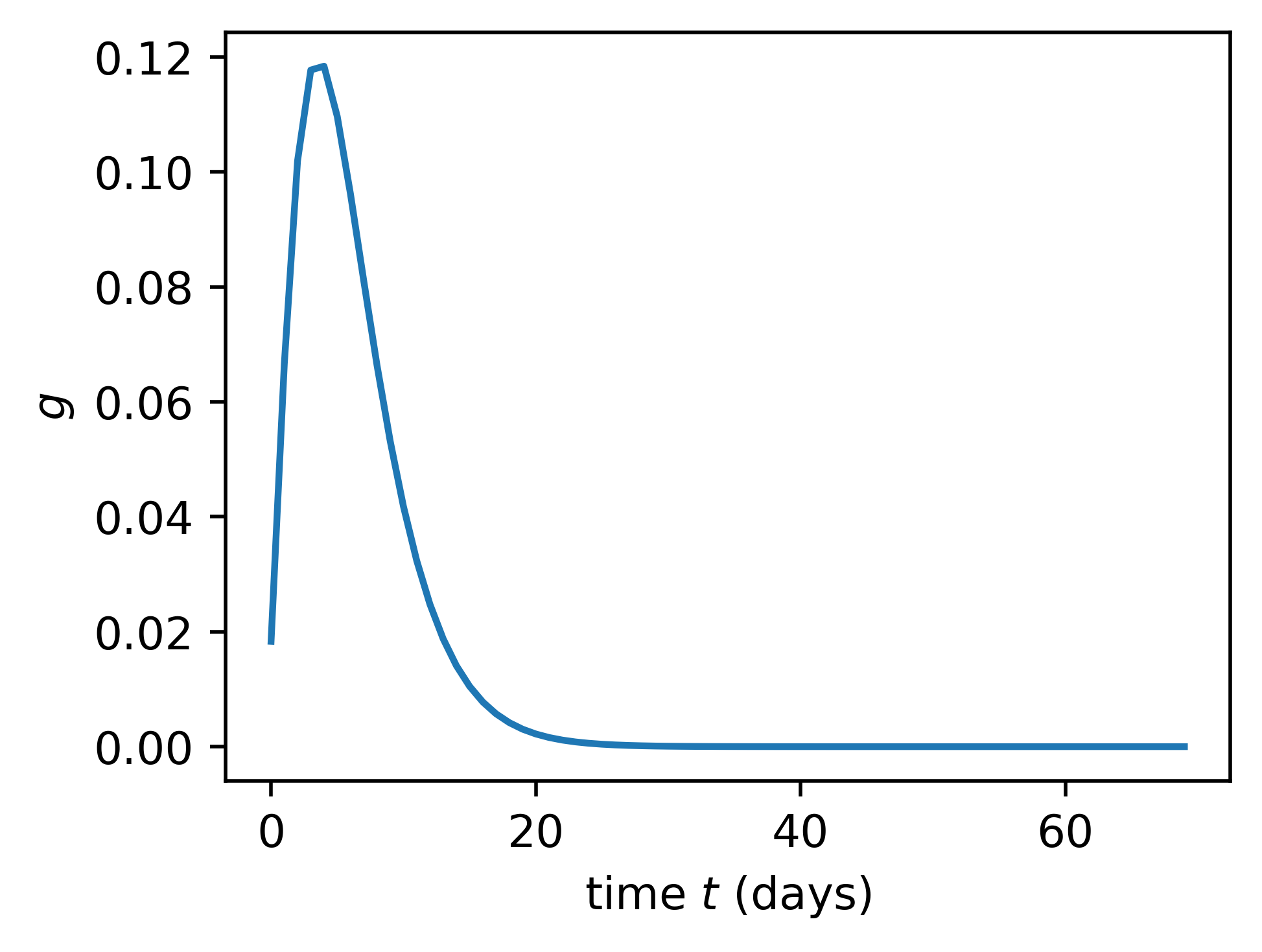}}
    \subfigure[][]{\label{fig::pi} \includegraphics[height=0.3\textwidth]{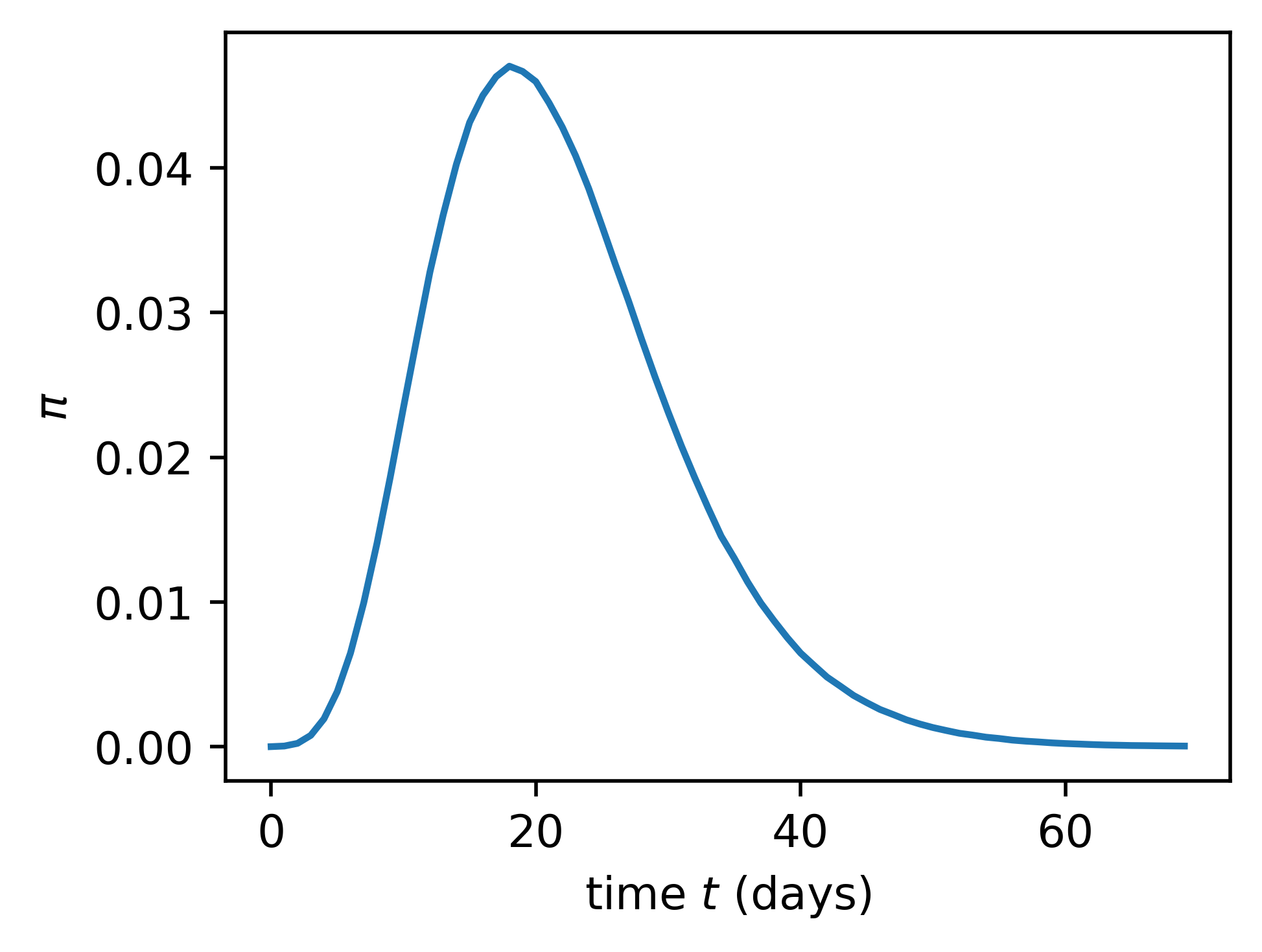}} 
    \caption{\textbf{ Generating and infection-to-death distributions $g$ and $\pi$} from \cite{bhatt2020semi}. 
    }
    \label{fig::gpi}
\end{figure}

In \cref{sec::simulated} we illustrate the proposed methods and check their statistical validity using simulated data.
Then in \cref{sec::data} we analyze the effect of a mobility measure on the Covid-19 death data used in \citet{Bonvini2022covid}.
We assume that the observed data (deaths) are either Poisson or NB distributed, with means specified in \cref{eq::model00}.
We focus on a single intervention, so that $\beta_1$ and $A_t$ in \cref{eq::Rt.model0} are scalars.
Several model parameters are completely unidentifiable and need to be fixed. 
In particular, the maximum possible transmission rate was assumed to be $K=6.5$, which is the largest value found in the literature \citep{liu2020reproductive}.
'The generating distribution $g$ and infection-to-death distribution $\pi$ in \cref{eq::model00} were assumed to be known and equal to those given in \citet{bhatt2020semi}; they are shown in \cref{fig::gpi}.
We also assumed that the ascertainment rate (the probability of death given infection) in \cref{eq::model00} remained constant at $\alpha_t = 0.01$. \V{The constant assumption
is sensible here because the period was short and the number of susceptibles at the beginning of the epidemic was large.} The specific value $\alpha_t=0.01$ does not have to be 
accurate because it has no effect on the model parameters -- any other value would yield the same estimates; it only has a multiplicative effect on the unknown latent infection values, which are not of primary interest here.
Finally, we set $I_{t}=0$ for $t \leq -T_0$ and $I_{t}=e^\mu$ for $t = -T_0+1, \dots, 0$, where $\mu$ is a parameter to be estimated and $T_0=40$, as described in \cref{sec::inference} under `{\em seeding values}'.

\subsection{Analyses of simulated datasets \label{sec::simulated}}

To illustrate that the algorithm in \cref{sec::algorithm} yields sensible ML estimates, we simulated infection and death processes 
$I_t$ and $Y_t$, $t=1, \ldots 120$, from NB distributions \V{(both $I_t$ and $Y_t$ are random variables)} with means 
specified in \cref{eq::model0}
and ``number of successes" parameters $r_I = 100$ for the infection process and $r_Y = 10$ for the death process, $R_t$ in \cref{eq::Rt.model0} with
$(\mu,\beta_0, \beta_1) = (\log(100), 0, -2.2)$, 
and time-varying binary interventions $A_{1t}=0$ for $t<30$ and $A_{1t}=1$ for $t\ge 30$. 
We fit \V{the true mean model} to the data by ML and using the Bayesian approach in \citet{bhatt2020semi}, assuming the correct NB model for the deaths $Y_t$ but assuming that infections $I_t$ were deterministic, \V{since we 
cannot identify both the distributions of $Y_t$ and $I_t$; see \cref{sec::inference}}.
For the Bayesian approach, we used a shifted gamma prior with shape $1/6$, scale $1$ and shift $\log(1.05)/6$ for $\beta_1$, corresponding to mean $(1+\log(1.05))/6 = 0.17$ and variance $1/6$, and a $N(0,1/4)$ prior for $\beta_0$.
These choices match the model, priors, intervention values and parameter estimates 
\V{in \cite{bhatt2020semi}, where they analyzed the effect of lockdown ($A_t=1$) versus no lockdown ($A_t=0$)
on Covid cases in European countries.}
Finally, for $\mu$ we used the default prior 
in the \texttt{Epidemia} package:
$e^\mu \sim \mathrm{Exp}(e^{-\mu_o})$ with $e^{\mu_o} \sim \mathrm{Exp}(\lambda)$ and $\lambda = 0.03$, so that $\mu$ has prior mean $\log(1/0.03)=1.52$, which is close to the true value $\log(100)=1.6$.

True parameters and simulated data are shown in orange in 
\cref{fig::sim_nbinom}.
Bayesian and ML estimates for $(\mu,\beta_0, \beta_1)$ and for the mean infection and death processes are overlaid in blue,
together with confidence and credible intervals.
All estimates appear reasonable and the confidence/credible intervals cover the true values, but with some notable differences between estimation methods.
\V{Appendix \cref{fig::theta_epidemia_priors} shows that the Bayesian estimates and credible 
intervals can have poor properties when we use more informative priors.}
Next, we conduct a simulation study to evaluate the coverage of confidence and credible intervals.

\begin{figure}[t!]
    \centering
    \large{\bf Maximum Likelihood Estimation}
    
    \subfigure[][]{\label{fig::theta_freqepid_nbinom} 
    \includegraphics[width=0.66\textwidth]{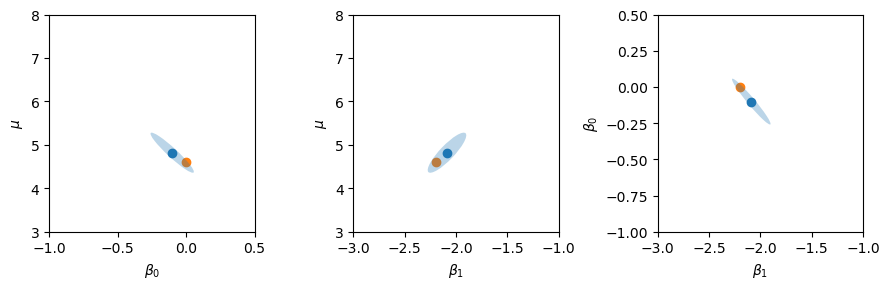}}

    \subfigure[][]{\label{fig::pred_I_freqepid_nbinom}
    \includegraphics[height=0.22\textwidth]{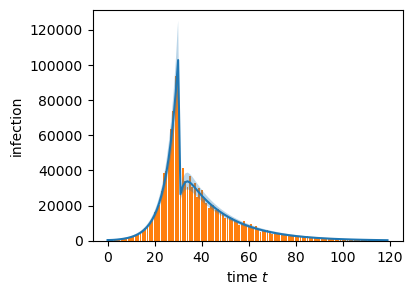}}
    \subfigure[][]{\label{fig::pred_EY_freqepid_nbinom}
    \includegraphics[height=0.22\textwidth]{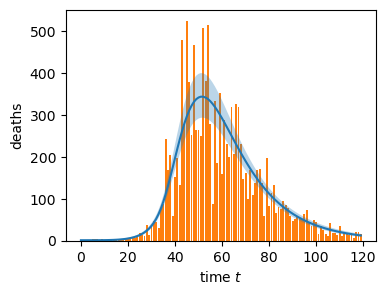}}

    \medskip
    
    \centering
    \large{\bf Bayesian Estimation}
    
    \subfigure[][]{\label{fig::theta_epidemia_nbinom} 
    \includegraphics[width=0.66\textwidth]{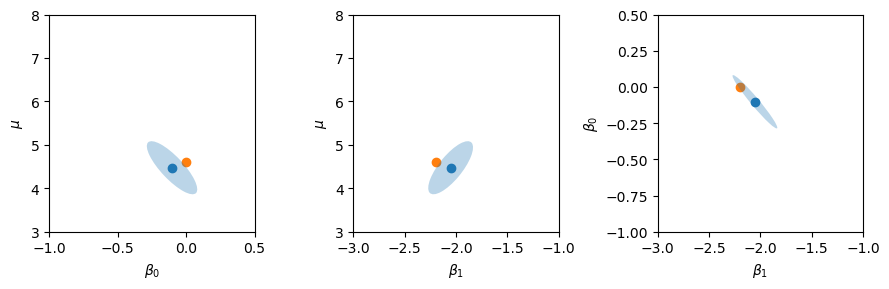}}

    \subfigure[][]{\label{fig::pred_I_epidemia_nbinom}
    \includegraphics[height=0.22\textwidth]{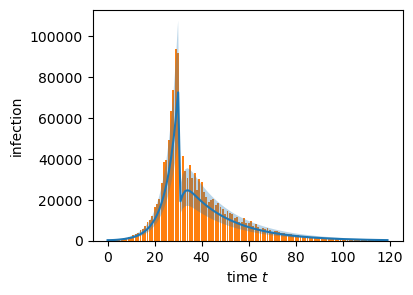}}
    \subfigure[][]{\label{fig::pred_EY_epidemia_nbinom}
    \includegraphics[height=0.22\textwidth]{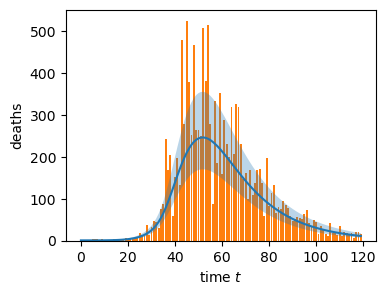}}
    
    \caption{\textbf{ True model parameters, simulated infections and deaths (orange), and overlaid ML and Bayesian estimates (blue), together with 95\% confidence/credible intervals (light blue).}
    The infection process in (b,e) and death process in (c,f) were simulated from NB distributions with means 
    in \cref{eq::model0}, $R_t$ in \cref{eq::Rt.model0}, true parameters $\mu= \log(100)$, $\beta_0=0$ and $\beta_1 = -2.2$ in (a,d), and binary intervention process $M_t=0$ for $t<30$ and $M_t=1$ for $t \ge 30$.
    The model in \cref{eq::model00,eq::Rt.model0} was fitted assuming NB deaths.}
    \label{fig::sim_nbinom}
\end{figure}

\paragraph{Coverage of ML and Bayesian intervals}

\begin{table}
  \caption{\label{tab::result_coverage} \emph{Coverage study results.} Coverages (\%) of ML confidence and Bayesian credible intervals of $\beta_1$ for $5$ different values of $\beta=(\beta_0, \beta_1)$ in \cref{eq::Rt.model0} and seeding parameter
  $\mu=\log(100)$ in all cases. The priors for $\mu$, $\beta_0$ and $\beta_1$ were hierarchical exponential,
  $N(0,1/4)$ and shifted Gamma with mean $0.17$ and variance $1/6$, respectively, as in \cite{bhatt2020semi}.}
  \begin{tabular}{lccccc}
    \toprule
    \textbf{$\beta$} & \textbf{(0, -2.2)} & \textbf{(0.25, -2.45)} & \textbf{(0.5, -2.7)} & \textbf{(0.75, -2.95)} & \textbf{(1, -3.2)} \\
    \midrule
    ML fit & 93.2 & 95.5 & 93.7 & 96.1 & 95.0 \\
    \citet{bhatt2020semi} & 93.3 & 95.2 & 93.6 & 92.3 & 84.1 \\
    \bottomrule
  \end{tabular}
\end{table}

Accurate inference about parameters is important not only to understand, for example, the effect of particular interventions, measured by $\beta_1$, but also to produce accurate future predictions of the time course of epidemics, which is one main objective of epidemic models.

To evaluate the coverage of confidence and credible intervals, we set $(\mu,\beta_0, \beta_1) = (\log(100), 0, -2.2)$, simulated $1000$ NB infection and death processes as above, 
fitted the NB model using both ML and Bayesian paradigms with priors as above, and recorded the proportion of times the confidence and credible intervals contained the true parameter. \cref{tab::result_coverage} contains the resulting empirical 95\% coverages; other coverage values gave qualitatively similar results (not shown). We repeated the simulation for four other parameter values. The ML intervals have correct coverage up to simulation
error in all cases, whereas the coverage of the Bayesian interval degrades when the prior means of $\beta_0$ and $\beta_1$,
0 and 0.17 respectively, further deviate from the true values.
Note that it is not easy to choose a prior without additional information, as illustrated in \cref{fig::theta_0_vs_4}, which shows simulated death processes when $(\mu,\beta_0, \beta_1) = (\log(100), 0, -2.2)$ and $(\mu,\beta_0, \beta_1) = (\log(25), 1, -3.2)$: all simulated data look very similar and we would be hard pressed to pick different priors for them.
It would therefore be wise to select priors that are not too informative.

\paragraph{Model checking}

\begin{figure}[t!]
    \centering
    \begin{minipage}{0.3\textwidth}
        \subfigure[][]{\label{fig::diagnostic_normal_vs_t} 
        \includegraphics[height=0.75\textwidth]{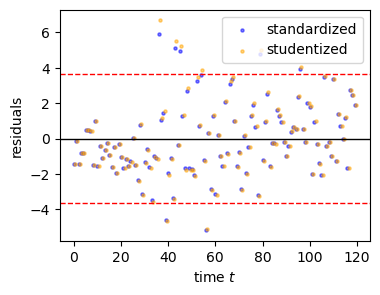}}
    \end{minipage}
    \hspace{-0.01\textwidth}
    \begin{minipage}{0.3\textwidth}
        \subfigure[][]{\label{fig::diagnostic_normal_vs_A} 
        \includegraphics[height=0.75\textwidth]{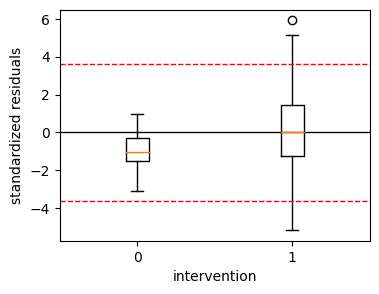}}
    \end{minipage}
    \hspace{-0.01\textwidth}
    \begin{minipage}{0.3\textwidth}
        \subfigure[][]{\label{fig::diagnostic_normal_vs_EYmle} 
        \includegraphics[height=0.75\textwidth]{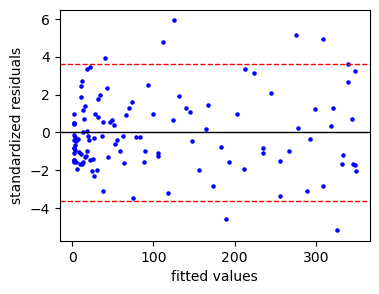}}
    \end{minipage}

    \begin{minipage}{0.3\textwidth}
        \subfigure[][]{\label{fig::diagnostic_normal_qq} 
        \includegraphics[height=0.75\textwidth]{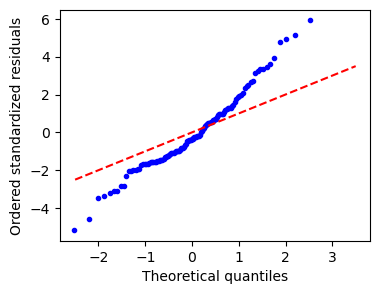}}
    \end{minipage}
    \hspace{-0.01\textwidth}
    \begin{minipage}{0.3\textwidth}
        \subfigure[][]{\label{fig::diagnostic_normal_influence} 
        \includegraphics[height=0.75\textwidth]{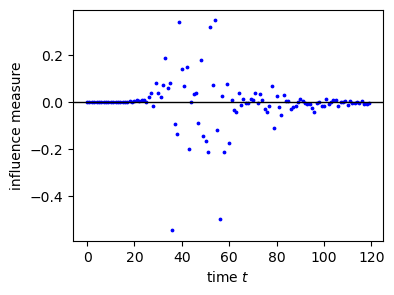}}
    \end{minipage}
    \hspace{-0.01\textwidth}
    \begin{minipage}{0.3\textwidth}
        \subfigure[][]{\label{fig::diagnostic_normal_cook} 
        \includegraphics[height=0.75\textwidth]{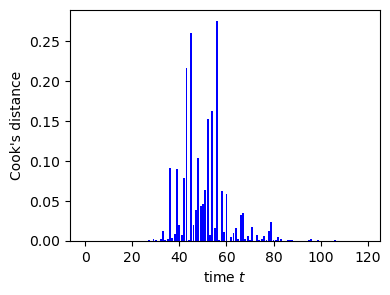}}
    \end{minipage}
    \vspace{-.1in}   
    \caption{\textbf{ Diagnostics for the Gaussian model fitted to the NB data shown in \cref{fig::sim_nbinom}.} (a) Standardized and studentized residuals versus $t$. (b) Standardized residuals versus $A_t$ and (c) versus $\hat Y_t$. (d) QQ plot of standardized residuals. (e) Influence on $\hat \beta_1$ and (f) Cook's distance, versus $t$. \V{Panels (b) and (d) suggest that the Gaussian distribution is not a good choice to model the variance of $Y_t$.}}
    \label{fig::diagnostic_normal}
\end{figure}

Bayesian and ML approaches provide the alternative option of fitting a Gaussian model to the data.
Here we illustrate that the standard diagnostics derived from an ML fit can detect model inadequacies; see \cref{sec::diagnostics}.
\cref{app::figs} \cref{fig::diagnostic_nbinom} shows these diagnostics for the NB model fitted by ML to the simulated NB data shown in \cref{fig::sim_nbinom}. As expected, all diagnostics look fine, since we fit the correct model (except that the model assumes that the latent infection process is deterministic).
\cref{app::figs} 
\cref{fig::sim_freqepid_normal} 
shows the Gaussian model fitted to the same data, and 
\cref{fig::diagnostic_normal} shows the 
corresponding diagnostics, where it is clear particularly from panels (b) and (d) that the Gaussian distribution is not a good choice to model the variance of $Y_t$.
(When model diagnostics fail, as is the case in \cref{fig::diagnostic_normal}, case diagnostics in (e,f) can look pathological even if there are no outliers or influential observations, so it is best to avoid over-interpreting them.)

\paragraph{\V{Empirical Bayes Analysis}}
Suppose now that we have data
from $N = 100 $ different geographic regions
with corresponding parameters $\theta_i = (\mu_i, \beta_{0i}, \beta_{1i})$, $i, \ldots, N$.
We estimate the $\theta_i$'s by maximum likelihood,
and take advantage of similarities between the parameters by shrinking the individual estimates towards each other, as described
in \cref{sec::shrinkage}. We now evaluate by simulation the coverage of the resulting confidence intervals, for the three scenarios \V{depicted in
\cref{fig::beta_ebcr_true}:} 
\begin{enumerate} 
    \item 
    $(\mu, \beta_0, \beta_1) $ are independent Gaussian random vectors with mean $(\log(300), 0, -2)$ and covariance diagonal $(0.5, 0.3, 0.3)$.
    \item 
    $(\beta_0, \beta_1)$ are Gaussian vectors with mean $(0, -2)$ and covariance 
    $\begin{pmatrix}
        0.3 & -0.28 \\ -0.28 & 0.3
    \end{pmatrix}$; $\mu$ is independent of $(\beta_0, \beta_1)$ and normally distributed with mean $\log(300)$ and variance $0.5$.    
    \item $(\mu, \beta_0, \beta_1)$ have a mixture distribution with two equal Gaussian components with means $(\log(150), 0.5, -2.5)$ and $(\log(600), -0.5, -1.5)$, and covariance matrix $\begin{pmatrix}
        1 & 0 & 0 \\ 
        0 & 0.5 & -0.25 \\
        0 & -0.25 & 0.5
    \end{pmatrix}$
    for both components. 
\end{enumerate}
The settings in (a), (b) and (c) were chosen so that the marginal means and variances of $\mu$, $\beta_0$ and $\beta_1$ would be the same, as well as the covariances of $(\beta_0, \beta_1)$ in (b) and (c).
These choices allowed us to use the same priors in the three scenarios, so coverage differences could not be attributed to difference in priors.
We used the hierarchical model of \citet{bhatt2020semi}:
\begin{equation}
\begin{aligned}
\label{eq::hier}
    e^{\mu_i} & \sim \mathrm{Expon}(e^{-\mu^*}), \\
    \beta_{0i} & \sim N(\beta_0^*, \sigma_0^2), \\
    \beta_{1i} & \sim N(\beta_1^*, \sigma_1^2),
\end{aligned}
\end{equation}
where 
$(\mu^*, \beta_0^*, \beta_1^*)$ and $(\mu_i, \beta_{0i}, \beta_{1i})$ denote the global and regional parameters, respectively,
with priors and hyper-priors
\begin{equation} \label{eq::hyper}
\begin{aligned}
    e^{\mu^*} & \sim \mathrm{Expon}(0.03), &  \\
    \beta_0^* & \sim N(0, 1/4), & \sigma_0 \sim \mathrm{Gamma}(2, 0.25),\\
    \beta_1^* & \sim -\log(1.05)/6 - \mathrm{Gamma}(1/6, 1), & \sigma_1 \sim \mathrm{Gamma}(0.5, 0.25). \\
\end{aligned}
\end{equation}

\cref{fig::beta_ebcr_true} shows $N=100$ simulated $(\beta_0, \beta_1)$ in each scenario,  \cref{fig::beta_ebcr_freqepid} shows their shrunk ML estimates together with 95\% confidence intervals, \cref{fig::beta_ebcr_epidemia} shows the corresponding Bayesian estimates, and
\V{\cref{tab:result_shrinkage}} contains the proportion of intervals that contain the true values of $(\mu, \beta_0, \beta_1)$.
Point and interval estimates are strikingly different across methods, the ML estimates being much closer to the true values.
The coverage of the credible intervals also degrades badly as the true dependence structure of the $(\mu_i, \beta_{0i}, \beta_{1i})$ deviates from their assumed independent prior distributions in \cref{eq::hier,eq::hyper}.
One could presumably obtain better Bayesian coverage with better priors, although they would be hard to design a priori.
On the other hand, 
ML estimation yields the correct coverage in all cases, up to random error.

\begin{figure}[t!]
    \centering
    \begin{minipage}{0.3\textwidth}
        \subfigure[][]{\label{fig::beta_true_ebcr_0} 
        \includegraphics[height=0.75\textwidth]{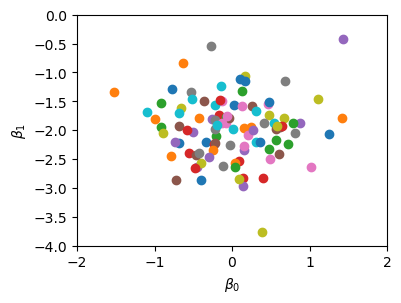}}
    \end{minipage}
    \hspace{-0.01\textwidth}
    \begin{minipage}{0.3\textwidth}
        \subfigure[][]{\label{fig::beta_true_ebcr_1} 
        \includegraphics[height=0.75\textwidth]{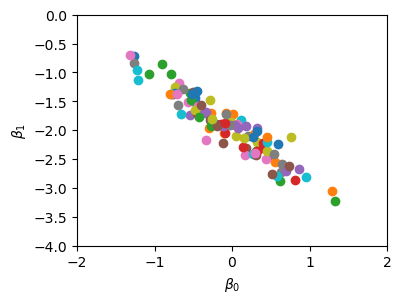}}
    \end{minipage}
    \hspace{-0.01\textwidth}
    \begin{minipage}{0.3\textwidth}
        \subfigure[][]{\label{fig::beta_true_ebcr_2} 
        \includegraphics[height=0.75\textwidth]{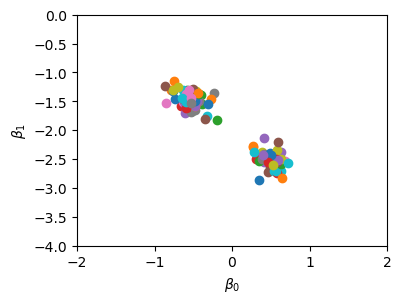}}
    \end{minipage}
    \vspace{-0.1in}
    \caption{\textbf{ Simultaneous estimation at $N=100$ locations.} True $\beta$'s simulated in three scenarios.}
    \label{fig::beta_ebcr_true}
\bigskip
    \centering
    \begin{minipage}{0.3\textwidth}
        \subfigure[][]{\label{fig::beta_freqepid_ebcr_0} 
        \includegraphics[height=0.75\textwidth]{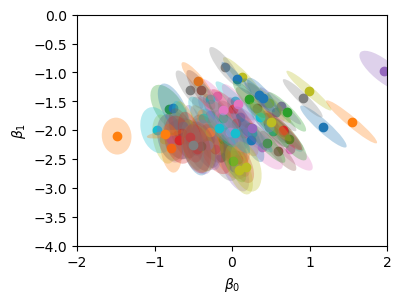}}
    \end{minipage}
    \hspace{-0.01\textwidth}
    \begin{minipage}{0.3\textwidth}
        \subfigure[][]{\label{fig::beta_freqepid_ebcr_1} 
        \includegraphics[height=0.75\textwidth]{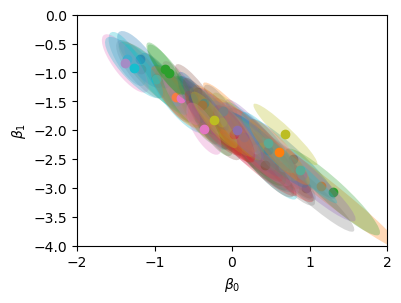}}
    \end{minipage}
    \hspace{-0.01\textwidth}
    \begin{minipage}{0.3\textwidth}
        \subfigure[][]{\label{fig::beta_freqepid_ebcr_2} 
        \includegraphics[height=0.75\textwidth]{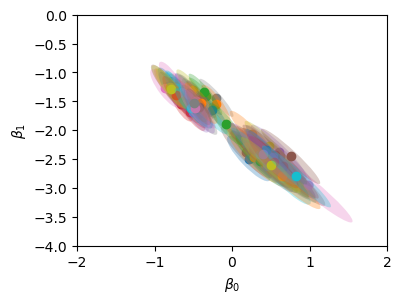}}
    \end{minipage}
    \vspace{-0.1in}
    \caption{Shrunk ML estimates of $\hat{\beta}$ and confidence intervals for the true $\beta$'s in \cref{fig::beta_ebcr_true}.}
    \label{fig::beta_ebcr_freqepid} 
\bigskip
    \centering
    \begin{minipage}{0.3\textwidth}
        \subfigure[][]{\label{fig::beta_epidemia_ebcr_0} 
        \includegraphics[height=0.75\textwidth]{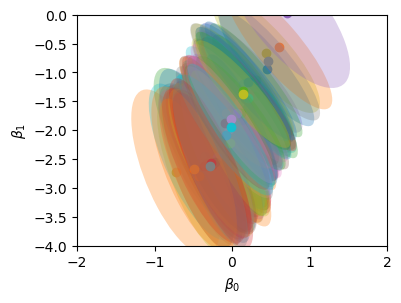}}
    \end{minipage}
    \hspace{-0.01\textwidth}
    \begin{minipage}{0.3\textwidth}
        \subfigure[][]{\label{fig::beta_epidemia_ebcr_1} 
        \includegraphics[height=0.75\textwidth]{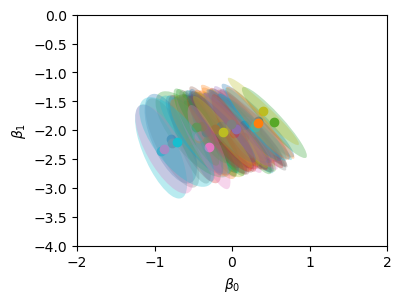}}
    \end{minipage}
    \hspace{-0.01\textwidth}
    \begin{minipage}{0.3\textwidth}
        \subfigure[][]{\label{fig::beta_epidemia_ebcr_2} 
        \includegraphics[height=0.75\textwidth]{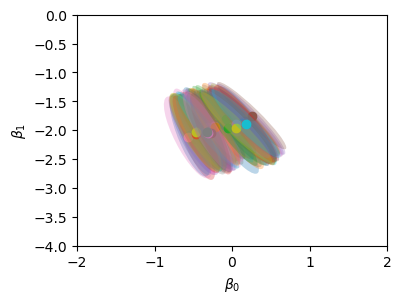}}
    \end{minipage}
    \vspace{-0.1in}
    \caption{Bayesian estimates and credible intervals for the true $\beta$'s in \cref{fig::beta_ebcr_true}, using \citet{bhatt2020semi}.}
    \label{fig::beta_ebcr_epidemia}
\end{figure}

\begin{table}
  \caption{\emph{Shrinkage study result.} Coverages (\%) of ML confidence and Bayesian credible intervals for the three sets of parameters shown in \cref{fig::beta_ebcr_true}. \label{tab:result_shrinkage}}
  \centering
  \begin{tabular}{p{4cm}p{2cm}p{2cm}c}
    \toprule
    \textbf{Scenario} & (a) & (b) & (c) \\
    \midrule
    ML fit & 92 & 98 & 95 \\
    \citet{bhatt2020semi} & 91 & 62 & 34 \\
    \bottomrule
  \end{tabular}
\end{table}

\subsection{Analysis of Covid-19 time series of deaths in the United States\label{sec::data}}

We analyzed the effect of a mobility measure on the Covid-19 death data in US states used in \citet{Bonvini2022covid}. The data are from the Delphi repository at Carnegie Mellon University
\url{delphi.cmu.edu}.
The data consist of daily observations, at the state level, on the
number of Covid-19 deaths $Y_t$ (\cref{fig::pred_EY_freqepid_delphi}) and a measure of mobility ``proportion
of full-time work", $A_t$,
(\cref{fig::iv_delphi}), which is the fraction of mobile devices that spent more than 6 hours at a location other than their home during the daytime (SafeGraph’s 
\texttt{full\_time\_work\_prop}).
\V{In all states, mobility varied approximately from 5\% to 11\%, depending on week day, up to mid-March 2020, when mobility dropped to about 3.5\% with much less variability across weekdays, and slowly climbed back up to about 4.5\% by August 2020.}
The time period considered in
the analysis was February 15 2020 to August 1 2020 (168 days).
We focused on the $30$ states that reported over 20 deaths on %
one or more days
at least one day, and we truncated the time series $30$ days prior to $10$ accumulated deaths, as in \citet{bhatt2020semi}. This shaved 10 days of data at the start of the period, leaving 158 days of data in each state for analysis, from February 25 to August 1. 
The death time series showed a strong weekend effect: fewer deaths were reported on Saturdays and Sundays than one would expect from the numbers reported the previous weeks, and these deaths were instead reported mostly on the following Mondays, and Tuesdays and Wednesdays. Because this effect is not accounted for by the model, it adds variability to the analysis that is due to the reporting process rather than to the epidemic process.
For that reason, we pre-processed the data to reduce the effect: for each state, we fitted a nonparametric smooth function to the data together with four additional parameters that estimated the excess Monday, Tuesday and Wednesday effects and deficit saturday effect; we subtracted the estimated effects from all corresponding days and added their sum to all Sundays. \cref{fig::preprocess} shows the original and adjusted data for four states.

\cref{fig::pred_EY_freqepid_delphi} displays the ML mean death fits for the 30 states, assuming NB death data and regression mean specified in \cref{eq::model00,eq::Rt.model0}. While the overall profiles of the fitted means appear reasonable in all states, in several states there are obvious mismatches between observed and fitted onsets of the pandemic. A possible explanation for this mismatch is a mis-specified infection-to-death distribution $\pi$ in \cref{fig::gpi}; perhaps $\pi$ should be state specific to account for differences between states, such as different delays in reporting deaths.

\begin{figure}[t!]
    \centering
    \includegraphics[width=1\textwidth]{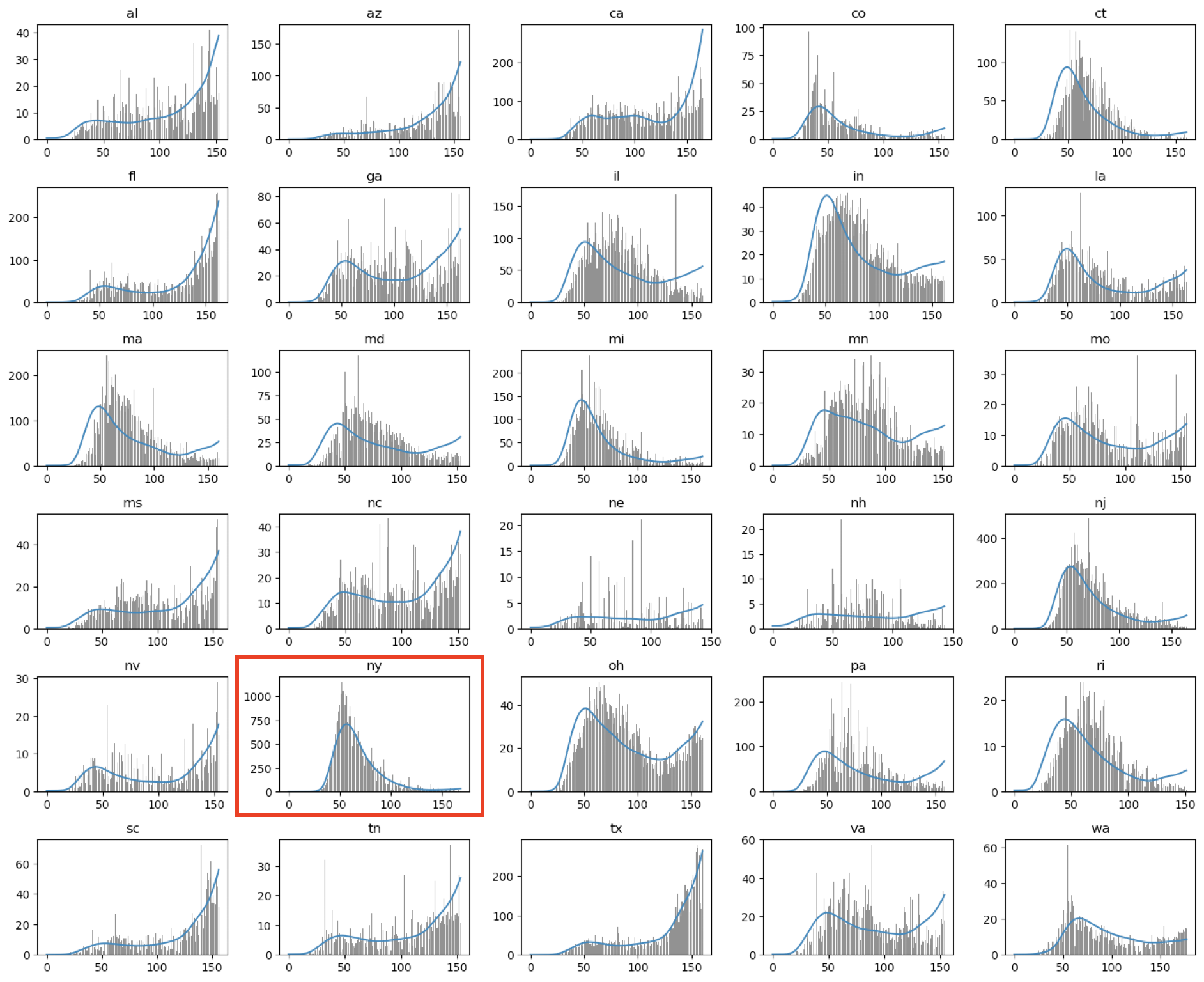}
    \caption{\textbf{ ML fits (blue) to Covid-19 deaths times series in 30 US states} from Feb 15 \V{(day 1)} to August 1st \V{(day 158)} 2020 period . 
    New York is indicated with a box.}\label{fig::pred_EY_freqepid_delphi}
\end{figure}

\cref{fig::diagnostic_ny} shows model and case diagnostics for the ML fit in NY state. In \cref{fig::pred_EY_freqepid_delphi}, there is no apparent onset mismatch in that state but the diagnostics in \cref{fig::diagnostic_ny_vs_t} clearly shows lack of fit at the start and very end of the epidemic, where residuals are not random. A more flexible model for $R_t$ is needed. The other plots don't point to additional issues other than two outliers in (a), only one of which, at $t=140$, has a very modest influence on the fit: the ML estimates of
$(\mu, \beta_0, \beta_1)$ are $(2.08, -7.66, 14.74)$ and $(2.02, -7.82, 15.14)$ with and without that point; the fitted $R_t$ and mean death process with and without that point are shown in Appendix \cref{fig::pred_wo_139}.
Finally, the seemingly influential points at the end of the range in \cref{fig::diagnostic_ny_influence} are due model lack of fit and are not otherwise suspicious.
We also checked the model adequacy in the other states (not shown): apart from the obvious onset mismatches pointed out above, the model and case diagnostics did not raise much concerns.

\begin{figure}[t!]
    \centering
    \begin{minipage}{0.3\textwidth}
        \subfigure[][]{\label{fig::diagnostic_ny_vs_t} 
        \includegraphics[height=0.75\textwidth]{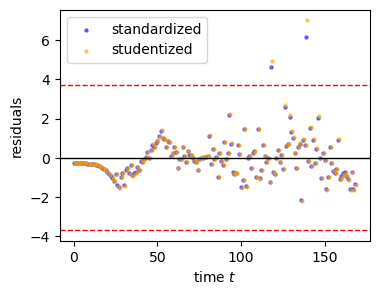}}
    \end{minipage}
    \hspace{-0.01\textwidth}
    \begin{minipage}{0.3\textwidth}
        \subfigure[][]{\label{fig::diagnostic_ny_vs_A} 
        \includegraphics[height=0.75\textwidth]{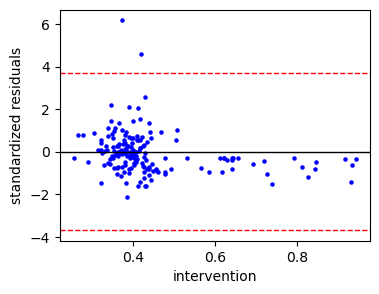}}
    \end{minipage}
    \hspace{-0.01\textwidth}
    \begin{minipage}{0.3\textwidth}
        \subfigure[][]{\label{fig::diagnostic_ny_vs_EYmle} 
        \includegraphics[height=0.75\textwidth]{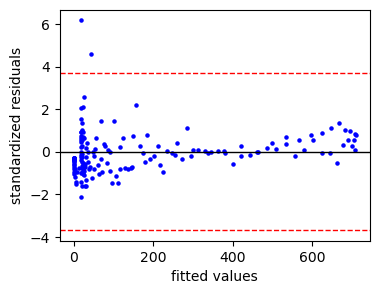}}
    \end{minipage}

    \begin{minipage}{0.3\textwidth}
        \subfigure[][]{\label{fig::diagnostic_ny_qq} 
        \includegraphics[height=0.75\textwidth]{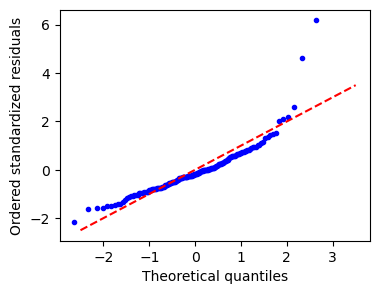}}
    \end{minipage}
    \hspace{-0.01\textwidth}
    \begin{minipage}{0.3\textwidth}
        \subfigure[][]{\label{fig::diagnostic_ny_influence} 
        \includegraphics[height=0.75\textwidth]{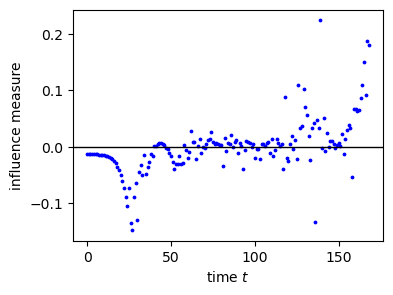}}
    \end{minipage}
    \hspace{-0.01\textwidth}
    \begin{minipage}{0.3\textwidth}
        \subfigure[][]{\label{fig::diagnostic_ny_cook} 
        \includegraphics[height=0.75\textwidth]{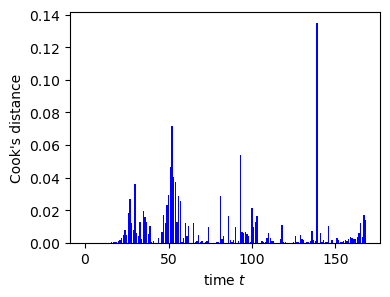}}
    \end{minipage}
    \vspace{-.1in}    
    \caption{\textbf{Model diagnostic plots for NY.} (a) Standardized and studentized residuals versus $t$. (b) Standardized residuals versus $A_t$ and (c) versus $\hat Y_t$. The red dashed lines are the $\alpha/2n$ quantiles of the Gaussian distribution. (d) QQ plot of standadized residuals. (e) Specific influence on $\hat \beta_1$ and (f) Cook's distance versus $t$.}
    \label{fig::diagnostic_ny}
    \medskip
    \centering
    \begin{minipage}{0.318\textwidth}
        \subfigure[][]{\label{fig::beta_freqepid_delphi} 
        \includegraphics[width=\textwidth]{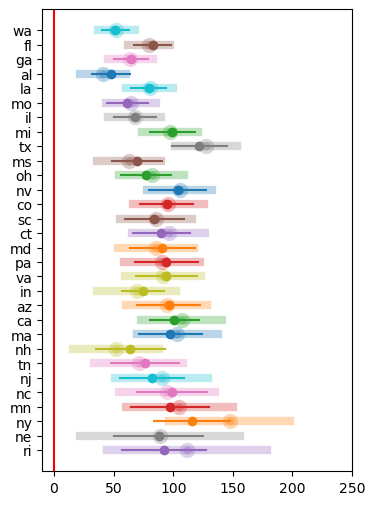}}
    \end{minipage}
    \hspace{-0.01\textwidth}
    \begin{minipage}{0.291\textwidth}
        \subfigure[][]{\label{fig::beta_epidemia_delphi} 
        \includegraphics[width=\textwidth]{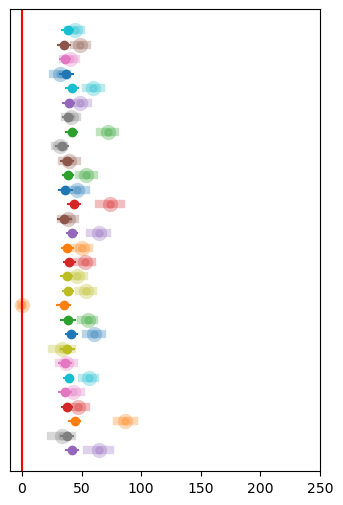}}
    \end{minipage}
    \hspace{-0.01\textwidth}
    \begin{minipage}{0.291\textwidth}
        \subfigure[][]{\label{fig::beta_epininfp_delphi} 
        \includegraphics[width=\textwidth]{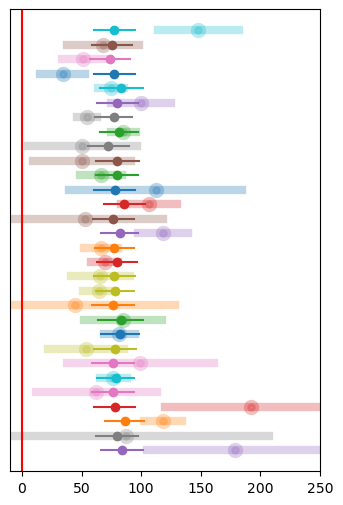}}
    \end{minipage}
    \vspace{-0.1in}
    \caption{\textbf{ Estimates and confidence intervals for $\beta_1$ in \cref{eq::Rt.model0} for 30 states using (a) ML
    \V{and robust empirical Bayes},
    (b) \citet{bhatt2020semi} with their priors and (c) with less informative priors.} The faint thick lines are the estimates and intervals before shrinkage and the dark thin lines after shrinkage. The 30 states are sorted along the y-axis by increasing order of $\Var(\hat \beta_1)$ before shrinkage.}
  \label{fig::beta_delphi}
\end{figure}

\cref{fig::beta_freqepid_delphi} displays the ML estimates of $\beta_1$ for the 30 states, fitted separately and together using the \V{empirical Bayes} shrinkage method in \cref{sec::shrinkage}. 
\cref{fig::beta_epidemia_delphi} displays the Bayesian estimates fitted separately and together, as in \citet{bhatt2020semi}.
Note that we scaled and shifted $A_t$ so that the values would be mostly in the $[0,1]$ range, to match the binary interventions in \citet{bhatt2020semi}, so we could use their prior distributions;
see \cref{fig::iv_A_delphi}.
\V{However, the estimates of $\beta_1$ in
\cref{fig::beta_delphi} are displayed on the original scale.}
The two estimation methods in
\cref{fig::beta_freqepid_delphi,fig::beta_epidemia_delphi}
produce strikingly different inferences: 
the Bayesian estimates of $\beta_1$ take values in a narrow range around 50 and they have very narrow credible intervals,
while the ML estimates have a much larger spread around 100 and
their confidence intervals are wider.
Clearly, the priors are inappropriate for this data; they are centered on the wrong values and their variances are too small.
\cref{fig::beta_epininfp_delphi} shows the Bayesian estimates using less informative priors, specifically \cref{eq::hier} with hyper-priors\footnote{For full disclosure, we originally used $N(0,8)$ for $\beta_0^*$ and $\beta_1^*$, which gave estimates of $\beta_1$ around 600; these were too different from the MLEs for comfort. We suspect that overly uninformative hyper priors yielded improper posteriors. We subsequently tweaked the hyper priors so that the shrunk Bayesian estimates would be closer to the MLEs.}
\begin{equation*}
\begin{aligned}
    e^{\mu^*} & \sim \mathrm{Expon}(0.03), &  \\
    \beta_0^* & \sim N(0, 4), & \sigma_0 \sim \mathrm{Gamma}(1, 8),\\
    \beta_1^* & \sim N(0, 4), & \sigma_1 \sim \mathrm{Gamma}(1, 8).
\end{aligned}
\end{equation*}
Now most Bayesian estimates of $\beta_1$ take values around 100, which is reassuring since they are closer to the ML estimates. But, ultimately, we have more confidence in the ML estimates given the simulation results in \cref{tab:result_shrinkage}.
Note also that, in some states, the shrunk Bayesian estimates are drastically different from the marginal Bayesian estimates, and the credible intervals for the shrunk estimates are rather narrow compared to their un-shrunk counterparts, which may point to some deficiencies of the Bayesian hierarchical model.

\V{Because the logistic function in \cref{eq::Rt.model0} is
non-linear, estimates of $\beta_1$ are useful for comparing qualitatively the effect of mobility across states -- e.g. decreasing mobility has a large effect in NY and TX and a small effect in AL and WA in \cref{fig::beta_delphi} -- but their specific values are not particularly interpretable. Changes in $R_t$ are more meaningful; for example, \cref{fig::pred_wo_139} shows that in NY state, $R_t$ decreased from 6 to just below 1 when mobility dropped suddently mid-March 2020. This drop is in the same ballpark as the drops observed in European countries when government mandated mobility interventions were implemented \citep{epidemia}. }

Lastly, \cref{fig::pred_EY_freqepid_ny} shows future death predictions and prediction intervals (\cref{sec::predictions}) for NY state under two hypothetical trajectory mobility $A_t$, based on the shrunk estimates. The Bayesian and ML predictions are strikingly different. We have much more confidence in the latter, although we do not have perfect confidence since the model exhibits some lack of fit, as evidenced in \cref{fig::diagnostic_ny}.

\begin{figure}[t!]
    \centering
    \begin{minipage}{0.3\textwidth}
        \subfigure[][]{\label{fig::Ae_ny} 
        \includegraphics[height=0.75\textwidth]{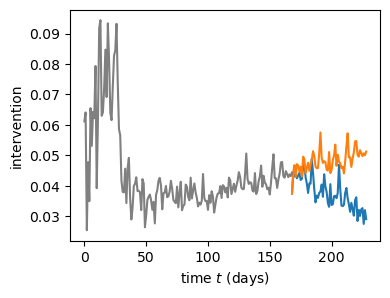}}
    \end{minipage}

    \centering
    \begin{minipage}{0.3\textwidth}
        \subfigure[][]{\label{fig::pred_EY_freqepid_ny_dec} 
        \includegraphics[height=0.75\textwidth]{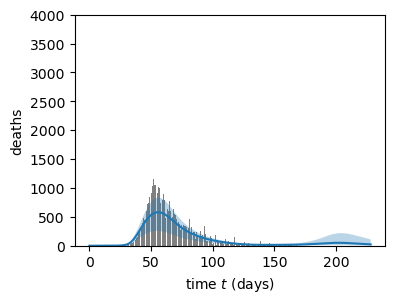}}
    \end{minipage}
    \hspace{-0.01\textwidth}
    \begin{minipage}{0.3\textwidth}
        \subfigure[][]{\label{fig::pred_EY_epidemia_ny_dec} 
        \includegraphics[height=0.75\textwidth]{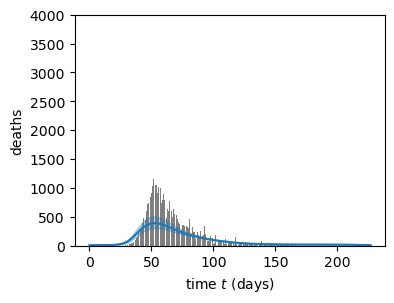}}
    \end{minipage}
    \hspace{-0.01\textwidth}
    \begin{minipage}{0.3\textwidth}
        \subfigure[][]{\label{fig::pred_EY_epininfp_ny_dec} 
        \includegraphics[height=0.75\textwidth]{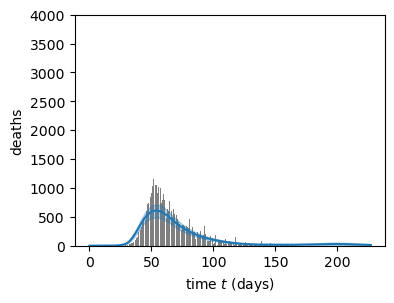}}
    \end{minipage}

    \centering
    \begin{minipage}{0.3\textwidth}
        \subfigure[][]{\label{fig::pred_EY_freqepid_ny_inc} 
        \includegraphics[height=0.75\textwidth]{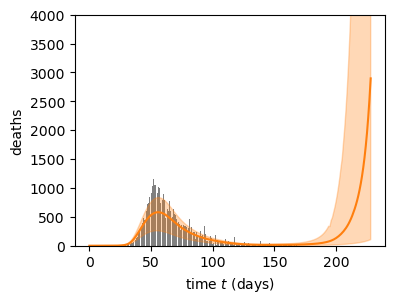}}
    \end{minipage}
    \hspace{-0.01\textwidth}
    \begin{minipage}{0.3\textwidth}
        \subfigure[][]{\label{fig::pred_EY_epidemia_ny_inc} 
        \includegraphics[height=0.75\textwidth]{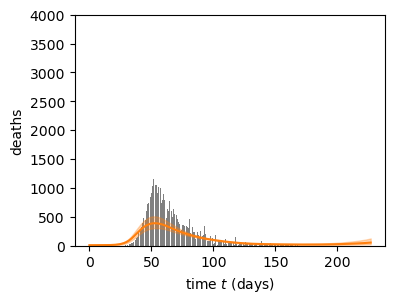}}
    \end{minipage}
    \hspace{-0.01\textwidth}
    \begin{minipage}{0.3\textwidth}
        \subfigure[][]{\label{fig::pred_EY_epininfp_ny_inc} 
        \includegraphics[height=0.75\textwidth]{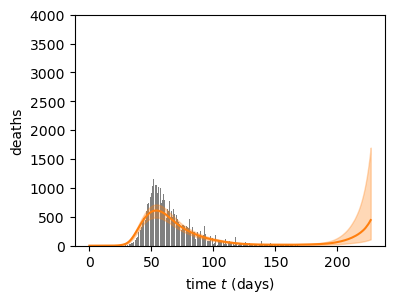}}
    \end{minipage}
     
    \caption{\textbf{ Counterfactual predictions for NY state.} (a) Observed intervention ($t \leq 168$) and two future scenarios ($t > 168$) in blue and orange. (b,e) Counterfactual mean deaths $\Exp[Y_t]$ with $95\%$ global confidence bands based on the ML shrinkage estimator (\cref{sec::shrinkage}). (c,f) Corresponding Bayesian predictions using the informative priors in \cref{eq::hier,eq::hyper}, and (d,g) using the less informative priors.}
    \label{fig::pred_EY_freqepid_ny}
\end{figure}

\section{Discussion \label{sec::discussion}}

We have discussed statistical inference
for a semi-mechanistic epidemic model
using frequentist methods.
The advantage of using
frequentist methods
is that there is no need to specify prior distributions
and the confidence intervals
have valid coverage properties.
However, the disadvantage
is that these coverage properties are asymptotic in nature
in the sense that they hold as the total number of time
points $T$ increases.

If the user has informative priors that they want to
include then it is possible to do so while preserving
frequentist coverage.
For example, inverting a test
based on the integrated
likelihood allows one to include
a prior and still have frequentist coverage.

We made simplifying assumptions that are not always reasonable.
We assumed $\alpha_t = \alpha_0$, meaning that the probability that an infected individual produces an outcome (e.g. dies) remains constant over time, which is not reasonable unless the time series are short enough that nothing in the environment changes. Otherwise, changes in virus mutation, medical treatment, season, etc. will affect $\alpha_t$. The same argument applies to $R_t$. 
Our model does allow $R_t$ to vary with $t$
through $A_t$
but a more elaborate model would allow
$K$ to also vary with $t$
since, for example,
not all variants are as deadly, and part of the population will become immune for a while. 
One could also use a more flexible model for $R_t$,
for example by replacing $A_t$ with a spline basis evaluated at $A_t$.
Another improvement would be to let $\pi$ vary by location.
Furthermore, one could let $g$ vary with time.
All of these generalizations are possible
in principle but one has to balance with desire
for more complex (and realistic) models
with the fact that there is limited information.
Informative priors can help but at the cost
of losing frequentist coverage properties
which is one of the goals of this paper.

The code for the methods in this paper
is freely available
in our Python package \texttt{freqepid}.
Code vignettes we used to generate the results in \cref{sec::examples} are available at \url{github.com/HeejongBong/freqepid}.
Inspiration for this project
came from the R package
\texttt{Epidemia}.
Currently, that package implements a wider range of models
compared to \texttt{freqepid}
but, in principle,
the methods there can be extended to provide the same utility as \texttt{Epidemia}.
We have found that the models in
\texttt{freqepid} often work well
even if the generating processes do not match the model.

\section*{Acknowledgement}
The authors thanks the Delphi group
(\url{delphi.cmu.edu})
for providing funding the first author
on this project.


\bibliographystyle{apalike}
\bibliography{refs}

\appendix\section{Appendix: Block Coordinate Descent Algorithm}
\label{sec::append.algos}
\label{sec::algorithm}

\V{
In the main paper, $A_t$ denotes 
a vector of interventions.
For notation brevity, in all the appendices we will append $1$ at the start of $A_t$
so that the expression $\beta_0 + \beta_1^\top A_t$
in the main text can simply be written 
$\beta^\top A_t$ in the proofs.
}

Given observed data $\overline Y_T$, we fit our model, (\cref{eq::model00}) by maximizing the log likelihood function
\begin{equation*}
\begin{aligned}
  \ell(\theta) 
    & \equiv \sum_{t=1}^T \log p(Y_t;\theta)
    = \sum_{t=1}^T \ell_t(\theta). \\
\end{aligned}
\end{equation*} Because the MLE for $\theta$ is not in a closed form, we use the coordinate descent algorithm to find the maximizer.
For the first iteration, let $\hat\beta^{(0)}$ and $\hat\mu^{(0)}$ be the starting values given below. For iteration $i > 1$, let $\hat\beta^{(i-1)}$ and $\hat\mu^{(i-1)}$ be the outcomes of the previous iteration. We set $\hat{\nu}^{(i)}$ by the maximizer of the likelihood function given $\beta$ and $\mu$ are fixed at $\hat\beta^{(i-1)}$ and $\hat\mu^{(i-1)}$. For the Normal distribution, the maximizer has a closed form:
\begin{equation} \label{eq::cd_G_sigma}
    \hat{\nu}^{(i)} = \frac{1}{\sqrt{T}} \norm{\overline Y_T - \hat\mu^{(i-1)} (\Pi \Omega(\hat\beta^{(i)}) + \Pi_0) \mathbf{1}}_2.
\end{equation}
For the NB
distribution, the maximizer $\hat{\nu}^{(i)}$ is obtained by numerically solving
\begin{equation*}
    \hat{\nu}^{(i)} = {\arg\min}_{\nu}  \ell(\nu, \hat{\mu}^{(i-1)}, \hat{\beta}^{(i-1)}).
\end{equation*}
Next, 
we use Newton's method to perform a coordinate descent step for $\beta$ and $\mu$.
The steps are:
\begin{equation*}
    (\hat{\beta}^{(i)}, \hat\mu^{(i)}) 
     = (\hat{\beta}^{(i-1)}, \hat{\mu}^{(i-1)}) 
    + \eta ~[\nabla_{(\beta,\mu)}^2 \ell (\hat\beta^{(i-1)}, \hat\mu^{(i-1)}, \hat\nu^{(i)})]^{-1}
    ~\nabla_{(\beta,\mu)} \ell(\hat\beta^{(i-1)}, \hat\mu^{(i-1)}, \hat\nu^{(i)})
\end{equation*}
where $\ell(\beta, \mu, \nu)$ is the log-likelihood function given data $\overline Y_T$, and $\eta$ is a step size hyperparameter. We iterate the algorithm until the log-likelihood converges at a threshold hyperparameter.

\paragraph{Starting Values}
The coordinate descent algorithm and Newton's method require
starting values $\hat\beta^{(0)}$ and $\hat\mu^{(0)}$.
Our approach is as follows.
Start by approximating $\pi$ with a point mass at its mean $\overline{\pi}$
so that $Y_t = \alpha I_{t-\overline{\pi}}$.
Then an initial estimate of $\overline I_{0}$ is $(Y_{-k+\overline\pi+1}/\alpha, \dots, Y_{\overline{\pi}}/\alpha)$. We obtain $\hat\mu^{(0)} = \log(\frac{1}{\alpha k} \sum_{t=-k+\overline{\pi}+1}^{ \overline{\pi}} Y_t)$.
We then have from \cref{eq::Rt.model0}
$$
Y_{t+\overline{\pi}} = K (1 + e^{-\beta^\top A_t})^{-1} \sum_{s<t} Y_{s+\overline{\pi}} g_{t-s}.
$$
Now approximate 
$(1+e^{- \beta^\top A_t})^{-1} \approx \frac{1}{2} + \frac{1}{4} \beta^\top A_t$.
This leads to the linear regression
$$
Y_{t+\overline{\pi}} = \frac{1}{2} K \sum_{s<t} Y_{s+\overline{\pi}} g_{t-s} +
\frac{1}{4} K \beta^\top A_t \sum_{s<t} Y_{s+\overline{\pi}} g_{t-s},
$$
and we obtain the least square estimates $\hat\beta^{(0)}$. We note that the starting value for $\nu$ is not required for the estimation algorithm.

\section{Appendix: Proofs}
\label{sec::append.proofs}

\paragraph{\underline{\bf Proof of Proposition \ref{thm::identifiability}}}
 Conditional on $\overline I_T$, the $Y_t$'s are
 mutually independent and $\overline I_T \equiv \overline I_T(\beta, \mu)$ is deterministic given parameters $\beta$ and $\mu$. 
 Hence, the conditional distribution of $Y_t$ given $\overline{A}_t$ is 
 equal to the conditional distribution of $Y_t$ given $(\overline I_t, \overline Y_{t-1}, \overline{A}_t)$ (and hence given $(\overline I_t, \overline Y_{t-1})$,) 
 which we take to be Gaussian or NB with nuisance parameter $\nu$ and mean $m_t(\beta, \mu) \equiv \Exp_\theta[Y_t | \overline A_t]$,
 where $\Exp_\theta$ denotes the expectation under model \eqref{eq::model00} with parameter $\theta$.
 Because $\overline I_0 = e^\mu \mathbf{1}$, $m_1(\beta, \mu)$ is given by $\alpha e^\mu \sum_{t=1}^{T_0} \pi_t$,
 and because $\alpha \sum_{s=1}^{T_0} \pi_t > 0$, 
 $\mu$ is identified by the marginal mean of $Y_1$, 
 that is $\mu = \log \left(\frac{\Exp_\theta[Y_1|A_1]}{\alpha \sum_{s=1}^{T_0} \pi_t} \right)$. Similarly, because $\nu$ is an identified parameter of the Gaussian and NB distributions, it is identified by the marginal distribution of $Y_1$.

 Now it is sufficient to show that $\beta$ is identified given that $\mu$ and $\nu$ are fixed. We proceed with a proof by contradiction.
 %
 Suppose that the marginal means of $\overline Y_T$ under $\beta$ and $\beta'$ are the same, that is $\overline m_T (\beta, \mu) = \overline m_T (\beta', \mu)$. This implies that $\overline I_{T-T_0}(\beta, \mu) = \overline I_{T-T_0} (\beta', \mu)$. 
 To prove this, suppose that $\exists t \in \{2, \dots, T-T_0\}$ such that $\overline I_{t-1}(\beta, \mu) = \overline I_{t-1} (\beta', \mu)$ and $I_t(\beta,\mu) \neq I_t(\beta',\mu)$. Defining $\tau \equiv \min\{t: g_t > 0\}$, we have $\tau \leq T_0$, and 
 \begin{equation*}
     m_{t+\tau}(\beta, \mu) = \alpha \sum_{s=0}^{t+T_0-1} \pi_{s+\tau} I_{t-s}(\beta,\mu).
 \end{equation*}
 Because $\overline I_{t-1}(\beta, \mu) = \overline I_{t-1}(\beta', \mu)$, $\sum_{s=1}^{t+T_0-1} \pi_{s+\tau} I_{t-s}(\beta,\mu) = \sum_{s=1}^{t+T_0-1} \pi_{s+\tau} I_{t-s}(\beta',\mu)$. As a result, $m_{t+\tau}(\beta, \mu) = m_{t+\tau}(\beta', \mu)$ implies that $I_t(\beta,\mu) = I_t(\beta',\mu)$. This contradicts our premise that $\overline I_{T-T_0}(\beta, \mu) \neq \overline I_{T-T_0}(\beta', \mu)$. The identity in infection $\overline I_{T-T_0}$ implies $\overline R_{T-T_0}(\beta) = \overline R_{T-T_0}(\beta')$. In other words, $(\beta-\beta') A_t = 0$ for all $t \leq T-T_0$, from which we conclude that  $\beta = \beta'$.


\subsection{Lemmas and Proofs for Section~\ref{sec::CLT}} 
\label{app::CLT}

\V{We need the following lemma for the proof of \cref{thm::asymp_norm}.}
\begin{lemma}\label{lemma::compact}
    Suppose that Assumption~\ref{assmp::eigenvalue} holds. For the NB distribution, we further assume that $r_{\min} > 1$. Then $\Theta$ as defined in Definition~\ref{def::Theta} is a compact set with probability $1$. Moreover, there exist $\mathfrak{F}(\overline{A}_T)$-measurable random variables $R_{\min}$, $R_{\max}$, $\nu_{\min}, \nu_{\max} > 0$ such that, for any $t \geq 1$
    \begin{equation*}
        (\beta, \mu, \nu) \in \Theta \implies 
        R_{\min} \leq R(\overline A_{t},\beta) \leq R_{\max}
        \textand
        \nu_{\min} \leq \nu \leq \nu_{\max},
    \end{equation*}
    with probability $1$.
\end{lemma}

\begin{proof}
    For $t \geq \tau$, $\Exp_\theta[I_t|\overline A_t] = R(\overline A_{t},\beta) \sum_{s=1}^{\tau-1} \pi_s I_{t-s}$. Because $I_{\min} \leq I_{t-s} \leq I_{\max}$ for $s \in \{0, \dots, \tau-1\}$, 
    \begin{equation*}
        R(\overline A_{t}, \beta) I_{\min} \leq \Exp_\theta[I_t|\overline A_t] = I_t \leq I_{\max}.
    \end{equation*}
    This implies for $t \geq \tau$, $R(\overline A_{t}, \beta) \leq \frac{I_{\max}}{I_{\min}}$ and similarly $R(\overline A_{t}, \beta) \geq \frac{I_{\min}}{I_{\max}}$.    
    According to \cref{eq::Rt.model0}, there exists $M > 0$ such that $\frac{I_{\min}}{I_{\max}} \leq R(\overline A_{t}, \beta) \leq \frac{I_{\max}}{I_{\min}} \implies \abs{\beta^\top A_t} < M$. As a result,  $\norm{\beta}_2^2 \leq \frac{1}{\lambda_{\min,A} \cdot (T-\tau)} \norm{\beta^\top (A_{\tau+1}, \dots, A_{T})}_2^2 \leq \frac{M^2}{\lambda_{\min,A}}$. Moreover, $\abs{\beta^\top A_t} \leq \norm{\beta}_2\norm{A_t}_2 \leq \frac{M A_{\max}}{\sqrt{\lambda_{\min,A}}}$, which implies by \cref{eq::Rt.model0} that $R_{\min} \leq R(\overline A_{t},\beta) \leq R_{\max}$ for any $t \geq 1$, where $R_{\min}$ and $R_{\max}$ are constants depending on $M$, $\lambda_{\min,A}$ and $A_{\max}$.
    
    Next, according to \cref{eq::model00},
    \begin{equation*}
    \begin{aligned}
        R_{\min} e^\mu \sum_{t=1}^{T_0} \pi_t 
        & \leq \Exp_\theta[I_1| \overline A_1] =
         R(\overline A_{1}, \beta) e^\mu \sum_{t=1}^{T_0} \pi_t 
        \leq R_{\max} e^\mu \sum_{t=1}^{T_0} \pi_t.
    \end{aligned}
    \end{equation*} 
    Because $I_{\min} \leq \Exp_\theta[I_1|\overline A_1] \leq I_{\max}$,  $\mu$ is bounded by 
    \begin{equation*}
        \log\left(\frac{I_{\min}}{R_{\max} \sum_{t=1}^{T_0} \pi_t}\right) \leq \mu \leq \log\left(\frac{I_{\max}}{R_{\min} \sum_{t=1}^{T_0} \pi_t}\right).
    \end{equation*}

    Finally, for the Normal distribution, $r_{\min} \leq \frac{T \nu^2}{\sum_t \Exp_\theta[Y_t |\overline A_t ]}  \leq r_{\max}$. Because 
    \begin{equation*}
        \min\{e^\mu, I_{\min}\} \sum_{t=1}^{T_0} \pi_t \leq \Exp_\theta[Y_t |\overline A_t] \leq e^\mu + I_{\max},
    \end{equation*} 
    $\nu$ is bounded away from $0$ and $\infty$. For the NB distribution, 
    \begin{equation*}
        r_{\min} \leq 1 + \frac{1}{\nu} \frac{\sum_t \Exp_\theta[Y_t |\overline A_t]^2}{\sum_t \Exp_\theta[Y_t |\overline A_t ] } \leq r_{\max}.
    \end{equation*}
    Because $r_{\min} > 1$, $\nu$ is again bounded away from $0$ and $\infty$. In sum, $\Theta$ is bounded almost surely.
    
    Furthermore, the functions 
    \begin{equation*}
        (\beta,\mu)  \mapsto \Exp_\theta[I_t|\overline A_t] \, \text{ and } \,
        \theta  \mapsto \Var_\theta[Y_t|\overline A_t]
    \end{equation*}
    are continuous. As an intersection of the preimages of 
    closed sets by the three continuous functions, $\Theta$ is closed almost surely. 
    It follows that
    $\Theta$ is compact almost surely.
\end{proof}

\paragraph{\underline{\bf Proof of \cref{thm::asymp_norm}}}

\V{Our proof is based on the asymptotic property of nonlinear time-series models established by \citet{potscher1997dynamic}. In particular, we use the consistency (Theorem 7.1) and asymptotic normality (Theorem 11.2 (b)) of general M-estimators, defined as follows under their own notations:
\begin{equation} \label{eq::M_estimator}
    \hat\beta_n \equiv {\arg\min}_{\beta} ~\vartheta_n\left(n^{-1} \sum_{t=1}^n q_t(\mathbf{z}_t, \hat\tau_n, \beta), \hat\tau_n, \beta\right),
\end{equation}
where $\vartheta_n$ is the risk function, $q_t$ is the contribution of data $\mathbf{z}_t$ at time $t$ to the risk, and $\hat\tau_n$ is an estimator of (potential) nuisance parameter $\tau$. 
To help readers cross-reference between the book and our proof, we explain which component in our MLE setting corresponds to which in \cref{eq::M_estimator} as follows:
\begin{itemize}
    \item the number of data points $n$ is $T$ in our setting,
    \item the risk function $\vartheta_n$ is given by $\vartheta_n(c, \tau, \beta) = c$ in our setting,
    \item the risk contribution $q_t$ at time $t$ is the negative log-likelihood contribution $\ell_t$ in our setting, 
    that is $\ell_t(\theta) = \log p_{Y_t}(Y_t| \overline{A}_t, \theta)$, where $p_{Y_t}(\cdot | \overline{A}_t, \theta)$ is the pdf of $Y_t$ conditional on $\overline{A}_t$ under our model with parameter $\theta$,
    \item the data point $z_t$ at each time $t$ is $Y_t$ in our setting, 
    \item the parameter $\beta$ and its estimate $\hat\beta_n$ correspond to $\theta$ and $\hat\theta$ in our setting, and
    \item the nuisance parameter $\tau$ and its estimate $\hat\tau_n$ can be disregarded in our setting.
\end{itemize}
}

We first prove the consistency of $\hat\theta$. \V{To make the dependence of $\hat\theta$ on $T$ clear, we use the notation $\hat\theta_T$ for $\hat\theta$. The consistency theorem (Theorem 7.1 of \citealp{potscher1997dynamic}) was established under assumptions 7.1, 7.2 therein and the assumption about identifiably unique sequence of minimizers, which consist the common assumption set in the MLE theory. Below we show that all of their assumptions are met almost surely conditional on $\overline{A}_T$ under our setting with given assumptions. Their assumption 7.1 (a) about the compact parameter space assumption is met almost surely by our Definition~\ref{def::Theta} and Assumption~\ref{assmp::eigenvalue} as shown in Lemma~\ref{lemma::compact}, 7.1 (b) about the equicontinuity of $\vartheta_n$ by $\vartheta_n(c, \tau, \beta) = c$, 7.1 (d) about the mixingness of data points by our Assumption~\ref{assmp::mixing}, and 7.1 (e) by the absence of nuisance parameter in our model. The tightness in their assumption 7.2 follows the second part of our Assumption~\ref{assmp::bounded}. The assumption about identifiably unique sequence of minimizers follows our Assumption~\ref{assmp::identified_minimizer}.}
Now it suffices to show the almost sure validity of their assumption 7.1 (c) and the equicontinuity in their assumption 7.2:
\begin{enumerate}
    \item $\sup_T \frac{1}{T} \sum_{t=1}^T \Exp[\sup_{\theta \in \Theta} \abs{\log p_{Y_t}(Y_t; \theta)}^{1+\zeta} | \overline{A}_t] < \infty$ for some $\zeta > 0$; and
    \medskip
    \item for any fixed $y \in \nats_0$, $\{ \log p_{Y_t}(y| \overline{A}_t, \theta): t \in (0,\infty)\}$ is equicontinuous on $\Theta$.
\end{enumerate}

\V{This requires essential boundedness of reproduction numbers $R(\bar{A}_t, \beta)$ and scale parameter $\nu$, which we have shown in Lemma~\ref{lemma::compact}.}
For the Normal distribution, $\log p_{Y_t}(y| \overline{A}_t, \theta) = - \log \nu - \frac{(y - m_t(\beta,\mu))^2}{2 \nu^2}$, where $m_t(\beta, \mu) \equiv \Exp_\theta[Y_t | \overline A_{t}]$. Because $\Theta$ is compact almost surely, and $(\beta, \mu, \nu) \in \Theta$ implies that $m_t(\beta, \mu)$'s are uniformly bounded away from $0$ and $\infty$, taking $\zeta = \frac{\gamma}{2}-1$, 
\begin{equation*}
    \sup_{\theta \in \Theta} \Exp[d_t(Y_t)^{1+\zeta}] \leq 2^{1+\zeta} \sup_{\theta \in \Theta} \left( \abs{\log\nu}^{1+\zeta} + \frac{\Exp[Y_t^\gamma] + (m_t(\beta, \mu))^\gamma}{\nu^\gamma} \right) < \infty,
\end{equation*}
almost surely.
For the NB distribution, $\log p_{Y_t}(y| \overline{A}_t, \theta) = \sum_{k=1}^{y} \log\left(1 + \frac{\nu-1}{k}\right) + \nu \log \nu + y \log m_t(\beta, \mu) - (\nu+y) \log(\nu + m_t(\beta, \mu))$ for $y \in \nats_0$. Because $\abs*{\sum_{k=1}^{Y_t} \log\left(1 + \frac{\nu-1}{k}\right)} \leq \abs{Y_t \log \nu}$, where $\zeta = \gamma-1$
\begin{equation*}
\begin{aligned}
    \sup_{\theta \in \Theta} \Exp[d_t(Y_t)^{1+\zeta}]
    & \leq 3^{1+\zeta} \sup_{\theta \in \Theta} \left[
        \left(\nu \log \frac{\nu + m_t(\beta,\mu)}{\nu} \right)^{1+\zeta} 
        + \Exp[Y_t^\gamma] \left( \abs{\log \nu}^\gamma + \left(\log \frac{\nu + m_t(\beta,\mu)}{m_t(\beta,\mu)}\right)^\gamma \right)
    \right] \\
    & < \infty,
\end{aligned}
\end{equation*}
almost surely. In sum, for the both cases, their assumption 7.1 (c) holds.

For the equicontinuity in their assumption 7.2, we note that $\log p_{Y_t}(y | \overline{A}_t, \theta) = f(\nu, m_t(\beta, \mu); y, \overline{A}_t)$ for a fixed function $f$ in both cases of the Normal and NB distributions. So it is sufficient to show that $m_t(\beta, \mu)$ is equicontinuous on $\Theta$. The equicontinuity of $m_t$ is implied almost surely by our Definition~\ref{def::Theta} of $\Theta$ because
\begin{equation*}
    \sup_t \sup_{\theta \in \Theta} \norm{\nabla_\theta \Exp_\theta[I_t | \overline A_{t}]}_\infty < \infty.
\end{equation*}
\V{In sum, $\hat\theta_T$ is consistent to $\theta^*$ conditional on $\overline{A}_T$ almost surely: for any $\delta > 0$,
\begin{equation*}
    \lim_{T \rightarrow \infty} \Pr[\norm{\hat\theta_T - \theta^*} > \delta | \overline{A}_T] = 0,
\end{equation*}
almost surely. By Fubini's theorem,
\begin{equation*}
\begin{aligned}
    \lim_{T \rightarrow \infty} \Pr[\norm{\hat\theta_T - \theta^*} > \delta]
    & = \lim_{T \rightarrow \infty} \int \Pr[\norm{\hat\theta_T - \theta^*} > \delta | \overline{A}_T] ~dP_{\bar{A}} \\
    & = \int \lim_{T \rightarrow \infty} \int \Pr[\norm{\hat\theta_T - \theta^*} > \delta | \overline{A}_T] ~dP_{\bar{A}} \\
    & = \int 0 ~dP_{\bar{A}}
    = 0.
\end{aligned}
\end{equation*}}
This concludes our proof for the consistency of $\hat\theta_T$.

Now we move on to the asymptotic normality of $\hat\theta_T$. The asymptotic normality (Theorem 11.2 (b) of \citealp{potscher1997dynamic}) was established under assumptions 11.1, 11.2, 11.3 and 11.5 therein. We again show that all of their assumptions are met almost surely conditional on $\overline{A}_T$ under our setting with given assumptions. Their assumption 11.1 (a) is held almost surely by our Definition~\ref{def::Theta} of $\Theta \in \reals^{d+2}$ where $d$ is the dimension of covariates $A_t$, 11.1 (b) by the twice differentiability of the negative log-likelihood contribution $\ell_t$ with respect to $\theta$, 11.1 (c) by our definition of MLE $\hat\theta_T$ as the minimizer of the risk function, 11.1 (d) by the previously shown consistency of $\hat\theta_T$, 11.1(e) by the mixing condition of $Y_t$ in our Assumption~\ref{assmp::mixing}, and 11.1(f) by the second part of our Assumption~\ref{assmp::bounded}. Their assumption 11.3 (a) follows the definition of $\theta^*_T$ as the minimizer of $\Exp[\sum_{t=1}^T \ell_t(\theta) | \overline{A}_T]$, and 11.3 (c, d) follows our Assumption~\ref{assmp::identified_minimizer}. Their assumption 11.3 (b) is satisfied trivially by the absence of nuisance parameter in our setting. 
\V{Their other assumptions 11.2 and 11.5, namely,
\begin{enumerate}
    \item $\sup_T \frac{1}{T} \sum_{t=1}^T \Exp[\sup_{\theta \in \Theta} \abs{\nabla_\theta^2 \ell_t(\theta)}^{1+\zeta} | \overline{A}_t] < \infty$ for some $\zeta > 0$; and
    \medskip
    \item $\sup_T \sup_t \Exp[\abs{\nabla_\theta \ell_t(\theta^*_T)}^\gamma | \overline{A}_t] < \infty$,
\end{enumerate}
are weaker versions of their assumptions 13.1 and 11.5*, respectively, which are used in their consistency result of the HAC variance estimator. We relegate their almost sure validaitions to the proof of Proposition~\ref{thm::sandwich}.}

\V{As a result, $\hat\theta_T$ satisfies conditional asymptotic normality almost surely: for some collection $\mathcal{U}$ of subsets in $\reals^d$, with probability $1$, 
\begin{equation*}
    \lim_{T \rightarrow \infty} \Pr[T^{1/2} \Upsilon^{-1/2} (\hat\theta - \theta^*) \in U | \overline{A}_T]
    = \Pr[Z_d \in U], ~\forall U \in \mathcal{U}
\end{equation*}
where $Z_d$ is an independent $d$-dimensional standard Normal random variable. By Fubini's theorem,
\begin{equation*}
\begin{aligned}
    \lim_{T \rightarrow \infty} \Pr[T^{1/2} \Upsilon^{-1/2} (\hat\theta - \theta^*) \in U]
    & = \lim_{T \rightarrow \infty} \int \Pr[T^{1/2} \Upsilon^{-1/2} (\hat\theta - \theta^*) \in U | \overline{A}_T] ~dP_{\bar{A}} \\
    & = \int \lim_{T \rightarrow \infty} \int \Pr[T^{1/2} \Upsilon^{-1/2} (\hat\theta - \theta^*) \in U | \overline{A}_T] ~dP_{\bar{A}} \\
    & = \int \Pr[Z_d \in U] ~dP_{\bar{A}}
    = \Pr[Z_d \in U].
\end{aligned}
\end{equation*}
This concludes the proof of asymptotic normality.}

\paragraph{\underline{\bf Proof of Proposition~\ref{thm::sandwich}}}
We use the consistency result of the HAC variance estimator established by \citet{potscher1997dynamic}. The corresponding theorem (Theorem 13.1 (b) therein) uses assumptions 11.1, 11.2, 11.3*, 11.5* and 13.1 therein. We already checked their assumptions 11.1, 11.2 and 11.3 in the proof of \cref{thm::asymp_norm}. 
\V{Their assumption 11.3* is stronger than 11.3 by necessitating $\Exp[\nabla_\theta\ell_t(\theta^*_T)] = 0$ at each $t$, rather than $\sum_t\Exp[\nabla_\theta\ell_{t=1}^T(\theta^*_T)] = 0$. This essentially means that the true mean $\Exp[Y_t|A_t]$ is correctly specified by the model mean $\Exp_{\theta}[Y_t|A_t]$ at $\theta = \theta^*$. That is why our Proposition 5.9 requires the additional assumption. Now it sufficies to validate their assumptions 11.5* and 13.1, namely,
\begin{enumerate}
    \item[(i)] $\nabla_\theta \ell_t(\theta^*)$ is near epoch dependent of size $-2(r-1)/(r-2)$ on an $\alpha$-mixing baseline process with coefficient of size $-2r/(r-2)$;
    \item[(ii)] $\sup_T \sup_t \Exp[\abs{\nabla_\theta \ell_t(\theta^*_T)}^{2\gamma} | \overline{A}_t] < \infty$;
    \medskip
    \item[(iii)] $\sup_T \frac{1}{T} \sum_{t=1}^T \Exp[\sup_{\theta \in \Theta} \abs{\nabla_\theta \ell_t(\theta)}^{2} | \overline{A}_t] < \infty$; and
    \medskip
    \item[(iv)] $\sup_T \frac{1}{T} \sum_{t=1}^T \Exp[\sup_{\theta \in \Theta} \abs{\nabla_\theta^2 \ell_t(\theta)}^{2} | \overline{A}_t] < \infty$,
\end{enumerate} 
in almost sure sense.
}

\V{Part (i) follows our Assumption~\ref{assmp::mixing} about the mixingness of the data points.
For parts (ii) -- (iv), it is sufficient to show that $\nabla_{(\nu, m_t)} \ell_t(\theta^*_T)$ has a finite $2\gamma$-moment and that $\sup_{\theta \in \Theta} \nabla_{(\nu, m_t)} \ell_t(\theta)$ and $\sup_{\theta \in \Theta} \nabla_{(\nu, m_t)}^2 \ell_t(\theta)$ have finite $2$-moments, because 
\begin{equation*}
\begin{aligned}
     & \sup_t \sup_{\theta \in \Theta} \norm{\nabla_\theta \Exp_\theta[I_t | \overline A_{t}]}_\infty < \infty, \\
     & \sup_t \sup_{\theta \in \Theta} \norm{\nabla_\theta^2 \Exp_\theta[I_t | \overline A_{t}]}_\infty < \infty
\end{aligned}
\end{equation*}
and Definition~\ref{def::Theta} implies
\begin{equation*}
\begin{aligned}
     \sup_t \sup_{\theta \in \Theta} \norm{\nabla_\theta m_t}_\infty < \infty
     \textand
     \sup_t \sup_{\theta \in \Theta} \norm{\nabla_\theta^2 m_t}_\infty < \infty.
\end{aligned}
\end{equation*}}

\V{For the Normal distribution, the derivatives are
\begin{equation*}
\begin{aligned}
    \frac{\partial \ell_t (\theta)}{\partial \nu}  
    & = - \frac{1}{\nu} + \frac{1}{\nu^3}(Y_t - m_t)^2, \\
    \frac{\partial \ell_t (\theta)}{\partial m_t}
    & = \frac{1}{\nu^2}(Y_t - m_t), \\
    \frac{\partial^2 \ell_t (\theta)}{\partial \nu^2}  
    & = \frac{1}{\nu^2} - \frac{3}{\nu^4} (Y_t - m_t)^2, \\
    \frac{\partial^2 \ell_t (\theta)}{\partial \nu \partial m_t} 
    & = -\frac{2}{\nu^3}(Y_t - m_t), \\
    \frac{\partial^2 \ell_t (\theta)}{\partial m_t^2}  
    & = -\frac{1}{\nu^2}.
\end{aligned}
\end{equation*}
By our Assumption~\ref{assmp::bounded} and Definition~\ref{def::Theta}, $\nabla_{(\nu, m_t)} \ell_t(\theta^*_T)$ has a finite $2\gamma$-moment, and $\sup_{\theta \in \Theta} \nabla_{(\nu, m_t)} \ell_t(\theta)$ and $\sup_{\theta \in \Theta} \nabla_{(\nu, m_t)}^2 \ell_t(\theta)$ have finite $2$-moments almost surely.}

\V{For the NB distribution, the derivatives are
\begin{equation*}
\begin{aligned}
    \frac{\partial \ell_t (\theta)}{\partial \nu} 
    & = \sum_{k=1}^{Y_t} \frac{1}{k + \nu -1} + \log \nu + 1 - \log(\nu + m_t) - \frac{\nu + Y_t}{\nu + m_t}, \\
    \frac{\partial \ell_t (\theta)}{\partial m_t} 
    & = \frac{Y_t}{m_t} - \frac{\nu + Y_t}{\nu + m_t}, \\
    \frac{\partial^2 \ell_t (\theta)}{\partial \nu^2} 
    & = - \sum_{k=1}^{Y_t} \frac{1}{(k+\nu-1)^2} + \frac{1}{\nu} - \frac{1}{\nu + m_t} + \frac{Y_t - m_t}{(\nu + m_t)^2}, \\
    \frac{\partial^2 \ell_t (\theta)}{\partial \nu \partial m_t} 
    & = \frac{Y_t - m_t}{(r+m_t)^2}, \\
    \frac{\partial^2 \ell_t (\theta)}{\partial m_t^2}  
    & = \frac{r+Y_t}{(r+m_t)^2} - \frac{Y_t}{m_t^2}.
\end{aligned}
\end{equation*}
Based on our observation that
\begin{equation*}
\begin{aligned}
    {\sum_{k=1}^{Y_t} \frac{1}{k + \nu -1}} & \leq \frac{Y_t}{\nu}, \\
    {\sum_{k=1}^{Y_t} \frac{1}{(k+\nu-1)^2}} & \leq \frac{Y_t}{\nu^2},
\end{aligned}
\end{equation*}
Assumption~\ref{assmp::bounded} and Definition~\ref{def::Theta} implies $\nabla_{(\nu, m_t)} \ell_t(\theta^*_T)$ has a finite $2\gamma$-moment and $\sup_{\theta \in \Theta} \nabla_{(\nu, m_t)} \ell_t(\theta)$ and $\sup_{\theta \in \Theta} \nabla_{(\nu, m_t)}^2 \ell_t(\theta)$ have finite $2$-moments almost surely.}

\V{
In sum, the HAC variance estimator $\hat\Upsilon$ is almost surely consistent to $\Upsilon$ conditional on $\overline{A}_T$:
for any $\delta > 0$,
\begin{equation*}
    \lim_{T \rightarrow \infty} \Pr[\norm{\hat\Upsilon - \Upsilon} > \delta | \overline{A}_T] = 0,
\end{equation*}
almost surely. By Fubini's theorem,
\begin{equation*}
\begin{aligned}
    \lim_{T \rightarrow \infty} \Pr[\norm{\hat\Upsilon - \Upsilon} > \delta]
    & = \lim_{T \rightarrow \infty} \int \Pr[\norm{\hat\Upsilon - \Upsilon} > \delta | \overline{A}_T] ~dP_{\bar{A}} \\
    & = \int \lim_{T \rightarrow \infty} \int \Pr[\norm{\hat\Upsilon - \Upsilon} > \delta | \overline{A}_T] ~dP_{\bar{A}} \\
    & = \int 0 ~dP_{\bar{A}}
    = 0.
\end{aligned}
\end{equation*}
This concludes our proof for the consistency of $\hat\Upsilon$.
}

\subsection{Lemmas and Proofs for Section~\ref{sec::shrinkage}}
\label{app::shrinkage}

In this appendix, we provide supplementary lemmas for the theoretical arguments regarding the asymptotic converage of the robust empirical Bayesian confidence regions proposed in \cref{sec::shrinkage}.

\begin{lemma} \label{thm::uniform_continuity_r}
    The collection $\{r_0(u, \chi), r_1(u, \chi), \dots \}$ of functions is uniformly equicontinuous with respect to both $u$ and $\chi$.
\end{lemma}

\begin{proof}
    Due to the uniform continuity of the Normal distribution, it is easy to show that $r_0$ is uniformly continuous. So for any $\epsilon > 0$, we can find $\delta > 0$ such that $\abs{u - u'} < \delta$ and $\abs{\chi^2 - \chi'^2} < \delta$ imply $\abs{r_0(u, \chi) - r_0(u', \chi')} < \epsilon$. 
    
    For any $d > 0$, based on \cref{eq::r_d-1}, $r_{d-1}(u,\chi) = \Exp_{\chi_{d-1}}[r_0(u, \sqrt{\chi^2 - \chi_{d-1}^2})]$, where $\chi_{d-1}^2$ is a $\chi^2$ random variable with degree of freedom $d-1$. Hence, for the $\delta$ defined above, $\abs{u - u} < \delta$ and $\abs{\chi^2 - \chi'^2} < \delta$ imply
    \begin{equation*} 
    \begin{aligned}
        \abs{r_{d-1}(u, \chi) - r_{d-1}(u', \chi')}
        & = \abs{\Exp[r_0(u, \sqrt{\chi^2 - \chi_{d-1}^2}) | \chi_{d-1}] - \Exp[r_0(u', \sqrt{\chi'^2 - \chi_{d-1}^2}) } \\
        & \leq \Exp_{\chi_{d-1}}[\abs{r_0(u, \sqrt{\chi^2 - \chi_{d-1}^2}) ] - \Exp_{\chi_{d-1}}[r_0(u', \sqrt{\chi'^2 - \chi_{d-1}^2})} ] \\
        & \leq \Exp_{\chi_{d-1}}[\epsilon] = \epsilon
    \end{aligned}
    \end{equation*}
    because $(\sqrt{\chi^2 - \chi_{d-1}^2})^2 - (\sqrt{\chi'^2 - \chi_{d-1}^2})^2 < \delta$ almost surely.
\end{proof}

Suppose that 
\begin{equation*}
    \mathcal{M} \equiv \left\{(m^{(2)}, m^{(4)}): m^{(2)} = \int u dF(u), m^{(4)} = \int u^2 dF(u), F \in \mathcal{F} \right\}.
\end{equation*} 
In light of Hölder's inequality, it is easy to see that
\begin{equation*}
    \mathcal{M} = \{(m^{(2)}, m^{(4)}) \in (0,\infty) \times (0, \infty): m^{(4)} \geq m^{(2)2} \} \cup \{(0,0)\}.
\end{equation*}

\V{\begin{lemma} \label{thm::continuity_rho}
    $\rho_{m^{(2)},m^{(4)}}(\chi)$ is continuous in $[0, \infty)$. Furthermore, $\rho_{m^{(2)}, m^{(4)}}(\chi)$ is continuous with respect to $(m^{(2)}, m^{(4)})$ in $\mathcal{M}$.
\end{lemma}}

\begin{proof}
    For the continuity with respect to $\chi$, 
    \begin{equation*}
        \abs{\rho_{m^{(2)},m^{(4)}}(\chi_1) - \rho_{m^{(2)},m^{(4)}}(\chi_2)}
        \leq \sup_{F \in \mathcal{F}_{m^{(2)},m^{(4)}}} \int \abs{r_{d-1}(u, \chi_1) - r_{d-1}(u, \chi_2)} ~dF(u),
    \end{equation*}
    where $\mathcal{F}_{m^{(2)}, m^{(4)}} \equiv \{F \in \mathcal{F}: \int u dF(u) = m^{(2)}, \int u^2 dF(u) = m^{(4)}\}$. Due to the uniform continuity of $r_{d-1}$, 
    \begin{equation*}
        \sup_{u \in [0, \infty)} \abs{r_{d-1}(u, \chi_1) - r_{d-1}(u, \chi_2)}
        \leq C \abs{\chi_1 - \chi_2}
    \end{equation*}
    for some constant $C > 0$. This shows the continuity with respect to $\chi$ in $[0, \infty)$.
    
    For the continuity with respect to $(m^{(2)}, m^{(4)})$, we first consider the case where $m^{(4)} > m^{(2)2} > 0$. Then, $(m^{(2)}, m^{(4)})$ is in the interior of $\mathcal{M}$. So by the duality theorem in \citet{smith1995generalized}, for $(m^{(2)'}, m^{(4)'})$ in a small enough neighborhood of $(m^{(2)}, m^{(4)})$, 
    \begin{equation*}
    \begin{aligned}
        & \rho_{m^{(2)'}, m^{(4)'}}(\chi) = \inf_{\lambda_0, \lambda_1, \lambda_2} \lambda_0 + \lambda_1 m^{(2)'} + \lambda_2 m^{(4)'} \\
        & \text{s.t.} \quad
        \lambda_0 + \lambda_1 u + \lambda_1 u^2 \geq r_{d-1}(u, \chi)
        \quad \text{for all} \quad u \in [0,\infty),
    \end{aligned}
    \end{equation*}
    and the infimum is finite. As the infimum of the collection of affine functions, $\rho_{m^{(2)}, m^{(4)}}(\chi)$ is a convex function of $(m^{(2)}, m^{(4)})$ in the neighborhood and hence is continuous.
    
    Now, consider the case where $m^{(4)} = m^{(2)2} \geq 0$. Then for $F \in \mathcal{F}_{m^{(2)}, m^{(4)}}$, $u \sim F$ is almost surely $m^{(2)}$. As a result, $\rho_{m^{(2)}, m^{(4)}}(\chi) = r_{d-1}(m^{(2)}, \chi)$ and is continuous in $\partial\mathcal{M} \equiv \{(m^{(2)}, m^{(4)}) \in \mathcal{M}: m^{(4)} = m^{(2)2} \geq 0\}$. It is now sufficient to show that $\rho_{m^{(2)}, m^{(4)}}(\chi)$ is continuous with respect to $m^{(4)}$ at $m^{(4)} = m^{(2)2}$. For any $\epsilon > 0$, suppose that $u \sim F \in \mathcal{F}_{m^{(2)}, m^{(2)2} + \epsilon^2}$. Then,
    \begin{equation*}
        \Exp[(u-m^{(2)})^2]
        = \Exp[u^2] - 2\Exp[u]m^{(2)} + m^{(2)2}
        = (m^{(2)2}+ \epsilon^2) - 2m^{(2)2} + m^{(2)2}
        = \epsilon^2.
    \end{equation*}
    By Markov's inequality, $\Pr[\abs{u - m^{(2)}} > \sqrt{\epsilon}] \leq \epsilon$, and
    \begin{equation*}
    \begin{aligned}
        & \Pr[\chi_{d-1}^2 + (Z - \sqrt{u})^2 \geq \chi^2] \\
        & \leq \Pr[\chi_{d-1}^2 + (Z-\sqrt{u})^2 \geq \chi^2, \abs{u-m^{(2)}} \leq \sqrt{\epsilon}]
        + \Pr[\abs{u-m^{(2)}} > \sqrt{\epsilon}] \\
        & \leq \Pr[\chi_{d-1}^2 + (Z-\sqrt{u})^2 \geq \chi^2, \abs{(Z-\sqrt{u})^2 - (Z - m^{(2)})^2} \leq \epsilon^{1/4} + \epsilon^{1/2}]
        + \epsilon \\
        & \leq \Pr[\chi_{d-1}^2 + (Z-m^{(2)})^2 \geq \chi^2 - \epsilon^{1/4} - \epsilon^{1/2}]
        + \epsilon \\
        & = \rho_{m^{(2)},m^{(2)2}}\left(\sqrt{\chi^2 - \epsilon^{1/4} - \epsilon^{1/2}}\right) + \epsilon,
    \end{aligned}
    \end{equation*}
    where $\chi_{d-1}^2$ and $Z$ are a $\chi^2$ random variable with degree of freedom $d-1$ and a univariate standard Gaussian random variable. 
    Taking the supremum over $F \in \mathcal{F}_{m^{(2)},m^{(4)}}$, we get $\rho_{m^{(2)},m^{(2)2}+\epsilon^2}(\chi) \leq \rho_{m^{(2)},m^{(2)2}}(\sqrt{\chi^2-\epsilon^{1/4}-\epsilon^{1/2}}) + \epsilon$.
    Similarly, we get $\rho_{m^{(2)},m^{(2)2}+\epsilon^2}(\chi) \geq \rho_{m^{(2)},m^{(2)2}}(\sqrt{\chi^2+\epsilon^{1/4}+\epsilon^{1/2}}) - \epsilon$. Because $\rho_{m^{(2)}, m^{(2)2}}(\chi)$ is continuous with respect to $\chi$, we get $\lim_{\epsilon \rightarrow 0+} \rho_{m^{(2)},m^{(2)2}+\epsilon}(\chi) = \rho_{m^{(2)},m^{(2)2}}(\chi)$. This proves the desired claim.
\end{proof}

\V{\begin{lemma} \label{thm::uniform_continuity_rho}
    Let $\mathcal{M}^o$ be a compact subset of $\mathcal{M}$. Then
    $\lim_{\chi \rightarrow \infty} \sup_{\mathcal{M}^o} \rho_{m^{(2)},m^{(4)}}(\chi) = 0$,
    and $\rho_{m^{(2)},m^{(4)}}(\chi)$ is uniformly continuous with respect to $(\chi, m^{(2)}, m^{(4)})$ in $[0,\infty) \times \mathcal{M}^o$.
\end{lemma}}

\begin{proof}
    The first claim follows by Markov's inequality and compactness of $\mathcal{M}^o$. 
    Specifically, because $\Exp_{u \sim F}[\chi_{d-1}^2 + (Z - \sqrt{u})^2] = d + \Exp_{u \sim F}[u]$,
    by Markov's inequality
    \begin{equation*}
    \begin{aligned}
        \rho_{m^{(2)},m^{(4)}}(\chi)
        & = \sup_{F \in \mathcal{F}_{m^{(2)},m^{(4)}}}
        \Pr_{u \sim F}[\chi_{d-1}^2 + (Z - \sqrt{u})^2 \geq \chi^2] \\
        & \leq \sup_{F \in \mathcal{F}_{m^{(2)},m^{(4)}}} \frac{1}{\chi^2} \Exp_{u \sim F}[\chi_{d-1}^2 + (Z - \sqrt{u})^2] \\
        & = \sup_{F \in \mathcal{F}_{m^{(2)},m^{(4)}}} \frac{1}{\chi^2} (d + \Exp_{u \sim F}[u])
        = \frac{d + m^{(2)}}{\chi^2}.
    \end{aligned}
    \end{equation*}
    So the compactness of $\mathcal{M}^o$ implies that $\sup_{\mathcal{M}^o} m^{(2)} < \infty$, and $\lim_{\chi \rightarrow \infty} \sup_{\mathcal{M}^o} \rho_{m^{(2)}, m^{(4)}}(\chi) = 0$. 

    Now given $\epsilon > 0$, let $\bar\chi$ be large enough so that $\chi \geq \bar\chi \implies \rho_{m^{(2)}, m^{(4)}}(\chi) < \epsilon$. By Lemma~\ref{thm::continuity_rho}, $\rho_{m^{(2)}, m^{(4)}}(\chi)$ is continuous on $[0, \bar\chi+1] \times \mathcal{M}^o$. Suppose that $(\chi_k, m^{(2)}_k, m^{(4)}_k) \in [0, \infty) \times \mathcal{M}^o$ for $k=1,2$. If both $\chi_1$ and $\chi_2$ are in $[0, \bar\chi+1]$,
    due to the compactness of $[0, \bar\chi+1] \times \mathcal{M}^o$, there exists $\delta < 1$ such that
    \begin{equation*}
        \max\{\abs{\chi_1 - \chi_2}, \abs{m^{(2)}_1, m^{(2)}_2}, \abs{m^{(4)}_1, m^{(4)}_1}\} \leq \delta \implies \abs{\rho_{m^{(2)}_1, m^{(4)}_1}(\chi_1) - \rho_{m^{(2)}-2, m^{(4)}_2}(\chi_2)} \leq \epsilon.
    \end{equation*}
    On the other hand, if (without loss of generality) $\chi_1 > \bar\chi+1$, 
    \begin{equation*}
    \begin{aligned}
        & \max\{\abs{\chi_1 - \chi_2}, \abs{m^{(2)}_1, m^{(2)}_2}, \abs{m^{(4)}_1, m^{(4)}_1}\} \leq \delta 
        \implies \chi_2 > \bar\chi \\
        & \implies \abs{\rho_{m^{(2)}_1, m^{(4)}_1}(\chi_1) - \rho_{m^{(2)}-2, m^{(4)}_2}(\chi_2)} 
        \leq \abs{\rho_{m^{(2)}_1, m^{(4)}_1}(\chi_1)} + \abs{\rho_{m^{(2)}-2, m^{(4)}_2}(\chi_2)}
        \leq 2\epsilon.
    \end{aligned}
    \end{equation*}
    In sum, $\rho_{m^{(2)},m^{(4)}}(\chi)$ is uniformly continuous with respect to $(\chi, m^{(2)}, m^{(4)})$ in $[0,\infty) \times \mathcal{M}^o$.
\end{proof}


\paragraph{\underline{\bf Proof of \cref{thm::average_coverage}}} 

    First, we show that for any deterministic $\chi_{[1]}, \dots, \chi_{[N]}$, subset $\mathcal{X}$ of $\mathcal{A}$ and $\epsilon > 0$, there exists $\bar{N}_\mathcal{X}$ such that $N_\mathcal{X} \geq \bar{N}_\mathcal{X}$ implies
    \V{\begin{equation} \label{eq::conv_to_r_d-1}
        \frac{1}{N_\mathcal{X}} \left(
        \begin{aligned}
            & \sum_{j \in \mathcal{I}_\mathcal{X}} \Pr[\norm{T_{[j]}^{1/2} (\hat{W}_{[j]} \hat\Upsilon_{[j]}^{1/2})^{-1}(\tilde\theta_{[j]} - \theta_{[j]})}_2 > \chi_{[j]} 
            | \theta_{[1,\dots,N]}, \Upsilon_{[1,\dots,N]} ] \\
            & - \sum_{j \in \mathcal{I}_\mathcal{X}} r_{d-1}(\norm{b_{[j]}}_2^2, \chi_{[j]})
        \end{aligned}
        \right)
        < \epsilon,
    \end{equation}}
    where \V{$\hat{W}_{[j]} \equiv \hat\Phi^{(2)} (\hat\Phi^{(2)} + T_{[j]}^{-1} \hat\Upsilon_{[j]})^{-1}$}, $\mathcal{I}_\mathcal{X} \equiv \{j: \Upsilon_{[j]} \in \mathcal{X}\}$ and $N_\mathcal{X} \equiv \abs{\mathcal{I}_\mathcal{X}}$. For the first term,
    \V{\begin{equation*}
    \begin{aligned}
        T_{[j]}^{1/2} (\hat{W}_{[j]} \hat\Upsilon_{[j]}^{1/2})^{-1}(\tilde\theta_{[j]} - \theta_{[j]})
        & = T_{[j]}^{1/2}(\hat{W}_{[j]} \hat\Upsilon_{[j]}^{1/2})^{-1}(\hat\theta_o + \hat{W}_{[j]}(\hat\theta_{[j]} - \hat\theta_o) - \theta_{[j]}) \\
        & = T_{[j]}^{1/2} \hat\Upsilon_{[j]}^{-1/2} (\hat\theta_{[j]} - \theta_{[j]}) 
        + T_{[j]}^{1/2} (\hat{W}_{[j]} \hat\Upsilon_{[j]}^{1/2})^{-1} (id - \hat{W}_{[j]}) (\hat\theta_o - \theta_{[j]}) \\
        & = T_{[j]}^{1/2} \hat\Upsilon_{[j]}^{-1/2} (\hat\theta_{[j]} - \theta_{[j]}) 
        + \hat{b}_{[j]}, \\
    \end{aligned}
    \end{equation*}}
    where \V{$\hat{b}_{[j]} \equiv - T_{[j]}^{-1/2} \hat\Upsilon_{[j]}^{1/2} \hat\Phi^{(2)-1} (\theta_{[j]} - \hat\theta_o)$}. 
    For any $\delta > 0$, \V{$\frac{1}{N_\mathcal{X}} \sum_{j \in \mathcal{I}_\mathcal{X}} \mathbb{I}\{\norm{T_{[j]}^{1/2} (\hat{W}_{[j]} \hat\Upsilon_{[j]}^{1/2})^{-1}(\tilde\theta_{[j]} - \theta_{[j]})}_2 > \chi_{[j]}\}$} is upper bounded by
    \begin{equation*}
    \begin{aligned}
        & \frac{1}{N_\mathcal{X}} \sum_{j \in \mathcal{I}_\mathcal{X}} \mathbb{I}\{\norm{T_{[j]}^{1/2}(\hat{W}_{[j]} \hat\Upsilon_{[j]}^{1/2})^{-1}(\tilde\theta_{[j]} - \theta_{[j]}) - \hat{b}_{[j]} + b_{[j]}}_2 > \chi_{[j]} - \delta\} \\
        & + \frac{1}{N_\mathcal{X}} \sum_{j \in \mathcal{I}_\mathcal{X}} \mathbb{I}\{\norm{\hat{b}_{[j]} - b_{[j]}}_2 \geq \delta \}.
    \end{aligned}
    \end{equation*}
    Because Assumption~\ref{assmp::asymp_normality} implies that \V{$(T_{[j]}^{1/2} \hat{W}_{[j]} \hat\Upsilon_{[j]}^{1/2})^{-1}(\tilde\theta_{[j]} - \theta_{[j]}) - \hat{b}_{[j]} = T_{[j]}^{1/2} \hat\Upsilon_{[j]}^{-1/2} (\hat\theta_{[j]} - \theta_{[j]}) \overset{d}{\rightarrow} Z_d$} conditional on $\theta_{[1,\dots,N]}$ and $\Upsilon_{[1,\dots,N]}$, where $Z_d$ is the $d$-dimensional standard Normal random variable, there exists $\bar{N}_\mathcal{X}^{(1)}$ such that $N_\mathcal{X} \geq \bar{N}_\mathcal{X}^{(1)}$ implies
    \begin{equation*}
    \begin{aligned}
        & \frac{1}{N_\mathcal{X}} \sum_{j \in \mathcal{I}_\mathcal{X}} \Pr[\norm{T_{[j]}^{1/2}(\hat{W}_{[j]} \hat\Upsilon_{[j]}^{1/2})^{-1}(\tilde\theta_{[j]} - \theta_{[j]}) - \hat{b}_{[j]} + b_{[j]}}_2 > \chi_{[j]} - \delta 
        | \theta_{[1,\dots,N]}, \Upsilon_{[1,\dots,N]} ] \\
        & - \frac{1}{N_\mathcal{X}} \sum_{j \in \mathcal{I}_\mathcal{X}} r_{d-1}(b_{[j]}, \chi_{[j]} - \delta) 
        \leq \epsilon.
    \end{aligned}
    \end{equation*}
    On the other hand, by Assumption~\ref{assmp::var_space} and Markov's inequality, we can take $C > 0$ such that
    \begin{equation*}
        \frac{1}{N_\mathcal{X}} \sum_{j \in \mathcal{I}_\mathcal{X}} \mathbb{I}\{\norm{\theta_{[j]}}_2 > C\} \leq \epsilon.
    \end{equation*}
    Then,
    \begin{equation*}
    \begin{aligned}
        & \frac{1}{N_\mathcal{X}} \sum_{j \in \mathcal{I}_\mathcal{X}} \Pr[\norm{\hat{b}_{[j]} - b_{[j]}}_2 \geq \delta | \theta_{[1,\dots,N]}, \Upsilon_{[1,\dots,N]} ] \\
        & \leq \frac{1}{N_\mathcal{X}} \sum_{j \in \mathcal{I}_\mathcal{X}} \Pr[\norm{\hat{b}_{[j]} - b_{[j]}}_2 \geq \delta | \theta_{[1,\dots,N]}, \Upsilon_{[1,\dots,N]} ] \mathbb{I}\{\norm{\theta_{[j]}}_2 \leq C\}
        + \frac{1}{N_\mathcal{X}} \sum_{j \in \mathcal{I}_\mathcal{X}} \mathbb{I}\{\norm{\theta_{[j]}}_2 > C\} \\
        & \leq \frac{1}{N_\mathcal{X}} \sum_{j \in \mathcal{I}_\mathcal{X}} \Pr[\norm{\hat{b}_{[j]} - b_{[j]}}_2 \geq \delta | \theta_{[1,\dots,N]}, \Upsilon_{[1,\dots,N]} ] \mathbb{I}\{\norm{\theta_{[j]}}_2 \leq C\}
        + \epsilon.
    \end{aligned}
    \end{equation*}
    The eigenvalue conditions in Assumption~\ref{assmp::var_space} imply that $\norm{\hat{b}_{[j]} - b_{[j]}}_2 \leq T_{[j]}^{-1/2} \{ C \lambda_{\max,*} ( \norm{\hat{\Upsilon}_{[j]}^{1/2} - \Upsilon_{[j]}^{1/2}}_\mathrm{op} + \norm{\hat\Phi^{(2)-1} - \Phi^{(2)-1}}_\mathrm{op} ) + \lambda_{\max,*}^2 \norm{\hat\theta_o - \theta_o}_2\}$ for $j$ satisfying $\norm{\theta_{[j]}}_2 \leq C$, and Assumption~\ref{assmp::moment_consistency} implies that for some $\bar{N}_\mathcal{X}^{(2)}$, 
    \begin{equation*}
        N_\mathcal{X} \geq \bar{N}_\mathcal{X}^{(2)} 
        \implies \frac{1}{N_\mathcal{X}} \sum_{j \in \mathcal{I}_\mathcal{X}} \Pr[\norm{\hat{b}_{[j]} - b_{[j]}}_2 \geq \delta | \theta_{[1,\dots,N]}, \Upsilon_{[1,\dots,N]}]
        \leq \epsilon.
    \end{equation*}
    In summary, for $N_\mathcal{X} \geq \max\{\bar{N}_\mathcal{X}^{(1)}, \bar{N}_\mathcal{X}^{(2)}\}$,
    \begin{equation*}
    \begin{aligned}
        & \frac{1}{N_\mathcal{X}} \sum_{j \in \mathcal{I}_\mathcal{X}} \Pr[ \norm{T_{[j]}^{1/2} (\hat{W}_{[j]} \hat\Upsilon_{[j]}^{1/2})^{-1}(\tilde\theta_{[j]} - \theta_{[j]})}_2 > \chi_{[j]} | \theta_{[1,\dots,N]}, \Upsilon_{[1,\dots,N]}] \\
        & \leq \frac{1}{N_\mathcal{X}} \sum_{j \in \mathcal{I}_\mathcal{X}} r_{d-1}(\norm{b_{[j]}}_2^2, \chi_{[j]} - \delta) + 3\epsilon.
    \end{aligned}
    \end{equation*}
    Similarly, there exists $\bar{N}_\mathcal{X}$ such that $N_\mathcal{X} \geq \bar{N}_\mathcal{X}$ implies
    \begin{equation*}
    \begin{aligned}
        & \frac{1}{N_\mathcal{X}} \sum_{j \in \mathcal{I}_\mathcal{X}} \Pr[ \norm{T_{[j]}^{1/2} (\hat{W}_{[j]} \hat\Upsilon_{[j]}^{1/2})^{-1}(\tilde\theta_{[j]} - \theta_{[j]})}_2 > \chi_{[j]} | \theta_{[1,\dots,N]}, \Upsilon_{[1,\dots,N]}] \\
        & \geq \frac{1}{N_\mathcal{X}} \sum_{j \in \mathcal{I}_\mathcal{X}} r_{d-1}(\norm{b_{[j]}}_2^2, \chi_{[j]} + \delta) - 3\epsilon.
    \end{aligned}
    \end{equation*}
    Noting that $r_0, r_1, \dots$ are uniformly equicontinuous, we proved the desired claim.

    We procceed to the proof of the main theorem. 
    Due to the way $\mathcal{U}$ is defined as in Assumption~\ref{assmp::var_space}, $\mathcal{U}$ is compact. Consider the functions $m^{(2)}(\Upsilon)$ and $m^{(4)}(\Upsilon)$ defined as in \cref{eq::m_from_Upsilon} with $\Upsilon$ in place of $\Upsilon_{[j]}$. Due to the continuity of those functions and the compactness of $\mathcal{U}$, $\mathcal{M}^\mathcal{U} \equiv \{(m^{(2)}(\Upsilon), m^{(4)}(\Upsilon)): \Upsilon \in \mathcal{U}\}$ is also compact. \V{In the following proof, we use the uniform continuity of $\rho_{m^{(2)}, m^{(4)}}(\chi)$ with respect to $(\chi, m^{(2)}, m^{(4)})$ in $[0,\infty) \times \mathcal{M}^\mathcal{U}$, as stated in Lemma~\ref{thm::uniform_continuity_rho}. For the full proof of the statement, see \cref{sec::shrinkage}.} 
    For $\delta > 0$, let
    \begin{equation*}
        \bar{\rho}_{m^{(2)}, m^{(4)}}^\delta(\chi)
        \equiv \sup\{\rho_{\tilde{m}^{(2)},\tilde{m}^{(4)}}(\chi): \abs{\tilde{m}^{(2)} - m^{(2)}} < \delta, \abs{\tilde{m}^{(4)} - m^{(4)}} < \delta\},
    \end{equation*}
    \begin{equation*}
        \underline{\rho}_{m^{(2)}, m^{(4)}}^\delta(\chi)
        \equiv \inf\{\rho_{\tilde{m}^{(2)},\tilde{m}^{(4)}}(\chi): \abs{\tilde{m}^{(2)} - m^{(2)}} < \delta, \abs{\tilde{m}^{(4)} - m^{(4)}} < \delta\}.
    \end{equation*}
    If $\delta$ is smaller than a constant depending on $\mathcal{M}^\mathcal{U}$, then $\bar{\rho}_{m^{(2)}, m^{(4)}}^\delta(\chi)$ and $\underline{\rho}_{m^{(2)}, m^{(4)}}^\delta(\chi)$ are continuous in $\chi$, and
    \begin{equation} \label{eq::rho_overline_rho_underline}
        \lim_{\epsilon \rightarrow 0} \sup\{ 
        \bar{\rho}_{m^{(2)}, m^{(4)}}^\delta(\chi) - \underline{\rho}_{m^{(2)}, m^{(4)}}^\delta(\chi):
        \chi \in [0, \infty), (m^{(2)}, m^{(4)}) \in \mathcal{M}^\mathcal{U}\}
        = 0.
    \end{equation}
    Moreover, there exists a finite number of $\norm{\cdot}_\text{op}$-balls $\mathcal{X}_1, \dots, \mathcal{X}_K$ covering $\mathcal{U}$ such that $\Upsilon_{[j]} \in \mathcal{X}_k \implies \abs{m_{[j]}^{(2)} - m^{(2)}(\mathcal{X}_k)} < \delta$ and $\abs{m_{[j]}^{(4)} - m^{(4)}(\mathcal{X}_k)} < \delta$, where $m^{(2)}(\mathcal{X}_k)$ and $m^{(4)}(\mathcal{X}_k)$ are evaluated at the center of $\mathcal{X}_k$. Consider an arbitrary partition $\mathcal{I}_1, \dots, \mathcal{I}_K$ of $\mathcal{I}_\mathcal{X}$ such that $\mathcal{I}_k \subseteq \mathcal{I}_{\mathcal{X} \cap \mathcal{X}_k}$ for all $k$. 
    By the upper claim of \cref{eq::conv_to_r_d-1}, for any $\epsilon > 0$, there exists $\bar{N}_k$ such that $N_k \geq \bar{N}_k$ implies
    \begin{equation} \label{eq::conv_to_r_d-1_k}
        \frac{1}{N_k} \abs*{
        \begin{aligned}
            & \sum_{j \in \mathcal{I}_k} \Pr[\norm{T_{[j]}^{1/2} (\hat{W}_{[j]} \hat\Upsilon_{[j]}^{1/2})^{-1}(\tilde\theta_{[j]} - \theta_{[j]})}_2 > \underline{\chi}_k 
            | \theta_{[1,\dots,N]}, \Upsilon_{[1,\dots,N]} ] \\
            & - \sum_{j \in \mathcal{I}_k} r_{d-1}(\norm{b_{[j]}}_2^2, \underline{\chi}_k)
        \end{aligned}
        }
        < \epsilon,
    \end{equation}
    where $N_k \equiv \abs{\mathcal{I}_k}$ and $\underline{\chi}_k = \min\{\chi: \underline{\rho}_{m^{(2)}(\mathcal{X}_k), m^{(4)}(\mathcal{X}_k)}^{2\epsilon}(\chi) \leq \alpha\}$.
    Let $\bar{N}_* \equiv \max\{\bar{N}_k, k=1,\dots,K\}$, $\mathcal{K} \equiv \{k: N_k \geq \bar{N}_*\}$ and $N_\mathcal{K} \equiv \sum_{k \in \mathcal{K}} N_k$.
    Then,
    \begin{equation*}
        \frac{1}{N_{\mathcal{X}}} \tsum_{j \in \mathcal{I}_\mathcal{X}} \mathbb{I}\{\theta_{[j]} \notin \tilde{\mathcal{C}}_{[j]} \}
        \leq \max_{k \in \mathcal{K}} \frac{1}{N_k} \tsum_{j \in \mathcal{I}_k} \mathbb{I}\{\theta_{[j]} \notin \tilde{\mathcal{C}}_{[j]} \}
        + \frac{N_\mathcal{X} - N_\mathcal{K}}{N_\mathcal{X}}.
    \end{equation*}
    Because $N_\mathcal{X} - N_\mathcal{K}$ is upper bounded by $\bar{N} K$, $\frac{N_\mathcal{X} - N_\mathcal{K}}{N_\mathcal{X}} \rightarrow 0$ as $N_\mathcal{X} \rightarrow \infty$.
    On the other hand,
    \begin{equation*}
        \max_{k \in \mathcal{K}} \frac{1}{N_k} \tsum_{j \in \mathcal{I}_k} \mathbb{I}\{\theta_{[j]} \notin \tilde{\mathcal{C}}_{[j]}\} 
        \leq \max_{k \in \mathcal{K}} \left[ \begin{aligned}
            & \frac{1}{N_k} \tsum_{j \in \mathcal{I}_k} \mathbb{I}\{\abs{\hat{m}_{[j]}^{(l)} - m^{(l)}(\mathcal{X}_k)} > \epsilon ~\text{for}~ l = 2 ~\text{or}~ 4\} \\
            & + \frac{1}{N_k} \tsum_{j \in \mathcal{I}_k} \mathbb{I}\{\norm{T_{[j]}^{1/2} (\hat{W}_{[j]} \hat\Upsilon_{[j]}^{1/2})^{-1}(\tilde\theta_{[j]} - \theta_{[j]})}_2 > \underline{\chi}_k\}
        \end{aligned} \right],
    \end{equation*}
    where $\underline{\chi}_k = \min\{\chi: \underline{\rho}_{m^{(2)}(\mathcal{X}_k), m^{(4)}(\mathcal{X}_k)}^{2\epsilon}(\chi) \leq \alpha\}$. 
    Due to the consistency of $\hat\Upsilon_{[j]}$, $\hat\Phi^{(2)}$ and $\hat\Phi^{(4)}$ in Assumptions~\ref{assmp::var_consistency} and \ref{assmp::moment_consistency} and the homoskedasticity in Assumption~\ref{assmp::var_space}, the first term converges to $0$ as $N_\mathcal{X} \rightarrow \infty$. 
    For the second term, \cref{eq::conv_to_r_d-1_k} implies that
    \begin{equation*}
        \max_{k \in \mathcal{K}} \frac{1}{N_k} \tsum_{j \in \mathcal{I}_k} \mathbb{I}\{\norm{T_{[j]}^{1/2} (\hat{W}_{[j]} \hat\Upsilon_{[j]}^{1/2})^{-1}(\tilde\theta_{[j]} - \theta_{[j]})}_2 > \underline{\chi}_k\}
        \leq \max_{k \in \mathcal{K}} \frac{1}{N_k} r_{d-1}(\norm{b_{[j]}}_2^2, \underline{\chi}_k) + \epsilon
    \end{equation*}
    Because the empirical distribution $F_k$ of $\{b_{[j]}: j \in \mathcal{I}_k\}$ satisfies $\abs{\Exp_{F_k}[\norm{b_{[j]}}_2^2] - m^{(2)}(\mathcal{X}_k)} < \delta$ and $\abs{\Exp_{F_k}[\norm{b_{[j]}}_2^4] - m^{(4)}(\mathcal{X}_k)} < \delta$,
    \begin{equation*}
    \begin{aligned}
        & \frac{1}{N_k} \sum_{j \in \mathcal{I}_k} r_{d-1}(\norm{b_{[j]}}_2^2, \underline\chi_k)
        = \Exp_{F_k} r_{d-1}(\norm{b_{[j]}}_2^2, \underline{\chi}_k)
        \leq \overline{\rho}_{m^{(2)}(\mathcal{X}_k), m^{(4)}(\mathcal{X}_k)}^{2\delta}(\underline{\chi}_k) \\
        & \leq \underline{\rho}_{m^{(2)}(\mathcal{X}_k), m^{(4)}(\mathcal{X}_k)}^{2\delta}(\underline{\chi}_k)
        + (\overline{\rho}_{m^{(2)}(\mathcal{X}_k), m^{(4)}(\mathcal{X}_k)}^{2\delta}(\underline{\chi}_k) -\underline{\rho}_{m^{(2)}(\mathcal{X}_k), m^{(4)}(\mathcal{X}_k)}^{2\delta}(\underline{\chi}_k)) \\
        & \leq \alpha + \sup\{ 
        \bar{\rho}_{m^{(2)}, m^{(4)}}^\delta(\chi) - \underline{\rho}_{m^{(2)}, m^{(4)}}^\delta(\chi):
        \chi \in [0, \infty), (m^{(2)}, m^{(4)}) \in \mathcal{M}^\mathcal{U}\},
    \end{aligned}
    \end{equation*}
    where the last term converges to $0$ as $\delta \rightarrow 0$.
    Therefore, there exists $\delta > 0$ such that the last term is smaller than $\epsilon$. By the $\delta$, we can find the covering $\mathcal{X}_1, \dots, \mathcal{X}_K$, $\bar{N}_k$ and then finally the desired $\bar{N}_\mathcal{X}$ so that $N_\mathcal{X} \geq \bar{N}_\mathcal{X}$ implies $\frac{N_\mathcal{X} - N_\mathcal{K}}{\mathcal{N}_\mathcal{X}} < \epsilon$ and
    \begin{equation*}
        \Exp\left[\left. \frac{1}{N_{\mathcal{X}}} \tsum_{j \in \mathcal{I}_\mathcal{X}} \mathbb{I}\{\theta_{[j]} \notin \tilde{\mathcal{C}}_{[j]} \} \right| \theta_{[1,\dots,N]}, \Upsilon_{[1,\dots,N]} \right]
        \leq \alpha + 3\epsilon.
    \end{equation*}

\section{Appendix: Supplementary Figures}
\label{app::figs}

\begin{figure}[t!]
    \centering
    \hspace{-0.01\textwidth}
    \subfigure[][Bayesian estimates using the least informative prior with $\sigma=1/2$]{\label{fig::theta_epidemia_prior_0} 
    \includegraphics[width=0.66\textwidth]{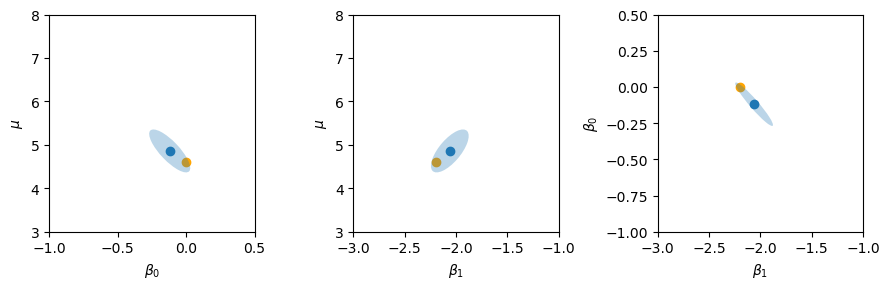}}
    \hspace{-0.01\textwidth}
    \subfigure[][$\sigma=3/8$]{\label{fig::theta_epidemia_prior_1} 
    \includegraphics[width=0.66\textwidth]{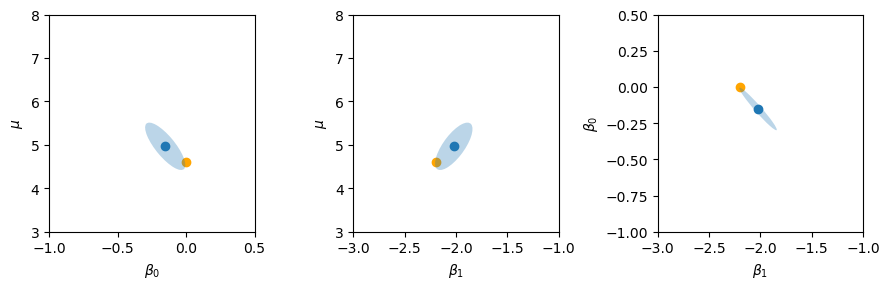}}
    \hspace{-0.01\textwidth}
    \subfigure[][$\sigma=1/4$]{\label{fig::theta_epidemia_prior_2} 
    \includegraphics[width=0.66\textwidth]{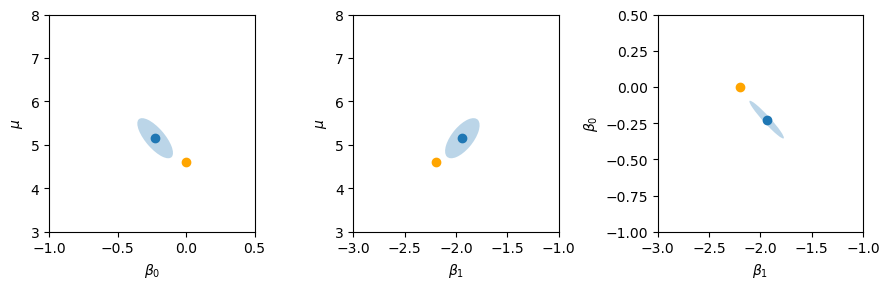}}
    \hspace{-0.01\textwidth}
    \subfigure[][Bayesian estimates using the most informative prior with $\sigma=1/8$]{\label{fig::theta_epidemia_prior_3} 
    \includegraphics[width=0.66\textwidth]{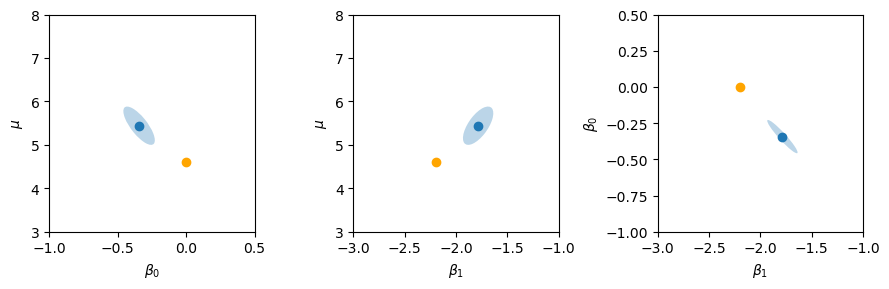}}

    \caption{\V{\textbf{True model parameters (orange), and Bayesian estimates (blue), together with 95\% credible intervals (light blue) using four priors.}
    The infection and death processes were simulated like in \cref{fig::sim_nbinom}.
    The model in \cref{eq::model00,eq::Rt.model0} was fitted using \citet{bhatt2020semi} with priors $N(0, \sigma)$ priors for $\beta_0$ and $\beta_1$, where (a) $\sigma=1/2$ (corresponding to the least informative priors),
    (b) $\sigma=3/8$, (c) $\sigma=1/4$ and (d) $\sigma=1/8$ (corresponding to the most informative priors). }}
\label{fig::theta_epidemia_priors}
\end{figure}

\begin{figure}[t!]
    \centering
    \subfigure[][]{\label{fig::I_theta_0_vs_4} 
    \includegraphics[height=0.22\textwidth]{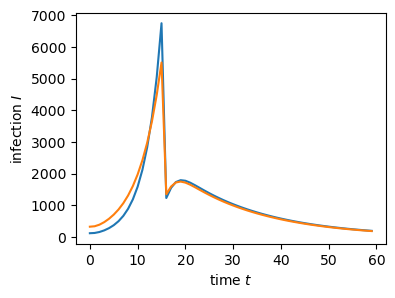}}
    \hspace{-0.01\textwidth}
    \subfigure[][]{\label{fig::EY_theta_0_vs_4} 
    \includegraphics[height=0.22\textwidth]{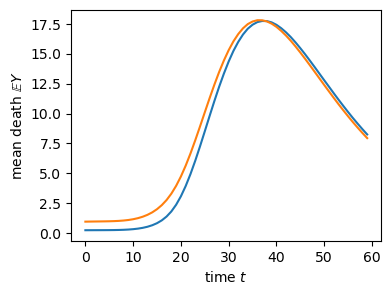}}
    \hspace{-0.01\textwidth}
    \subfigure[][]{\label{fig::Y_theta_0_vs_4} 
    \includegraphics[height=0.22\textwidth]{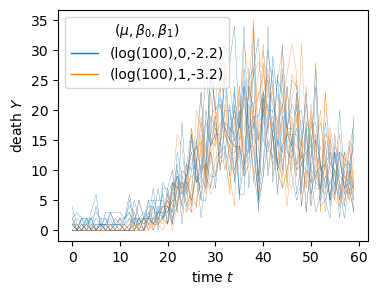}}
    
    \caption{(a) Mean infection and (b) death curves simulated from \cref{eq::model00,eq::Rt.model0} with $(\mu, \beta_0, \beta_1) = (\log(100), 0, -2.2)$ (blue) and $(\log(25), 1, -3.2)$ (orange),
    and (c) ten death processes simulated from each of these models, assuming NB distributions. \V{The observed data look very similar for the two values of $(\mu, \beta_0, \beta_1)$: it would be difficult to pick priors for $(\mu, \beta_0, \beta_1)$ by looking at the data.}}
    \label{fig::theta_0_vs_4}
\end{figure}

\begin{figure}[h!]
     \centering
     \begin{minipage}{0.6\textwidth}
        \subfigure[][]{\label{fig::theta_freqepid_normal} 
        \includegraphics[width=\textwidth]{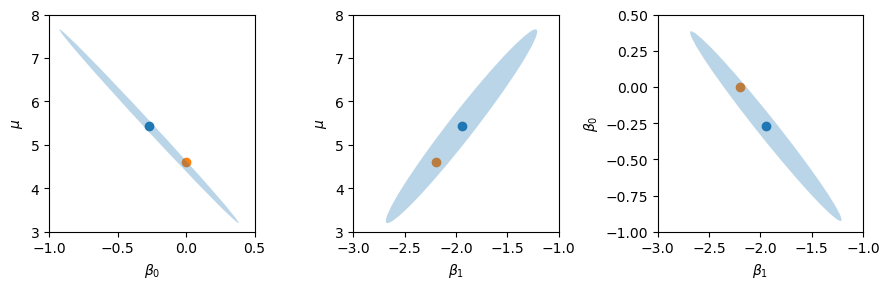}}
    \end{minipage}

    \begin{minipage}{0.3\textwidth}
        \subfigure[][]{\label{fig::pred_I_freqepid_normal} 
        \includegraphics[height=0.75\textwidth]{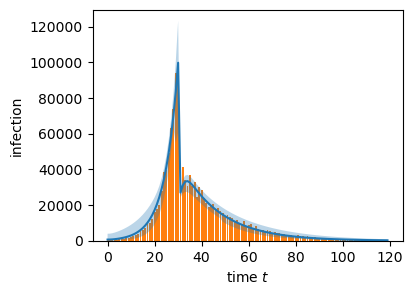}}
    \end{minipage}
    \begin{minipage}{0.3\textwidth}
        \subfigure[][]{\label{fig::pred_EY_freqepid_normal} 
        \includegraphics[height=0.75\textwidth]{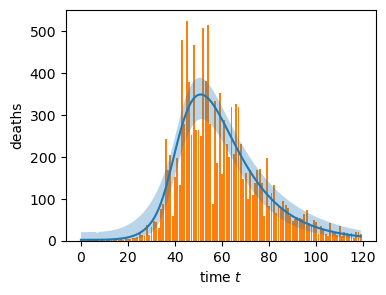}}
    \end{minipage}
    
    \caption{\textbf{ Gaussian model fitted by ML to NB data.} True quantities (model parameters, simulated infections and deaths) are in orange; ML estimates estimates are in blue with 95\% confidence intervals (light blue).
    The infection and death processes in (b,c) were simulated from NB distributions with means 
    in \cref{eq::model0}, $R_t$ in \cref{eq::Rt.model0}, true parameters $\mu= \log(100)$, $\beta_0=0$ and $\beta_1 = -2.2$ in (a), and binary intervention process $A_t=0$ for $t<30$ and $A_t=1$ for $t \ge 30$.
    The model in \cref{eq::model00,eq::Rt.model0} was fitted assuming Gaussian distributed deaths.
    \V{The diagnostics for this fit are in \cref{fig::diagnostic_ny}}.}
     \label{fig::sim_freqepid_normal}
\end{figure}

\begin{figure}
    \centering
    \begin{minipage}{0.3\textwidth}
        \subfigure[][]{\label{fig::diagnostic_nbinom_vs_t} 
        \includegraphics[height=0.75\textwidth]{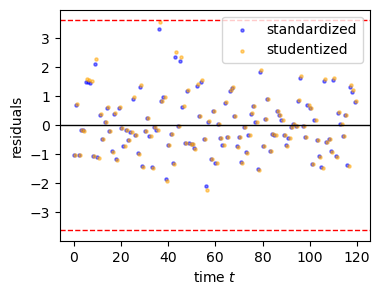}}
    \end{minipage}
    \hspace{-0.01\textwidth}
    \begin{minipage}{0.3\textwidth}
        \subfigure[][]{\label{fig::diagnostic_nbinom_vs_A} 
        \includegraphics[height=0.75\textwidth]{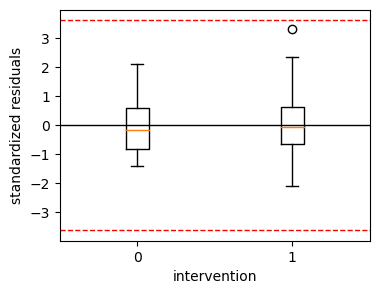}}
    \end{minipage}
    \hspace{-0.01\textwidth}
    \begin{minipage}{0.3\textwidth}
        \subfigure[][]{\label{fig::diagnostic_nbinom_vs_EYmle} 
        \includegraphics[height=0.75\textwidth]{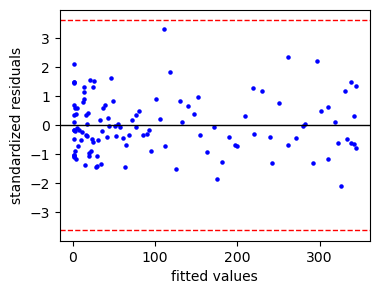}}
    \end{minipage}

    \begin{minipage}{0.3\textwidth}
        \subfigure[][]{\label{fig::diagnostic_nbinom_qq} 
        \includegraphics[height=0.75\textwidth]{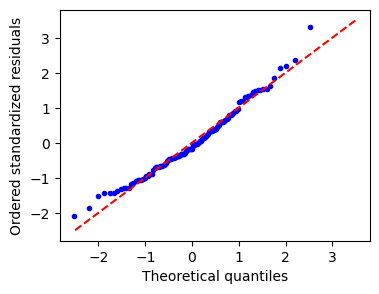}}
    \end{minipage}
    \hspace{-0.01\textwidth}
    \begin{minipage}{0.3\textwidth}
        \subfigure[][]{\label{fig::diagnostic_nbinom_influence} 
        \includegraphics[height=0.75\textwidth]{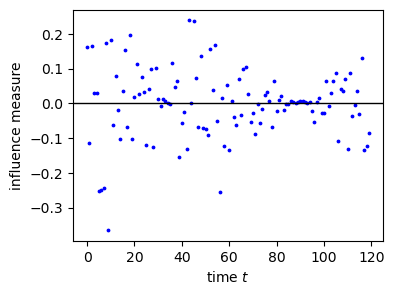}}
    \end{minipage}
    \hspace{-0.01\textwidth}
    \begin{minipage}{0.3\textwidth}
        \subfigure[][]{\label{fig::diagnostic_nbinom_cook} 
        \includegraphics[height=0.75\textwidth]{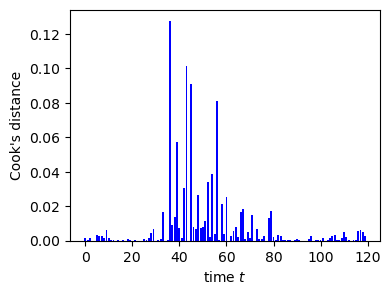}}
    \end{minipage}
    \vspace{-.1in}    
    \caption{\textbf{Diagnostics for the NB model fitted to the NB data shown in \cref{fig::sim_nbinom}}. (a) Standardized and studentized residuals versus $t$. (b) Standardized residuals versus $A_t$ and (c) versus $\hat Y_t$. (d) QQ plot of standardized residuals. (e) Influence on $\hat \beta_1$ and (f) Cook's distance, versus $t$. \V{All diagnostics look fine.}}  
    \label{fig::diagnostic_nbinom}
\end{figure}

\begin{figure}
    \centering
    \includegraphics[width=0.75\textwidth, height=0.35\textheight]{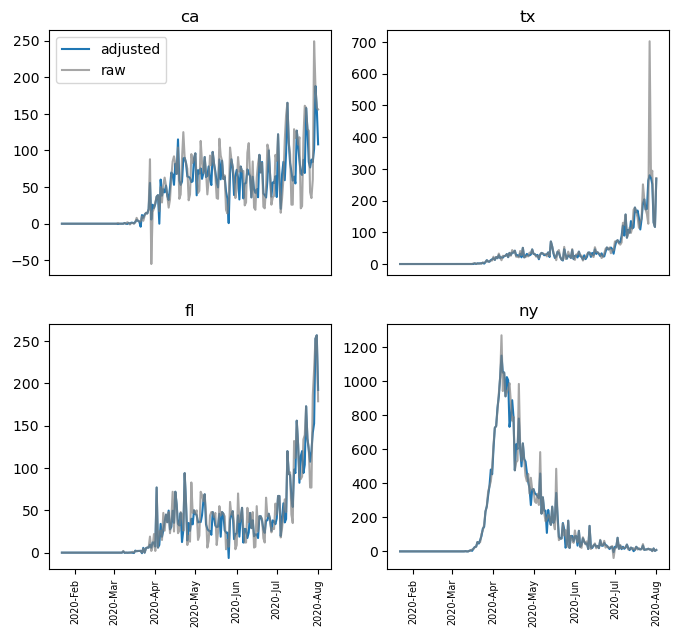}
    
    \caption{\textbf{Observed and weekend corrected death time series for four representative states}. The large grey spikes lying above the blue curves are the monday effects, and those below are the weekend effects.}
    \label{fig::preprocess}
\end{figure}

\begin{figure}[t!]
    \centering
    \begin{minipage}{0.4\textwidth}
        \subfigure[][]{\label{fig::iv_delphi} 
        \includegraphics[height=0.75\textwidth]{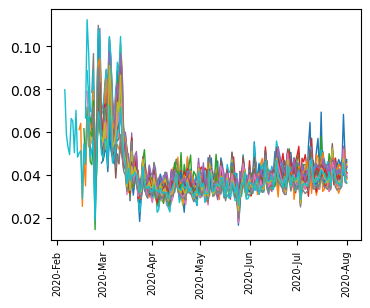}}
    \end{minipage}
    \hspace{-0.01\textwidth}
    \begin{minipage}{0.4\textwidth}
        \subfigure[][]{\label{fig::A_delphi} 
        \includegraphics[height=0.75\textwidth]{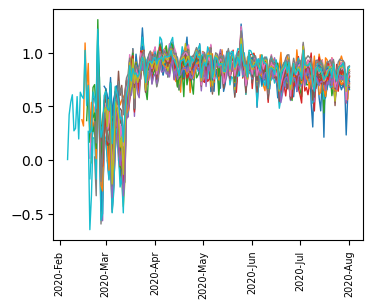}}
    \end{minipage}
    
    \caption{\textbf{Rescaling and shifting mobility $A_t$.} (a) Observed Mobility time series in the 30 states and (b) scaled and shifted versions so the new values are mostly in the $[0,1]$ range. \V{The ML estimates of $\beta_1$ on the original scale are around 
    100. When we rescale $A_t$, the $\beta_1$ estimates have approximately the same magnitude and sign as in \cite{bhatt2020semi}.}}
    \label{fig::iv_A_delphi}
\end{figure}

\begin{figure}
    \centering
    \begin{minipage}{0.3\textwidth}
        \subfigure[][]{\label{fig::pred_R_freqepid_ny_wo_139} 
        \includegraphics[height=0.75\textwidth]{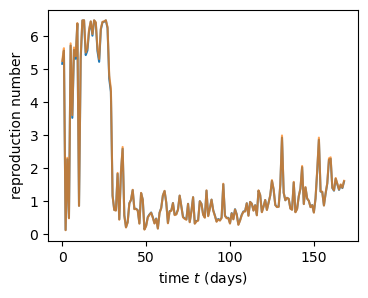}}
    \end{minipage}
    \hspace{-0.01\textwidth}
    \begin{minipage}{0.3\textwidth}
        \subfigure[][]{\label{fig::pred_EY_freqepid_ny_wo_139} 
        \includegraphics[height=0.75\textwidth]{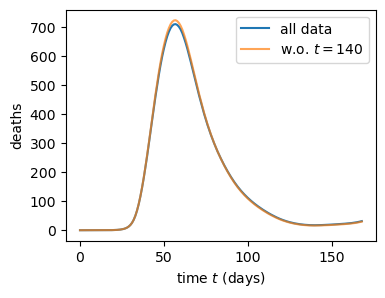}}
    \end{minipage}
    \caption{\textbf{ Influence of $Y_{140}$ for NY deaths}. (a) Fitted reproduction numbers $R_t$ and (b) mean fitted deaths with and without $Y_{140}$.
    The curves are nearly identical. The data and model fit in NY are in \cref{fig::pred_EY_freqepid_delphi} and diagnostics are in \cref{fig::diagnostic_ny}.}
    \label{fig::pred_wo_139}
\end{figure}

\begin{figure}[h!]
    
    \centering
    \subfigure[][]{\label{fig::acf_tmle} \includegraphics[height=0.4\textwidth]{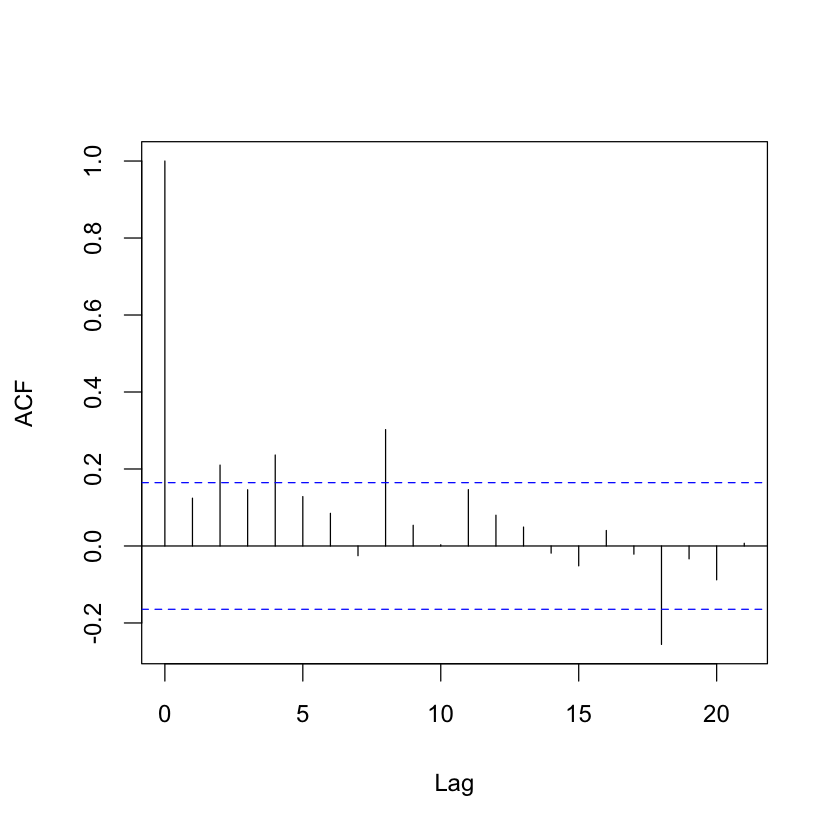}}
    \subfigure[][]{\label{fig::acf_diff} \includegraphics[height=0.4\textwidth]{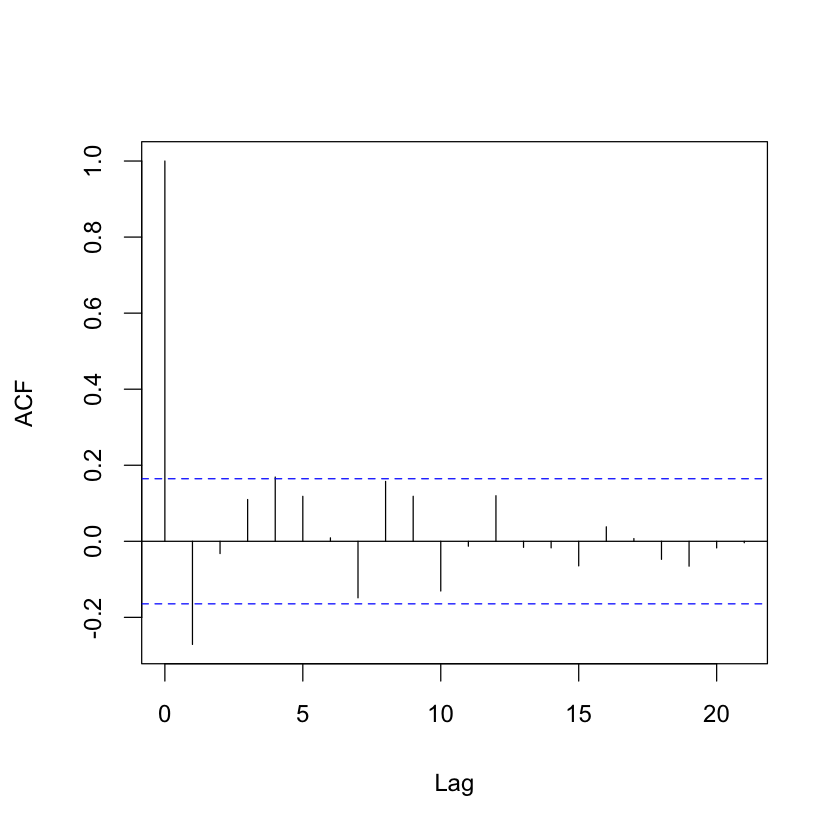}} 
    
    \caption{\V{\textbf{Autocorrelation function (ACF) of NY state residuals in \cref{sec::data}.} (a) ACF using standardized model residuals. The positive correlations for small lags may be due to the model lack of fit seen in \cref{fig::diagnostic_ny_vs_t}. (b) ACF using first differences, which does not rely on a model. There does not appear to be long range correlations in (a) or (b), which is a condition for \cref{thm::asymp_norm}. ACFs for the other states are similar.}}
     \label{fig::ACF}
\end{figure}

\end{document}